%% file: lvr-fee-model.tex
\documentclass[11pt]{article}

\usepackage{epsfig}
\usepackage{latexsym}
\usepackage{enumerate}
\usepackage{url}
\usepackage{float}
\usepackage[hypertexnames=false,hyperfootnotes=false]{hyperref}
\usepackage{texnansi}
\usepackage{color}
\usepackage{tikz}
\usepackage[margin=10pt,font=small,labelfont=bf]{caption}
\usepackage[subrefformat=parens,labelformat=parens]{subcaption}
\usepackage{afterpage}
\usepackage{enumitem}
\usepackage[boxed]{algorithm}
\usepackage{algpseudocode}
\usepackage[normalem]{ulem}
\usepackage{lmodern}
\usepackage{booktabs}
\usepackage{sectsty}
\usepackage{ifthen}
\usepackage{amsmath}
\usepackage{amsthm}
\usepackage{amssymb}
\usepackage{amsfonts}
\usepackage{thmtools}
\usepackage{thm-restate}
\usepackage{mathtools}
\usepackage{xspace}
\usepackage{titling}
\usepackage{natbib}
\usepackage{xfrac}
\usepackage{multirow}
\usepackage{bigdelim}
\usepackage{bm}
\usepackage{cleveref}
\usepackage[textsize=tiny]{todonotes}
\newcommand{\field}[1]{\ensuremath{\mathbb{#1}}}
\newcommand{\R}{\ensuremath{\field{R}}} %
\newcommand{\Rp}{\ensuremath{\R_+}} %
\newcommand{\I}[1]{\ensuremath{\mathbb{I}_{\left\{#1\right\}}}} %
\newcommand{\tends}{\ensuremath{\rightarrow}} %
\newcommand{\E}{\ensuremath{\mathsf{E}}} %
\newcommand{\defeq}{\ensuremath{\triangleq}}
\newcommand{\subjectto}{\text{\rm subject to}} %

\newcommand{\Ascr}{\ensuremath{\mathcal A}}

\newcommand{\Cscr}{\ensuremath{\mathcal C}}

\newcommand{\Fscr}{\ensuremath{\mathcal F}}

\DeclarePairedDelimiter\parens{\lparen}{\rparen}
\DeclarePairedDelimiter\abs{\lvert}{\rvert}

\DeclarePairedDelimiter\bracks{\lbrack}{\rbrack}

\newcommand{\minimize}{\ensuremath{\mathop{\mathrm{minimize}}\limits}}
\newcommand{\maximize}{\ensuremath{\mathop{\mathrm{maximize}}\limits}}
\declaretheoremstyle[headfont=\sffamily\bfseries,bodyfont=\itshape]{thm-sf}
\declaretheorem[style=thm-sf]{theorem}

\declaretheorem[style=thm-sf]{assumption}
\crefname{assumption}{assumption}{assumptions}

\declaretheorem[style=thm-sf]{example}

\declaretheorem[style=thm-sf]{corollary}
\declaretheorem[style=thm-sf]{lemma}

\newcommand{\proofnamest}[1]{{\normalfont\sffamily\bfseries #1}}
\renewcommand{\thmcontinues}[1]{\hyperref[#1]{continued}}
\newcommand{\paraheader}[1]{\noindent{\sffamily\bfseries #1}}
\usetikzlibrary{arrows,patterns,plotmarks,pgfplots.groupplots}
\tikzstyle{every picture} += [>=stealth]
\tikzset{axis/.style={semithick, line join=miter}}
\allsectionsfont{\sffamily}
\makeatletter
\def\@seccntformat#1{\csname the#1\endcsname.\quad}
\makeatother

\floatname{algorithm}{\normalfont\sffamily\bfseries Algorithm}
\floatstyle{ruled}
\newcommand{\emailhref}[1]{\href{mailto:#1}{\tt #1}} %
\provideboolean{fastcompile}
\newcommand{\hidefastcompile}[1]{\ifthenelse{\boolean{fastcompile}}{}{#1}}
\usepackage{pgfplots}
\usepackage{pgfplotstable}
\usetikzlibrary{calc}
\usepackage{mathtools}
\definecolor{orange}{rgb}{0.85,0.33,0.13} %
\definecolor{green}{rgb}{0.13,0.85,0.33}
\definecolor{purple}{rgb}{0.33,0.13,0.85}
\definecolor{lime}{rgb}{0.65,0.85,0.13}
\definecolor{blue}{rgb}{0.13,0.65,0.85}
\pgfplotscreateplotcyclelist{tricolor}{%
  orange,every mark/.append style={fill=orange!80!black},mark=*\\%
  green,every mark/.append style={fill=green!80!black},mark=square*\\%
  purple,every mark/.append style={fill=purple!80!black},mark=otimes*\\%
  black,mark=star\\%
  orange,every mark/.append style={fill=orange!80!black},mark=diamond*\\%
  green,densely dashed,every mark/.append style={solid,fill=green!80!black},mark=*\\%
  purple,densely dashed,every mark/.append style={solid,fill=purple!80!black},mark=square*\\%
  black,densely dashed,every mark/.append style={solid,fill=gray},mark=otimes*\\%
  orange,densely dashed,mark=star,every mark/.append style=solid\\%
  green,densely dashed,every mark/.append style={solid,fill=green!80!black},mark=diamond*\\%
}
\pgfplotsset{colormap={tricolormap}{color=(orange) color=(green) color=(purple)},
  colormap={quadcolormap}{color=(orange) color=(lime) color=(blue) color=(purple)}}
\pgfplotstableset{%
  font=\small,
  every head row/.style={before row=\toprule[1pt], after row=\midrule},
  every last row/.style={after row=\bottomrule[1pt]}}
\pgfplotsset{compat=1.15}

\usepackage{setspace}
\usepackage[margin=1in]{geometry}

\provideboolean{submissionversion}
\setboolean{submissionversion}{true}
\provideboolean{blindversion}
\setboolean{blindversion}{false}
\provideboolean{simplegraphics}
\setboolean{simplegraphics}{true}

\ifthenelse{\boolean{submissionversion}}{
  \renewcommand{\todo}[2][1]{}
  
  \newcommand{\deledit}[1]{}
}{
  
  \newcommand{\deledit}[1]{{\color{orange} \sout{#1}}}
}

\newcommand{\pnl}{\ensuremath{\mathsf{P\&L}}\xspace}
\newcommand{\LVR}{\ensuremath{\mathsf{LVR}}\xspace}
\newcommand{\bLVR}{\ensuremath{\overline{\mathsf{LVR}}\xspace}}

\newcommand{\ARB}{\ensuremath{\mathsf{ARB}}\xspace}
\newcommand{\bARB}{\ensuremath{\overline{\mathsf{ARB}}}\xspace}

\newcommand{\FEE}{\ensuremath{\mathsf{FEE}}\xspace}
\newcommand{\bFEE}{\ensuremath{\overline{\mathsf{FEE}}}\xspace}
\newcommand{\bGAS}{\ensuremath{\overline{\mathsf{GAS}}}\xspace}
\newcommand{\bLP}{\ensuremath{\overline{\mathsf{LP}}}\xspace}
\newcommand{\nbFEE}{\ensuremath{\overline{\mathsf{nFEE}}}\xspace}

\DeclareMathOperator{\bound}{bound}

\newcommand{\gammau}{\ensuremath{\gamma_{+}}\xspace}
\newcommand{\gammal}{\ensuremath{\gamma_{-}}\xspace}
\newcommand{\Ptr}{\ensuremath{\mathsf{P}_{\mathsf{trade}}}\xspace}

\ifthenelse{\boolean{blindversion}}{
\title{\bf\sffamily Automated Market Making and Arbitrage Profits}
\author{}
\date{}
}{
\title{\bf\sffamily Automated Market Making and Arbitrage Profits
      \\ in the Presence of Fees}

\author{
    Jason Milionis \\
    Department of Computer Science \\
    Columbia University \\
    \emailhref{jm@cs.columbia.edu} \\
    \and
    Ciamac C. Moallemi \\
    Graduate School of Business \\
    Columbia University \\
    \emailhref{ciamac@gsb.columbia.edu} \\
    \and
    Tim Roughgarden \\
    Department of Computer Science \\
    Columbia University \\
    a16z Crypto \\
    \emailhref{tim.roughgarden@gmail.com}
}
\date{Initial version: February 6, 2023 \\
    Current version: July 23, 2025
  }
}

\begin{document}
\maketitle
\singlespacing

\begin{abstract}
  We consider the impact of trading fees on the profits of arbitrageurs trading against an
  automated market maker (AMM) or, equivalently, on the adverse selection incurred by liquidity
  providers (LPs) due to arbitrage. We extend the model of \citet{lvr2022} for a general class of
  two asset AMMs to introduce both fees and discrete Poisson block generation times. In our
  setting, we are able to compute the expected instantaneous rate of arbitrage profit in closed
  form. When the fees are low, in the fast block asymptotic regime, the impact of fees takes a
  particularly simple form: fees simply scale down arbitrage profits by the fraction of blocks
  which present profitable trading opportunities to arbitrageurs. This fraction decreases with an
  increasing block rate, hence our model yields an important practical insight: faster blockchains
  will result in reduced LP losses.  Further introducing gas fees (fixed costs) in our model, we
  show that, in the fast block asymptotic regime, lower gas fees lead to smaller losses for LPs.
\end{abstract}

\onehalfspacing

\input{introduction}

\input{table-data}

\input{model}

\input{analysis}

\input{asymptotic-analysis}

\input{gas-fees}

\input{optimal-fees}

\section*{Acknowledgments}
The authors wish to thank Nihar Shah, Rithvik Rao, Alexander Nezlobin, and Dan Robinson for helpful comments.
The first author is supported in part by NSF awards CNS-2212745, CCF-2332922, CCF-2212233, DMS-2134059, and CCF-1763970, by an Onassis Foundation Scholarship, and an A.G. Leventis educational grant.
The third author's research at Columbia University is supported in part by NSF awards CCF-2006737 and CNS-2212745.

\section*{Disclosures}
The second author is an advisor to fintech companies. The third author is Head of Research at a16z crypto, which is an investor in various decentralized finance projects, including Uniswap, as well as in the crypto ecosystem more broadly (for general a16z disclosures, see https://www.a16z.com/disclosures/).
Notwithstanding, the ideas and opinions expressed herein are those of the authors, rather than of any companies or their affiliates.

{\small\singlespacing
  \bibliographystyle{plainnat}
  \bibliography{references}
}

\input{appendix}

\end{document}

%% file: introduction.tex
\section{Introduction}

For automated market makers (AMMs), the primary cost incurred by liquidity providers (LPs) is
adverse selection.  Adverse selection arises from the fact that agents (``arbitrageurs'') with an
informational advantage, in the form of knowledge of current market prices, can exploit stale
prices on the AMM versus prices on other markets such as centralized exchanges. Because trades
between arbitrageurs and the AMM are zero sum, any arbitrage profits will be realized as losses to
the AMM LPs. \citet{lvr2022} quantify these costs through a metric called
loss-versus-rebalancing (\LVR). They establish that \LVR can be simultaneously interpreted as: (1)
arbitrage profits due to stale AMM prices; (2) the loss incurred by LPs relative
to a trading strategy (the ``rebalancing strategy'') that holds the same risky positions as the
pool, but that trades at market prices rather than AMM prices; and (3) the
value of the lost optionality when an LP commits upfront to a particular liquidity demand curve. They
develop formulas for \LVR in closed form, and show theoretically and empirically that, once market
risk is hedged, the profit-and-loss (\pnl) of an LP reduces to trading fee income minus \LVR. In
this way, \LVR isolates the costs of liquidity provision.

Despite its benefits, \LVR suffers from a significant flaw: it is derived under the
simplification that arbitrageurs do not pay trading fees. In practice, however, trading fees pose
a significant friction and limit arbitrage profits.  The main contribution of the present work is
to develop a tractable model for arbitrage profits \emph{in the presence of trading fees}. We are
able to obtain general formulas for arbitrageur profits in this setting. We
establish that arbitrage profits in the presence of fees are roughly equivalent to the arbitrage
profits in the frictionless case (i.e., \LVR), but scaled down to adjust for the fraction of time
where the AMM price differs from the market price significantly enough that arbitrageurs can make
profits even in the presence of fees. That is, the introduction
of fees can be viewed as a rescaling of time.

Our goal is to introduce fees and understand how they impact arbitrageur behavior. As a starting
point, one could directly introduce fees into the model of \citet{lvr2022},
where prices follow a geometric Brownian motion and arbitrageurs continuously monitor
the AMM. However, this approach suffers a major pathology: when arbitrageurs monitor the
market continuously in the presence of even negligible non-zero fees, the arbitrage profits are
zero!
Intuitively, when there are no fees, every instantaneous price movement provides a profitable
arbitrage opportunity.  With fees, this is true only for movements outside a (fee-dependent)
``no-trade region'' around the AMM price which, with continuous monitoring, then results in an
immediate repositioning of that region.  One can show that the fraction of time for which this
happens is zero, with the market price inside the no-trade region at all other times. This is
analogous to the fact that, in continuous time, a reflected random walk spends almost none of its
time at the boundaries.  In reality, however, arbitrageurs cannot continuously monitor and trade
against the AMM.  For example, for an AMM implemented on a blockchain, the arbitrageurs can only
act at the discrete times at which blocks are generated. Thus, in order to understand arbitrage
profits in the presence of fees, it is critical to model the discreteness of block generation.

\subsection{Model}

Our starting point is the model of \citet{lvr2022}, where arbitrageurs continuously monitor an AMM
to trade a risky asset versus the num\'eraire, and the risky asset price follows geometric
Brownian motion parameterized by volatility $\sigma > 0$. However, we assume that the AMM has a
trading fee $\gamma \geq 0$, and that arbitrageurs arrive to trade on the AMM at discrete times
according to the arrivals of a Poisson process with rate $\lambda > 0$. The Poisson process is a
natural choice because of its memoryless nature and standard usage throughout continuous time
finance. It is natural to assume arrival times correspond to block generation times, since the
arbitrageurs can only trade at instances at which a block is generated, so the parameter $\lambda$
should be calibrated so that the mean interarrival time $\Delta t \defeq \lambda^{-1}$ corresponds
to the mean interblock time.

When an arbitrageur arrives, they seek to make a
trade that myopically maximizes their immediate profit. Arbitrageurs trade myopically
because of competition. If they choose to forgo immediate profit but instead wait for
a larger mispricing, they risk losing the profitable trading
opportunity to the next arbitrageur. If the AMM price net of fees is below (respectively, above)
the market price, the arbitrageur will buy (sell) from the pool and sell (buy) at the market. They
will do so until the net marginal price of the AMM equals the market price. We describe
these dynamics in terms of a mispricing process that is the difference between the AMM and
market log-prices. At each arrival time, a myopic arbitrageur will
trade in a way such that the pool mispricing to jumps to the nearest point in band.
The width of the band is determined by the fee $\gamma$. We call this band the
\emph{no-trade region}, since if the arbitrageur arrives and the mispricing is already in the
band, there is no profitable trade possible. At all non-arrival times, the mispricing is a
diffusion, driven by the geometric Brownian motion governing market prices.

\subsection{Results}\label{sec:results}

In our setting, the mispricing process is a Markovian jump-diffusion process. Our first result
(\Cref{lem:stationary}) is
to establish that this process is ergodic, and to identify its steady state distribution in
closed form. Under this distribution, the probability that, at the instance a block is generated, an arbitrageur can make
profitable trade, i.e., the fraction of time that the mispricing process is outside the no-trade
region in steady state, is given by
\[
\Ptr \defeq \frac{1}{1 + \underbrace{\sqrt{2\lambda} \gamma / \sigma}_{\defeq \eta}}.
\]
This can also be interpreted as the long run fraction of blocks that contain an arbitrage trade.
\Ptr has intuitive structure in that it is a function of the composite parameter
$\eta \defeq \gamma / (\sigma \sqrt{\lambda^{-1}/2})$, the fee measured as a multiple of the
typical (one standard deviation) movement of returns over half the average interarrival time. When
$\eta$ is large (e.g., high fee, low volatility, or frequent blocks), the width of the no-fee
region is large relative to typical interarrival price moves, so the mispricing process is less
likely to exit the no-trade region in between
arrivals, and $\Ptr \approx \eta^{-1}$.

Given the steady state distribution of the pool mispricing, we can quantify the arbitrage
profits. Denote by $\ARB_T$ the cumulative arbitrage profits over the
time interval $[0,T]$.  We compute the expected instantaneous rate of arbitrage profit $\bARB \defeq
\lim_{T\tends 0} \E[\ARB_T]/T$, where the expectation is over the steady state distribution of
mispricing.
We derive a semi-closed form expression
(involving an expectation) for \bARB (\Cref{th:arb-rate}). For specific
cases, such as geometric mean or constant product market makers, this expectation can be evaluated
resulting in an explicit closed form (\Cref{cor:cpmm}).

We further consider an asymptotic analysis in the \emph{fast block regime} where
$\lambda \tends \infty$ (\Cref{th:asymptotic_arb}). Equivalently, this is the
limit as the mean interblock time $\Delta t \defeq \lambda^{-1}\tends 0$). In order to explain our
asymptotic results, we begin with the frictionless base case of \citet{lvr2022}, where there is no
fee ($\gamma = 0$) and continuous monitoring ($\lambda=\infty$). \citet{lvr2022} establish that
the expected instantaneous rate of arbitrage profit is
\begin{equation}\label{eq:blvr-def}
    \bLVR \defeq \lim_{T\tends 0} \frac{\E\left[ \LVR_T \right]}{T}
    =
    \frac{\sigma^2 P}{2}
    \times
    y^{*\prime}\left(P \right).
\end{equation}
Here, $P$ is the current market price, while $y^*(P)$ is the quantity of num\'eraire held by the
pool when the market price is $P$, so that $y^{*\prime}(P)$ is the \emph{marginal
    liquidity} of the pool at price $P$, denominated in the num\'eraire. In the presence of fees and discrete monitoring, our rigorous
analysis establishes that as $\lambda\tends \infty$,
\begin{equation}\label{eq:barb-def}
  \bARB
  \defeq
  \lim_{T\tends 0} \frac{\E\left[ \ARB_T \right]}{T}
  =
  \frac{\sigma^2 P}{2}
  \times
  \underbrace{
      \frac{
        y^{*\prime}\left(P e^{-\gamma} \right)
        +
        e^{+\gamma} \cdot y^{*\prime}\left(P e^{+\gamma} \right)
      }{2}
  }_{\text{$=  y^{*\prime}\left(P \right) + O(\gamma)$ for $\gamma$ small}}
  \times
  \underbrace{
    \frac{1}{1 + \sqrt{2 \lambda}\gamma/\sigma}
  }_{=\Ptr}
  + o\left(\sqrt{\lambda^{-1}}\right).
\end{equation}
Equations~\eqref{eq:blvr-def} and \eqref{eq:barb-def} differ in two ways. First, \eqref{eq:blvr-def} involves
the marginal liquidity $y^{*\prime}(P)$ at the current price $P$, while \eqref{eq:barb-def} averages
the marginal liquidity at the endpoints of the no-trade interval of prices
$[Pe^{-\gamma},P e^{+\gamma}]$. This difference is minor if the fee $\gamma$ is small. The second
difference, which is major, is that arbitrage profits in \eqref{eq:barb-def} are scaled down relative
to \eqref{eq:blvr-def} by precisely the factor \Ptr. In other words, if the fee is low, in
the fast block regime we can view the impact of the fee on arbitrage profits
as scaling down \LVR by the fraction of time that an arriving arbitrageur can
profitably trade: $\bARB \approx \bLVR \times \Ptr$.

Focusing on the dependence on problem parameters, when $\gamma > 0$, \eqref{eq:barb-def} implies that
in the fast block regime
arbitrage profits are proportional to the square root of the mean interblock time
($\sqrt{\lambda^{-1}})$, the cube of the volatility ($\sigma^3$),
and the reciprocal of the fee ($\gamma^{-1}$). These scaling dependencies are consistent with the
results of \citet{alexDET}, who consider a similar problem
with a stylized AMM and general block-time distributions.
Equation~\eqref{eq:barb-def} also highlights an interesting phase transition with the introduction
of fees. Specifically, in the absence of fees ($\gamma = 0$), in the fast block regime
($\lambda\tends\infty$), we have the $\bARB = \bLVR + o(1) = \Theta(1)$, i.e., up to a first order,
arbitrage profits per unit time are constant and do not depend on the interblock time. On the other
hand, when there are fees ($\gamma > 0$), we have that $\bARB=\Theta(\sqrt{\lambda^{-1}})$,
arbitrage profits per unit time scale with the square root of the interblock time. \emph{In other
  words, our model yields an important practical insight: that LP losses to arbitrageurs are reduced on faster blockchains.}

Considering the fees paid by arbitrageurs to the pool, define \bFEE to be the instantaneous rate
of arbitrage fees. We establish that (\Cref{th:asymptotic_fee}), asympotitically,
$\bFEE \approx \bLVR \times \left( 1 - \Ptr \right)$ when the fee $\gamma$ is small in the fast
block regime. This implies that, assuming fee $\gamma$ is small and we are in the fast block
regime, $\bARB + \bFEE \approx \bLVR$, which can be interpreted as $\bLVR$ being split among fees
and arbitrage profits, according to $\Ptr$.  In particular, as the blocks become more and more
frequent (for a fixed fee $\gamma$), $\bLVR$ is redirected from arbitrage profits to fees paid by
arbitrageurs.

Finally, we construct a model as above with the addition of gas fees (\Cref{sec:gas}), i.e., fixed
transaction costs associated with executing any potential arbitrage transaction. Our results in
the fast block regime show that \emph{lower gas fees result in smaller losses to LPs}.
We also establish in this model that, when both fixed (gas) and marginal (trading) fees are small, all of these LP losses leak to the validators in gas fees, elucidating that they are the true recipients of the informational losses due to the stale prices of AMMs.

\subsection{Conclusion}

This work has broad implications around liquidity provision and the design of automated market
makers:

\begin{itemize}
\item \paraheader{Blockchain architecture implications:} The asymptotic regime
  analysis $\lambda\to\infty$ above points to a significant potential mitigator of arbitrage
  profits: running a chain with lower mean interblock time $\Delta t \defeq \lambda^{-1}$
  (essentially, a faster chain), since we show that this effectively reduces arbitrage profit
  without negatively impacting LP fee income derived from noise trading. Similarly, reduction of
  gas costs reduces arbitrage profits. We discuss this in \Cref{sec:amm-design}.

\item \paraheader{Pricing accuracy:} Setting a low fee enables accurate prices, since arbitrageurs
  can then correct even small discrepancies, but this comes at the cost of higher arbitrage
  profits. Our model can crisply characterize this tradeoff. We discuss this in \Cref{sec:pricing}.

\item \paraheader{Improved LP performance modeling.} Our model provides a more accurate
  quantification of LP \pnl, accounting both for arbitrageurs paying trading fees and discrete
  arbitrageur arrival times.  Our results thus have the potential to better inform AMM design, and
  in particular, provide guidance around how to set trading fees in an AMM to balance
  LP fee income from noise traders and LP loss due to arbitrageurs. Our results can also be used
  to contruct equilibria for LPs in counterfactual settings. We discuss this in \Cref{sec:lp_decomp}.
\end{itemize}

These findings provide a comprehensive framework for understanding and optimizing automated market maker performance in the presence of realistic market frictions.

\subsection{Literature Review}

There is a rich literature on automated market makers.
\citet{angeris2020improved} and
\citet{angeris2021replicatingmarketmakers,angeris2021replicatingmonotonicpayoffs} apply tools from convex analysis (e.g., the pool reserve value function) that we also use in this paper.
In the first paper to decompose the return of an LP into an instantaneous market risk component and a non-negative, non-decreasing, and predictable component called ``loss-versus-rebalancing'' (\LVR, pronounced ``lever''), \citet{lvr2022} analyze the frictionless, continuous-time Black-Scholes setting in the absence of trading fees to show that it is exactly the adverse selection cost due to the arbitrageurs' informational advantage to the pool.
This work extends the model of \citet{lvr2022} to account for arbitrage profits both in the presence of fees and discrete-time arbitrageur arrivals.
Broader classes of AMMs that have locally smooth demand curves but are not necessarily constant function market makers have been given by \citet{jason_exchange_complexity,jason_revenue_optimal_LP}; our model here applies to such a general case as well.
\citet{evans2021optimalfeesGeoMeanAMMs} observe that, in the special case of geometric mean market makers, taking the limit to continuous time while holding the fees $\gamma>0$ fixed and strictly positive yields vanishing arbitrage profits; this is also a special case of our results.
\citet{angeris2021replicatingmonotonicpayoffs} also analyze arbitrage profits, but do not otherwise express them in closed-form.
Black-Scholes-style options pricing models, like the ones developed in this paper, have been
applied to weighted geometric mean market makers over a finite time horizon by
\citet{evans2020liquidity}, who also observes that constant product pool values are a
super-martingale because of negative convexity. \citet{clark2020replicating} replicates the payoff
of a constant product market over a finite time horizon in terms of a static portfolio of European
put and call options. \citet{tassy2020growth} compute the growth rate of a constant product market maker with fees.
\citet{dewey2023} develop a model of pricing and hedging AMMs with arbitrageurs and noise traders
and conjecture that arbitrageurs induce the same stationary distribution of mispricing that we
rigorously develop here.

Since it first appeared, our model has been influential in the broader discussion of Maximal
Extracted Value (MEV) \citep{juels2019flash}. Arbitrage profit in our setting models real world CEX-DEX arbitrage
profits, which are thought to be the dominant form of MEV. Reducing this MEV has been an important
goal for practitioners, and our work has been cited by practitioners as a motivation to seek
smaller block times in order to reduce MEV.  Recent empirical work by \citet{fritsch2024measuring}
(discussed in \Cref{sec:amm-design}) provides strong validation of our theoretical predictions
regarding the relationship between block times and arbitrage profits.

Subsequent to the initial publication of this work, \citet{alexDET} consider a setting similar to ours, and the main innovation of their important work is to propose an alternative methodology which allows for more general block-time distributions.
They study a specific, stylized AMM where the arbitrageur needs to trade $\ell \times \Delta$ in num\'eraire value to move the quoted price by $\Delta$ units of return, where $\ell$ is a constant marginal liquidity parameter. In this setting, they can asymptotically compute the intensity of arbitrage profits $\bARB$, and derive a general decomposition of the form $\bARB \approx \Ptr \times \bLVR$, as we do here. In the case of a Poisson block-generation process, their results recover and validate our original results (using a different technical methodology).\footnote{To see the equivalence, two observations need to be made: first, \citet{alexDET} derive formulas for the \emph{per block} rates of $\bARB$ and $\bLVR$; we derive \emph{per unit of time} (i.e., instantaneous) rates, hence their rates need to have the units changed by dividing with the block time to match ours. Secondly, they use the no-trade interval $[0, \gamma]$ while in our fee framework, it is given by $[-\gamma, +\gamma]$. Hence, their internal spread needs to be doubled to be comparable to ours. After performing these two adjustments, their formulas for Poisson block arrivals match our original formulas.}
However, their results are more general than ours in that they can handle arbitrary block-time distributions. Indeed, they establish the striking result that $\bARB$ is minimized (among all block-time distributions with fixed mean) by the deterministic block arrival process. 
That said, their methodology has some limitations with respect to our methodology. Our results  are more general in that they are applicable to all AMMs that satisfy very mild technical conditions, and are not restricted to a stylized constant marginal liquidity AMM. 
We also derive closed form as well as asymptotic expressions, are able to derive the fees obtained by the AMM, and are able to handle gas fees.

%% file: table-data.tex
\pgfplotstableread{
gamma   arb     stdev   ptr     lvrptr  lvr     pcterror
  0.01   3.117519651   5.892564982   0.997605746   3.117517957   3.125   5.43403E-07
  1   2.520287411   5.966002846   0.806451613   2.52016129   3.125   5.00422E-05
  2   2.111697737   6.152833279   0.675675676   2.111486486   3.125   1.00038E-04
  3   1.817133093   6.419547828   0.581395349   1.816860465   3.125   1.50032E-04
  4   1.594706734   6.745200696   0.510204082   1.594387755   3.125   2.00023E-04
  5   1.420809765   7.115568363   0.454545455   1.420454545   3.125   2.50012E-04
  6   1.281122039   7.520524436   0.409836066   1.280737705   3.125   2.99998E-04
  7   1.166453014   7.952648004   0.373134328   1.166044776   3.125   3.49982E-04
  8   1.070633694   8.406392074   0.342465753   1.070205479   3.125   3.99963E-04
  9   0.989369210   8.877546385   0.316455696   0.988924051   3.125   4.49942E-04
  10   0.919577361   9.362873421   0.294117647   0.919117647   3.125   4.99918E-04
  11   0.858988835   9.859854346   0.274725275   0.858516484   3.125   5.49892E-04
  12   0.805895799   10.3665081   0.257731959   0.805412371   3.125   5.99863E-04
  13   0.758988361   10.88126067   0.242718447   0.758495146   3.125   6.49832E-04
  14   0.717245046   11.4028494   0.229357798   0.716743119   3.125   6.99798E-04
  15   0.679857558   11.93025221   0.217391304   0.679347826   3.125   7.49762E-04
  16   0.646177921   12.4626347   0.20661157   0.645661157   3.125   7.99723E-04
  17   0.615680613   12.99931016   0.196850394   0.61515748   3.125   8.49682E-04
  18   0.587934944   13.53970923   0.187969925   0.587406015   3.125   8.99638E-04
  19   0.562584586   14.08335639   0.179856115   0.56205036   3.125   9.49592E-04
  20   0.539332189   14.62985191   0.172413793   0.538793103   3.125   9.99544E-04
  21   0.517927667   15.17885763   0.165562914   0.517384106   3.125   1.04949E-03
  22   0.498159160   15.73008576   0.159235669   0.497611465   3.125   1.09944E-03
  23   0.479846005   16.28329006   0.153374233   0.479294479   3.125   1.14938E-03
  24   0.462833193   16.8382587   0.147928994   0.462278107   3.125   1.19932E-03
  25   0.446986975   17.39480857   0.142857143   0.446428571   3.125   1.24926E-03
  26   0.432191337   17.95278064   0.138121547   0.431629834   3.125   1.29920E-03
  27   0.418345152   18.5120362   0.13368984   0.417780749   3.125   1.34913E-03
  28   0.405359870   19.07245374   0.129533679   0.404792746   3.125   1.39906E-03
  29   0.393157622   19.63392641   0.125628141   0.39258794   3.125   1.44899E-03
  30   0.381669653   20.19635992   0.12195122   0.381097561   3.125   1.49892E-03
  31   0.370835029   20.75967072   0.118483412   0.370260664   3.125   1.54884E-03
  32   0.360599555   21.32378459   0.115207373   0.360023041   3.125   1.59876E-03
  33   0.350914870   21.88863533   0.112107623   0.350336323   3.125   1.64868E-03
  34   0.341737681   22.45416378   0.109170306   0.341157205   3.125   1.69860E-03
  35   0.333029114   23.02031686   0.106382979   0.332446809   3.125   1.74851E-03
  36   0.324754170   23.58704683   0.10373444   0.324170124   3.125   1.79842E-03
  37   0.316881249   24.15431068   0.101214575   0.316295547   3.125   1.84833E-03
  38   0.309381747   24.7220695   0.098814229   0.308794466   3.125   1.89824E-03
  39   0.302229713   25.29028806   0.096525097   0.301640927   3.125   1.94814E-03
  40   0.295401546   25.85893436   0.094339623   0.294811321   3.125   1.99804E-03
  41   0.288875734   26.42797929   0.092250923   0.288284133   3.125   2.04794E-03
  42   0.282632629   26.99739629   0.090252708   0.282039711   3.125   2.09784E-03
  43   0.276654250   27.56716109   0.088339223   0.276060071   3.125   2.14773E-03
  44   0.270924109   28.13725149   0.08650519   0.27032872   3.125   2.19763E-03
  45   0.265427059   28.7076471   0.084745763   0.264830508   3.125   2.24751E-03
  46   0.260149162   29.27832921   0.083056478   0.259551495   3.125   2.29740E-03
  47   0.255077567   29.84928058   0.081433225   0.254478827   3.125   2.34728E-03
  48   0.250200411   30.42048534   0.079872204   0.249600639   3.125   2.39717E-03
  49   0.245506722   30.99192882   0.078369906   0.244905956   3.125   2.44704E-03
  50   0.240986339   31.56359745   0.076923077   0.240384615   3.125   2.49692E-03
  51   0.236629838   32.13547867   0.075528701   0.23602719   3.125   2.54679E-03
  52   0.232428465   32.70756085   0.074183976   0.231824926   3.125   2.59667E-03
  53   0.228374079   33.27983316   0.072886297   0.227769679   3.125   2.64654E-03
  54   0.224459100   33.85228554   0.071633238   0.223853868   3.125   2.69640E-03
  55   0.220676459   34.42490864   0.070422535   0.220070423   3.125   2.74627E-03
  56   0.217019557   34.99769372   0.069252078   0.216412742   3.125   2.79613E-03
  57   0.213482227   35.57063264   0.068119891   0.212874659   3.125   2.84599E-03
  58   0.210058699   36.14371778   0.067024129   0.209450402   3.125   2.89584E-03
  59   0.206743568   36.71694203   0.065963061   0.206134565   3.125   2.94570E-03
  60   0.203531767   37.29029871   0.064935065   0.202922078   3.125   2.99555E-03
  61   0.200418538   37.86378159   0.063938619   0.199808184   3.125   3.04540E-03
  62   0.197399412   38.43738479   0.062972292   0.196788413   3.125   3.09524E-03
  63   0.194470186   39.01110281   0.062034739   0.193858561   3.125   3.14509E-03
  64   0.191626904   39.58493048   0.061124694   0.19101467   3.125   3.19493E-03
  65   0.188865838   40.15886291   0.060240964   0.188253012   3.125   3.24477E-03
  66   0.186183472   40.73289553   0.059382423   0.185570071   3.125   3.29460E-03
  67   0.183576490   41.307024   0.058548009   0.182962529   3.125   3.34444E-03
  68   0.181041756   41.88124424   0.057736721   0.180427252   3.125   3.39427E-03
  69   0.178576310   42.45555241   0.056947608   0.177961276   3.125   3.44410E-03
  70   0.176177348   43.02994486   0.056179775   0.175561798   3.125   3.49393E-03
  71   0.173842217   43.60441815   0.055432373   0.173226164   3.125   3.54375E-03
  72   0.171568403   44.17896901   0.054704595   0.17095186   3.125   3.59357E-03
  73   0.169353522   44.75359437   0.05399568   0.168736501   3.125   3.64339E-03
  74   0.167195312   45.3282913   0.053304904   0.166577825   3.125   3.69321E-03
  75   0.165091626   45.90305701   0.052631579   0.164473684   3.125   3.74302E-03
  76   0.163040422   46.47788889   0.051975052   0.162422037   3.125   3.79283E-03
  77   0.161039763   47.05278442   0.051334702   0.160420945   3.125   3.84264E-03
  78   0.159087801   47.62774124   0.050709939   0.15846856   3.125   3.89245E-03
  79   0.157182780   48.20275707   0.0501002   0.156563126   3.125   3.94225E-03
  80   0.155323028   48.77782978   0.04950495   0.15470297   3.125   3.99205E-03
  81   0.153506950   49.3529573   0.048923679   0.152886497   3.125   4.04185E-03
  82   0.151733024   49.9281377   0.048355899   0.151112186   3.125   4.09165E-03
  83   0.149999801   50.50336911   0.047801147   0.149378585   3.125   4.14144E-03
  84   0.148305895   51.07864975   0.047258979   0.14768431   3.125   4.19124E-03
  85   0.146649984   51.65397794   0.046728972   0.146028037   3.125   4.24102E-03
  86   0.145030803   52.22935206   0.046210721   0.144408503   3.125   4.29081E-03
  87   0.143447143   52.80477057   0.045703839   0.142824497   3.125   4.34060E-03
  88   0.141897849   53.38023198   0.045207957   0.141274864   3.125   4.39038E-03
  89   0.140381815   53.95573489   0.044722719   0.139758497   3.125   4.44016E-03
  90   0.138897979   54.53127795   0.044247788   0.138274336   3.125   4.48993E-03
  91   0.137445328   55.10685987   0.043782837   0.136821366   3.125   4.53971E-03
  92   0.136022888   55.68247941   0.043327556   0.135398614   3.125   4.58948E-03
  93   0.134629727   56.25813538   0.042881647   0.134005146   3.125   4.63925E-03
  94   0.133264949   56.83382665   0.042444822   0.132640068   3.125   4.68902E-03
  95   0.131927697   57.40955213   0.042016807   0.131302521   3.125   4.73878E-03
  96   0.130617146   57.98531077   0.041597338   0.129991681   3.125   4.78854E-03
  97   0.129332504   58.56110157   0.041186161   0.128706755   3.125   4.83830E-03
  98   0.128073010   59.13692357   0.040783034   0.127446982   3.125   4.88806E-03
  99   0.126837934   59.71277584   0.040387722   0.126211632   3.125   4.93781E-03
  100   0.125626571   60.28865749   0.04   0.125   3.125   4.98756E-03
  110   0.114680109   66.04891709   0.03649635   0.114051095   3.125   5.48495E-03
  120   0.105496862   71.8113492   0.033557047   0.104865772   3.125   5.98208E-03
  130   0.097682571   77.57546027   0.031055901   0.097049689   3.125   6.47896E-03
  140   0.090952366   83.34089561   0.028901734   0.090317919   3.125   6.97560E-03
  150   0.085095290   89.10739379   0.027027027   0.084459459   3.125   7.47199E-03
  160   0.079951787   94.87475789   0.025380711   0.079314721   3.125   7.96813E-03
  170   0.075398944   100.6428368   0.023923445   0.074760766   3.125   8.46402E-03
  180   0.071340545   106.4115126   0.022624434   0.070701357   3.125   8.95966E-03
}\efftable

%% file: model.tex
\section{Model}\label{sec:model}

\noindent \textbf{\textsf{Assets.}} %
Fix a filtered probability space $\big(\Omega,\Fscr,\{\Fscr_t\}_{t\geq 0})$ satisfying the usual
assumptions. Consider two assets: a risky asset $x$ and a num\'eraire asset $y$.  Working over
continuous times $t \in \R_+$, assume that there is observable external market price $P_t$ at each
time $t$. The price $P_t$ evolves exogenously according to the geometric Brownian motion
\[
\frac{dP_t}{P_t} = \mu\, dt + \sigma \, dB_t,\quad\forall\ t \geq 0,
\]
with drift $\mu$, volatility $\sigma > 0$, and where $B_t$ is a Brownian motion.

\medskip
\noindent \textbf{\textsf{AMM Pool.}} %
We assume that the AMM operates as a constant function market maker (CFMM).\footnote{For ease of exposition, we present our results here with the background of a CFMM, but our results generally hold for any locally smooth AMM, in the notion of \citet{lvr2022}, i.e., LPs committing to an integrable and locally continuously differentiable demand curve function $x^*(P)$ as in \Cref{as:smooth-v}; hence almost all AMMs are covered by our \emph{same theorems} under their mild technical conditions, regardless of whether they are CFMMs.}
The state of a CFMM
pool is characterized by the reserves $(x,y)\in\R^2_+$, which describe the current holdings of the
pool in terms of the risky asset and the num\'eraire, respectively. Define the feasible set of
reserves \Cscr\ according to
\[
  \Cscr \defeq \{ (x,y) \in \R^2_+\ :\ f(x,y) = L \},
\]
where $f \colon \Rp^2 \rightarrow \R$ is referred to as the \emph{bonding function} or
\emph{invariant}, and $L\in\R$ is a constant.\footnote{We assume that liquidity providers hold
  their positions fixed over the interval of analysis, i.e., no mints or burns of
  liquidity. Moreover, while the overall liquidity is held constant, the marginal liquidity of
  each asset available to trade varies as a function of the price level.}
In other words, the feasible set is a level set of the bonding function.
The pool is
defined by a smart contract which allows an agent to transition the pool reserves from the current
state $(x_0,y_0)\in\Cscr$ to any other point $(x_1,y_1)\in\Cscr$ in the feasible set, so
long as the agent contributes the difference $(x_1-x_0,y_1-y_0)$ into the pool, see
\Cref{fig:cfmm-transition}.

\begin{figure}
  \begin{subfigure}[t]{.475\textwidth}
    \centering
    \begin{tikzpicture}
      \begin{axis}[
        clip=false,
        no markers,
        axis lines=left,
        xtick=\empty,
        xticklabels={},
        ytick=\empty,
        yticklabels={},
        xlabel=$x$,
        ylabel=$y$,
        every axis y label/.style={at=(current axis.above origin),anchor=south},
        every axis x label/.style={at=(current axis.right of origin),anchor=west},
        small,
        xmin=0, xmax=3,
        ymin=0, ymax=3,
        ];
        \path[draw=black,dashed](axis cs:0.5,2) -- (axis cs:2,2)
        node [midway,below] {\small $x_1-x_0$};
        \path[draw=black,dashed] (axis cs:2,2) -- (axis cs:2,0.5)
        node [midway,right] {\small $y_0-y_1$};

        \addplot[mark=none,line width=1pt,domain=0.333:3] {1/x}
        node [right=2pt,pos=0] {\small $f(x,y) = L$};
        \path[draw=black,fill=purple] (axis cs:0.5,2) circle (1.5pt)
        node [above right] {\small $(x_0,y_0)$};
        \path[draw=black,fill=green] (axis cs:2,0.5) circle (1.5pt)
        node [below] {\small $(x_1,y_1)$};
      \end{axis}
    \end{tikzpicture}
    \caption{Transitions between any two points on the bonding curve $f(x,y)=L$ are permitted, if
      an agent contributes the difference into the pool.\label{fig:cfmm-transition}}
  \end{subfigure}
  \hfill
  \begin{subfigure}[t]{.475\textwidth}
    \centering
    \begin{tikzpicture}
      \begin{axis}[
        clip=false,
        no markers,
        axis lines=left,
        xtick=\empty,
        xticklabels={},
        ytick=\empty,
        yticklabels={},
        xlabel=$x$,
        ylabel=$y$,
        every axis y label/.style={at=(current axis.above origin),anchor=south},
        every axis x label/.style={at=(current axis.right of origin),anchor=west},
        small,
        xmin=0, xmax=3,
        ymin=0, ymax=3,
        ];
        \addplot[mark=none,line width=1pt,domain=0.333:3] {1/x}
        node [right=2pt,pos=0] {\small $f(x,y) = L$};

        \addplot[domain=0.333:1.667,dashed] {2-x}
        node [below left,pos=0.9] {\small $\text{slope}=-P$};

        \path[draw=black,fill=orange] (axis cs:1,1) circle (1.5pt)
        node [above right] {\small $\big(x^*(P),y^*(P)\big)$};
      \end{axis}
    \end{tikzpicture}
    \caption{The pool value optimization problem relates points on the bonding curve
      to supporting hyperplanes defined by prices.\label{fig:cfmm-price}}
  \end{subfigure}
  \caption{Illustration of a CFMM.}
\end{figure}

Define the \emph{pool
  value function} $V\colon \R_+ \rightarrow \R_+$ by the optimization problem
\begin{equation}\label{eq:pool-min}
  \begin{array}{lll}
    V(P) \defeq & \minimize_{(x,y)\in\R^2_+} & P x + y \\
    & \subjectto & f(x,y) = L.
  \end{array}
\end{equation}
The pool value function yields the value of the pool, assuming that the external market price of
the risky asset is given
by $P$, and that arbitrageurs can instantaneously trade against the pool maximizing their
profits (and simultaneously minimizing the value of the pool).
Geometrically, the pool value function implicitly defines a reparameterization of the pool state
from primal coordinates (reserves) to dual coordinates (prices); this is illustrated in
\Cref{fig:cfmm-price}.

Following \citet{lvr2022}, we assume that the pool value function satisfies:
\begin{assumption}\label{as:smooth-v}
  \begin{enumerate}[label=(\roman*)]
  \item\label{pt:opt} An optimal solution $\big(x^*(P),y^*(P)\big)$ to the pool value optimization
    \eqref{eq:pool-min} exists for
    every $P \geq 0$.
  \item\label{pt:smooth-v} The pool value function $V(\cdot)$ is everywhere twice continuously
    differentiable.
  \item\label{pt:bounded} For all $t \geq 0$,%
    \[
      \E\left[
        \int_0^t x^*(P_s)^2 P_s^2 \, ds
      \right] < \infty.
    \]
  \end{enumerate}
\end{assumption}
We refer to  $\big(x^*(P),y^*(P)\big)$ as the \emph{demand curves} of the pool for the risky
asset and  num\'eraire, respectively.
\Cref{as:smooth-v}\ref{pt:opt}--\ref{pt:smooth-v} is a sufficient condition for the
following:
\begin{lemma}\label{le:envelope}
  For all prices $P \geq 0$, the pool value function satisfies:
   \begin{enumerate}[label=(\roman*)]
  \item\label{pt:V} $V(P) \geq 0$.
  \item\label{pt:Vp} $V'(P) = x^*(P) \geq 0$.
  \item\label{pt:Vpp} $V''(P) = x^{*\prime}(P) = - P y^{*\prime}(P) \leq 0$.
  \end{enumerate}
\end{lemma}
The proof of \Cref{le:envelope} follows from standard arguments in convex analysis; see
\citet{lvr2022} for details.

\medskip
\noindent \textbf{\textsf{Fee Structure.}} %
Suppose that $(\Delta x, \Delta y)$ is a feasible trade permitted by the pool invariant, i.e.,
given initial pool reserves $(x,y)$ with $f(x,y)=L$, we have $f(x+\Delta x, y + \Delta y) = L$. We
assume that an additional proportional trading fee is paid to the LPs in the pool. The mechanics
of this trading fee are as follows:
\begin{enumerate}
\item The fee is paid in the input asset, i.e., the asset that is contributed to the pool.
\item The fee is realized as a separate cashflow to the LPs.\footnote{This is consistent, for
    example, with Uniswap~v3 \citep{adams2021uniswap}. On the other hand, Uniswap~v2
    \citep{adams2020uniswap} reinvests fees in the pool reserves. Over short time horizons of
    practical interest, these differences are of second order.}
\item We allow for different fees to be paid when the risky asset is bought from the pool and when
  the risky asset is sold to the pool.
\item We denote the fee in units of log price by $\gamma_+, \gamma_- > 0$. In particular, when the
  agent purchases the risky asset from the pool (i.e., $\Delta x < 0$, $\Delta y > 0$), the total
  fee charged is
  \begin{equation}\label{eq:fee1}
    \parens*{ e^{+\gamma_+} - 1} \abs*{ \Delta y },
  \end{equation}
  while  the total fee charged when the agent sells the risky asset to the pool (i.e., $\Delta x >
  0$, $\Delta y < 0$ is $\parens*{e^{+\gamma_-} - 1} \abs*{ \Delta x }$, or, valued in the
  num\'eraire at price $P$,
  \begin{equation}\label{eq:fee2}
    P \parens*{ e^{+\gamma_-} - 1} \abs*{ \Delta x }.
  \end{equation}
\end{enumerate}

Note that, for notational simplicity, we have chosen to denominate the fee in units of log
price. This is mathematically equivalent to standard proportional fees, as illustrated in the
following example:
\begin{example} In our notation, a 30 basis point proportional fee on either buys or sales (e.g.,
  as in Uniswap v2) would correspond to setting $\gamma_+,\gamma_-$ so that
  \[
    e^{+\gamma_+} - 1 = e^{+\gamma_-} - 1 = 0.003,
  \]
  so that
  \[
    \gamma_+ = \gamma_- = \log( 1 + 0.003) \approx 0.002995509.
  \]
  To a first order, $\gamma_+ = \gamma_- \approx 30 \text{ (basis points)}$.
\end{example}

\section{Arbitrageurs \& Pool Dynamics}

At any time $t \geq 0$, define $\tilde P_t$ to be the price of the risky asset implied by pool
reserves, i.e., the reserves are given by $\big(x^*(\tilde P_t), y^*(\tilde P_t)\big)$. Denote
by
\begin{equation}\label{eq:z-def}
  z_t \defeq \log P_t/\tilde P_t,
\end{equation}
the log mispricing of the pool, so that $\tilde P_t = P_t e^{-z_t}$.

We imagine that arbitrageurs arrive to trade against the pool at discrete times according to a
Poisson process of rate $\lambda > 0$. Here, we imagine that arbitrageurs are continuously
monitoring the market, but can only trade against the pool at discrete times when blocks are
generated in a blockchain. Hence, we will view the arrival process as both equivalently describing
the arrival of arbitrageurs to trade or times of block generation.  For a proof-of-work
blockchain, Poisson block generation is a natural assumption \citep{nakamoto2008bitcoin}. However,
modern proof-of-state blockchains typically generate blocks at deterministic times.
In these cases, we will view the Poisson assumption as an approximation that is necessary for
tractability.\footnote{As discussed in \Cref{sec:results}, our results have the same parameter
  scaling dependencies as the results of \citet{alexDET}, who consider a similar problem for a
  stylized AMM under general block-time distributions. This suggests that, at least from the
  perspective of parameter scaling laws, the particular choice of Poisson block-times is not
  important.}  In any case, the mean interarrival time $\Delta t \defeq \lambda^{-1}$ should be
calibrated to the mean interblock time in a blockchain.

Denote the arbitrageur arrival times (or block generation times) by
$0 < \tau_1 < \tau_2 < \cdots$. When an arbitrageur arrives at time $t=\tau_i$, they can trade
against the pool (paying the relevant trading fees) according to the pool mechanism, and
simultaneously, frictionlessly trade on an external market at the price $P_t$. We assume that the
arbitrageur will trade to myopically maximize their instantaneous trading profit.\footnote{Given
  trading fees, if there was a single, monopolist arbitrageur, this may not be optimal, e.g., it
  may be optimal to wait for a large mispricing before trading. However, we assume that there
  exists a universe of competing arbitrageurs, and that an arbitrageur that forgoes any immediate
  profit will lose it to a competitor. Hence, in our setting, competition forces arbitrageurs to
  trade myopically.} While we presently ignore any blockchain transaction fees such as ``gas'', we will revisit this in \Cref{sec:gas}.

The following lemma (with proof in \Cref{app:myopic}) characterizes the myopic behavior of the arbitrageurs in terms of the demand
curves of the pool and the fee structure:
\begin{lemma}\label{lem:myopic}%
  Suppose that an arbitrageur arrives at time $t=\tau_i$, observing external market price $P_t$,
  and implied pool price $\tilde P_{t^{-}}$ or, equivalently, mispricing $z_{t^{-}}$. Then, one of
  the following three cases applies:
  \begin{enumerate}
  \item\label{en:myopic-i} If $P_t > \tilde P_{t^{-}} e^{+\gammau}$ or, equivalently,
    $z_{t^{-}} > +\gamma_+$, the arbitrageur can profitably buy in the pool and sell on the
    external market. They will do so until the pool price satisfies
    $\tilde P_{t} = P_t e^{-\gammau}$ or, equivalently, $z_t=+\gamma_+$.  The arbitrageur profits
    are then
    \[
      P_t
      \left\{ x^*\left( P_t e^{-z_{t^{-}}} \vphantom{\tilde P} \right)
        - x^*\left(P_{t} e^{-\gammau}  \vphantom{\tilde P} \right)  \right\}
      + e^{+\gammau} \left\{ y^*\left(P_t e^{-z_{t^{-}}}  \vphantom{\tilde P}\right)
        -  y^*\left(P_{t} e^{-\gammau} \vphantom{\tilde P}
        \right)\right\} \geq 0.
    \]
  \item If $P_t < \tilde P_{t^{-}} e^{-\gammal}$ or, equivalently, $z_{t^{-}} < -\gamma_-$,
    the arbitrageur can profitably sell in the pool and buy
    the external market. The will do so until the pool price satisfies
    $\tilde P_{t} = P_t e^{+\gammal}$ or, equivalently, $z_t=-\gamma_-$.
    The arbitrageur profits are then
    \[
      P_t
      e^{+\gammal}      
      \left\{ x^*\left(P_t e^{-z_{t^{-}}}  \vphantom{\tilde P} \right)
        - x^*\left(P_{t} e^{+\gammal}  \vphantom{\tilde P} \right) \right\}
      +  \left\{  y^*\left(P_t e^{-z_{t^{-}}}  \vphantom{\tilde P}  \right)
        - y^*\left(P_{t} e^{+\gammal}  \vphantom{\tilde P} \right) \right\} \geq 0.
    \]
  \item If $\tilde P_{t^{-}} e^{-\gammal}  \leq P_t \leq \tilde P_{t^{-}} e^{+\gammau}$, or,
    equivalently,
    $-\gammal  \leq z_{t^{-}} \leq +\gammau$,
    then the arbitrageur makes
    no trade, and $\tilde P_{t} = \tilde P_{t^{-}}$ or, equivalently, $z_t=z_{t^{-}}$.
  \end{enumerate}
\end{lemma}

Considering the three cases in \Cref{lem:myopic}, it is easy to see that, at an arbitrageur
arrival time $\tau_i$, the mispricing process $z_t$ evolves according to\footnote{Define $\bound\{x,u,\ell\} \defeq \max(\min(x,u),\ell)$.}
\begin{equation}\label{eq:z-dyn1}
  z_{\tau_i} = \bound\big\{ z_{\tau_i^{-}}, -\gammal, +\gammau \big\}.
\end{equation}
On the other hand, applying It\^{o}'s lemma to \eqref{eq:z-def}, we have that, at other times $t
> 0$, the process evolves according to
\begin{equation}\label{eq:z-dyn2}
  dz_t = \left( \mu - \tfrac{1}{2} \sigma^2\right) \, dt
  + \sigma\, dB_t.
\end{equation}
Combining \eqref{eq:z-dyn1}--\eqref{eq:z-dyn2}, for all $t \geq 0$,
\begin{equation}\label{eq:z-dyn3}
  z_t =
  \left( \mu - \tfrac{1}{2} \sigma^2\right) t
  + \sigma B_t
  + \sum_{i:\ \tau_i \leq t} J_i,
  \qquad
  J_i \defeq \bound\big\{ z_{\tau_i^{-}}, -\gammal, +\gammau \big\} - z_{\tau_i^{-}}.
\end{equation}
Therefore, the mispricing process $z_t$ is a Markovian jump-diffusion process.
Possible sample paths of these stochastic processes are shown in
\Cref{fig:mispricing}.

\begin{figure}[H]
\centering
\resizebox{0.9\textwidth}{!}{
\begin{tikzpicture}[line cap=round, line join=round]
\pgfmathsetseed{42}
\coordinate (A_start) at (0,0);
\coordinate (B_start) at (3,1);
\coordinate (C_start) at (8,-1);

\newcommand{\AValues}{0.011250470417904723, 0.34288108325214806, 0.8219150148754603, 0.45175297110185, 0.9224084777122268, 0.9001423981074643, 0.5180575635930742, 0.22927840972959113, 0.4328934587109807, 0.338618861454722, 0.7924572770218864, 0.20246903228216018, 0.6370498810826143, 0.9383798585963812, 0.5276314176425759, 0.7729484729604496, 0.4052251382014528, 0.09924175838470473, 0.0520676351553262, 0.007093186964127973, 0.6024511810844448, 0.9516798997427518, 0.17221200676323956, 0.5345013071926525, 0.5801963400192941, 0.4977728138953813, 0.22429064837083623, 0.43796597447298224, 0.20021289884265037, 0.5497249160743183}
\newcommand{\BValues}{0.10442613924327426, 0.46301244767439886, 0.7193555820241478, 0.345283188773032, 0.471976747849589, 0.27740645937330677, 0.07252435858453421, 0.9585417852743978, 0.6698095150211022, 0.5374401965462466, 0.6651227674483586, 0.6962690854696243, 0.1653229115894831, 0.08253982614487931, 0.2807904873784477, 0.7678190458405367, 0.7758273649725949, 0.3519213728867273, 0.36688853237399743, 0.6166695288976702, 0.12797182361978, 0.6522966099731213, 0.5248508004117186, 0.0941343857689021, 0.01922388651752882, 0.3753594736382886, 0.36640872415938064, 0.21730558496228458, 0.01039873807267544, 0.47351062014964396, 0.8719968364127538, 0.9199632577526732, 0.7633409945904036, 0.5016622930758847, 0.5857774203954361, 0.8432996744084691, 0.3874493699425424, 0.1284579533036212, 0.8405796676716136, 0.2913165528409941, 0.5320772551338702, 0.022986727498014115, 0.5769692394826083, 0.6382527869400435, 0.8285083865127345, 0.9109287653271457, 0.8777617229609602, 0.6052154159997601, 0.19979915444965257, 0.3472092301303883}
\newcommand{\CValues}{0.1156113231043282, 0.8893134856716919, 0.34236602100214086, 0.958243186866229, 0.8324715019632951, 0.6996725553026466, 0.472750947501262, 0.08776762585062337, 0.8923209222333626, 0.9230243274553672, 0.004797794207829442, 0.740053607370147, 0.2989451017975667, 0.31293728500969076, 0.5246715505348947, 0.3415983122621412, 0.2190521270035799, 0.44393655916181585, 0.1187056117238191, 0.6788027341108385}

\draw[brown] (A_start)
\foreach \i in \AValues{
    -- ++(0.1, 1.0-2.0*\i)
} coordinate (A_end); %

\draw[brown] (B_start)
\foreach \i in \BValues{
    -- ++(0.1, 0.33-0.8*\i)
} coordinate (B_end); %

\draw[brown] (C_start)
\foreach \i in \CValues{
    -- ++(0.1, -0.5+1*\i)
} coordinate (C_end); %

\draw (A_end) circle (1.5pt);
\fill (B_start) circle (1.5pt);
\draw (B_end) circle (1.5pt);
\fill (C_start) circle (1.5pt);

\draw[dashed] (0,1) node[left] {$+\gammau$} -- (10,1);
\draw[dashed] (0,-1) node[left] {$-\gammal$} -- (10,-1);

\draw[red, thick] (0.9, -0.1) -- (1.1, 0.1);
\draw[red, thick] (0.9, 0.1) -- (1.1, -0.1);
\draw[green!70!black, thick] (2.9, -0.1) -- (3.1, 0.1);
\draw[green!70!black, thick] (2.9, 0.1) -- (3.1, -0.1);
\draw[red, thick] (4.9, -0.1) -- (5.1, 0.1);
\draw[red, thick] (4.9, 0.1) -- (5.1, -0.1);
\draw[green!70!black, thick] (7.9, -0.1) -- (8.1, 0.1);
\draw[green!70!black, thick] (7.9, 0.1) -- (8.1, -0.1);

\draw[->, thick] (0,0) -- (10,0) node[right] {time $t$}; %
\draw[->, thick] (0,-2.5) -- (0,2.5) node[above] {mispricing $z_t$}; %
\end{tikzpicture}
}

\resizebox{0.9\textwidth}{!}{
\hspace{14pt}
\begin{tikzpicture}[line cap=round, line join=round, x=1cm, y=150cm] %
\pgfmathsetseed{42}
\coordinate (A_start) at (0.0, 1.0);
\coordinate (B_start) at (3,1.0051);
\coordinate (C_start) at (8,0.99744);

\newcommand{\AValues}{0.011250470417904723, 0.34288108325214806, 0.8219150148754603, 0.45175297110185, 0.9224084777122268, 0.9001423981074643, 0.5180575635930742, 0.22927840972959113, 0.4328934587109807, 0.338618861454722, 0.7924572770218864, 0.20246903228216018, 0.6370498810826143, 0.9383798585963812, 0.5276314176425759, 0.7729484729604496, 0.4052251382014528, 0.09924175838470473, 0.0520676351553262, 0.007093186964127973, 0.6024511810844448, 0.9516798997427518, 0.17221200676323956, 0.5345013071926525, 0.5801963400192941, 0.4977728138953813, 0.22429064837083623, 0.43796597447298224, 0.20021289884265037, 0.5497249160743183}
\newcommand{\BValues}{0.10442613924327426, 0.46301244767439886, 0.7193555820241478, 0.345283188773032, 0.471976747849589, 0.27740645937330677, 0.07252435858453421, 0.9585417852743978, 0.6698095150211022, 0.5374401965462466, 0.6651227674483586, 0.6962690854696243, 0.1653229115894831, 0.08253982614487931, 0.2807904873784477, 0.7678190458405367, 0.7758273649725949, 0.3519213728867273, 0.36688853237399743, 0.6166695288976702, 0.12797182361978, 0.6522966099731213, 0.5248508004117186, 0.0941343857689021, 0.01922388651752882, 0.3753594736382886, 0.36640872415938064, 0.21730558496228458, 0.01039873807267544, 0.47351062014964396, 0.8719968364127538, 0.9199632577526732, 0.7633409945904036, 0.5016622930758847, 0.5857774203954361, 0.8432996744084691, 0.3874493699425424, 0.1284579533036212, 0.8405796676716136, 0.2913165528409941, 0.5320772551338702, 0.022986727498014115, 0.5769692394826083, 0.6382527869400435, 0.8285083865127345, 0.9109287653271457, 0.8777617229609602, 0.6052154159997601, 0.19979915444965257, 0.3472092301303883}
\newcommand{\CValues}{0.1156113231043282, 0.8893134856716919, 0.34236602100214086, 0.958243186866229, 0.8324715019632951, 0.6996725553026466, 0.472750947501262, 0.08776762585062337, 0.8923209222333626, 0.9230243274553672, 0.004797794207829442, 0.740053607370147, 0.2989451017975667, 0.31293728500969076, 0.5246715505348947, 0.3415983122621412, 0.2190521270035799, 0.44393655916181585, 0.1187056117238191, 0.6788027341108385}

\pgfmathsetmacro\cumulativeSum{0}
\pgfmathsetmacro\prevresult{0}
\draw[red] (A_start)
\foreach \i in \AValues{
    \pgfextra{
        \xdef\prevresult{\cumulativeSum}
        \pgfmathsetmacro\cumulativeSum{\prevresult+0.003*(1.0-2.0*\i)}
        \xdef\cumulativeSum{\cumulativeSum}
        \pgfmathparse{exp(\cumulativeSum)-exp(\prevresult)}
        \pgfmathsetmacro\myresult{\pgfmathresult}
    }
    -- ++(0.1,\myresult)
} coordinate (A_end);

\pgfmathsetmacro\cumulativeSum{0.005087}
\pgfmathsetmacro\prevresult{0.005087} %
\draw[red] (B_start)
\foreach \i in \BValues{
    \pgfextra{
        \xdef\prevresult{\cumulativeSum}
        \pgfmathsetmacro\cumulativeSum{\prevresult+0.003*(0.33-0.8*\i)}
        \xdef\cumulativeSum{\cumulativeSum}
        \pgfmathparse{exp(\cumulativeSum)-exp(\prevresult)}
        \pgfmathsetmacro\myresult{\pgfmathresult}
    }
    -- ++(0.1,\myresult)
} coordinate (B_end);

\pgfmathsetmacro\cumulativeSum{-0.0025633}
\pgfmathsetmacro\prevresult{-0.0025633} %
\draw[red] (C_start)
\foreach \i in \CValues{
    \pgfextra{
        \xdef\prevresult{\cumulativeSum}
        \pgfmathsetmacro\cumulativeSum{\prevresult+0.003*(-0.5+1*\i)}
        \xdef\cumulativeSum{\cumulativeSum}
        \pgfmathparse{exp(\cumulativeSum)-exp(\prevresult)}
        \pgfmathsetmacro\myresult{\pgfmathresult}
    }
    -- ++(0.1,\myresult)
} coordinate (C_end);

\draw[blue,thick] (0.0, 1.0) -- (3.0, 1.0) coordinate (AA_end);
\draw[blue,thick] (3.0, 1.002089) coordinate (BB_start) -- (8.0, 1.002089) coordinate (BB_end);
\draw[blue,thick] (8.0, 1.000437) coordinate (CC_start) -- (10.0, 1.000437);

\draw[black,dashed] (0.0, 1.003) -- (3.0, 1.003);
\draw[black,dashed] (0.0, 0.997) -- (3.0, 0.997);
\fill[gray!20, fill opacity=0.5] (0.0, 1.003) -- (3.0, 1.003) -- (3.0, 0.997) -- (0.0, 0.997) -- cycle;
\draw[black,dashed] (3.0, 1.0051) -- (8.0, 1.0051);
\draw[black,dashed] (3.0, 0.99909) -- (8.0, 0.99909);
\fill[gray!20, fill opacity=0.5] (3.0, 1.0051) -- (8.0, 1.0051) -- (8.0, 0.99909) -- (3.0, 0.99909) -- cycle;
\draw[black,dashed] (8.0, 1.003443) -- (10.0, 1.003443);
\draw[black,dashed] (8.0, 0.99744) -- (10.0, 0.99744);
\fill[gray!20, fill opacity=0.5] (8.0, 1.003443) -- (10.0, 1.003443) -- (10.0, 0.99744) -- (8.0, 0.99744) -- cycle;

\fill (A_start) circle (1.5pt);
\draw (AA_end) circle (1.5pt);
\fill (BB_start) circle (1.5pt);
\draw (BB_end) circle (1.5pt);
\fill (CC_start) circle (1.5pt);

\draw[red, thick] (0.9, 0.9938) -- (1.1, 0.9952);
\draw[red, thick] (0.9, 0.9952) -- (1.1, 0.9938);
\draw[green!70!black, thick] (2.9, 0.9938) -- (3.1, 0.9952);
\draw[green!70!black, thick] (2.9, 0.9952) -- (3.1, 0.9938);
\draw[red, thick] (4.9, 0.9938) -- (5.1, 0.9952);
\draw[red, thick] (4.9, 0.9952) -- (5.1, 0.9938);
\draw[green!70!black, thick] (7.9, 0.9938) -- (8.1, 0.9952);
\draw[green!70!black, thick] (7.9, 0.9952) -- (8.1, 0.9938);

\draw[dashed] (3,0.9945) -- (3,1.0051);
\draw[dashed] (8,0.9945) -- (8,1.0051);

\draw[->, thick] (0,0.9945) -- (10,0.9945) node[right] {time $t$}; %
\draw[->, thick] (0,0.9945) -- (0,1.005) node[above] {price}; %
\end{tikzpicture}
}

\caption{Top: example sample path of the mispricing process $z_t$. Bottom: in red, example
  external market price process $P_t$; in blue, example implied pool price process
  $\tilde P_{t^{-}}$. The no-trade interval is shown in shaded gray; whenever the external market
  price is within this interval, no trade will happen even if a block is generated. The
  red- and green-colored crosses in the x-axis show the (Poisson-distributed) times of block
  generation; red indicates blocks where arbitrageurs do not trade with the pool because
  the mispricing does not exceed the trading fee, while
  green indicates blocks where the arbitrageurs do trade. At the green instances, the arbitrageurs
trade until the mispricing is equal to the fee and the marginal profit is zero, i.e., the market
price is at the edge of the no-trade interval.}
\label{fig:mispricing}
\end{figure}

%% file: analysis.tex
\section{Exact Analysis}

We will make the following assumption:
\begin{assumption}[Symmetry]\label{as:symmetric}
  \[
    \mu = \tfrac{1}{2} \sigma^2,\qquad \gammau = \gammal \defeq \gamma.
  \]
\end{assumption}
\Cref{as:symmetric} ensures that the mispricing jump-diffusion process, with dynamics given by
\eqref{eq:z-dyn1}--\eqref{eq:z-dyn2}, is driftless and has a stationary distribution that is
symmetric around $z=0$. This assumption will considerably simplify notation and
expressions and is without loss of generality.  All of our conclusions downstream can be derived
without this assumption, at the expense of additional algebra. We discuss this in greater detail
in \Cref{app:general_stationary}, where we also provide a non-symmetric variation of
\Cref{lem:stationary}.

\subsection{Stationary Distribution of the Mispricing Process}

The following lemma characterizes the stationary distribution of the mispricing
process.\footnote{Contemporaneous with the present work, \citet{dewey2023} conjecture this
  stationary distribution.}
We defer the proof of this lemma until \Cref{app:stationary}.
\begin{theorem}[Stationary Distribution of Mispricing]\label{lem:stationary}
  The process $z_t$ is an ergodic process on $\R$, with unique invariant distribution
  $\pi(\cdot)$ given by the density
  \[
    p_\pi(z) =
    \begin{cases}
      \pi_+ \times p^{\exp}_{\eta/\gamma}(z-\gamma) & \text{if $z > +\gamma$}, \\
      \pi_0 \times \frac{1}{2 \gamma} & \text{if $z\in[-\gamma,+\gamma]$}, \\
      \pi_- \times p^{\exp}_{\eta/\gamma}(-\gamma-z) & \text{if $z < -\gamma$},
    \end{cases}
  \]
  for $z \in \R$. Here, we define the composite parameter $\eta\defeq \sqrt{2 \lambda}
  \gamma/\sigma$.
  The probabilities ${\pi_{-},\pi_0,\pi_{+}}$ of the three segments are given by
  \[
    \pi_0 \defeq \pi\big([-\gamma,+\gamma]\big) = \frac{\eta}{1 + \eta},
    \quad
    \pi_+ \defeq \pi\big((+\gamma,+\infty)\big)
    = \pi_- \defeq \pi\big((-\infty,-\gamma)\big) = \tfrac{1}{2} \frac{1}{1 + \eta}.
  \]
  Finally, $p^{\exp}_{\eta/\gamma}(x) \defeq (\eta/\gamma) e^{-(\eta/\gamma) x}$ is the density of an exponential distribution over
  $x\geq 0$ with parameter $\eta/\gamma = \sqrt{2 \lambda} /\sigma$.
\end{theorem}

The stationary distribution is illustrated in \Cref{fig:mispricing_stationary}.

\begin{figure}[H]
    \centering
    \begin{tikzpicture}[scale=3.5]
        \draw[<->] (-2,0) -- (2,0) node[below left] {pool mispricing $z$};
        \draw[->] (0,0) -- (0,1.5) node[right] {$p_\pi(z)$};
        \ifthenelse{\boolean{simplegraphics}}{
          \fill[gray!20,opacity=0.5] (-0.25,0) rectangle (0.25,1);
        }{
          \draw[pattern=north west lines, pattern color=blue] (-0.25,0) rectangle (0.25,1);
        }
        \draw[thick] (-0.25,1) -- (0.25,1);
        \draw[thick, domain=-2:-0.25, smooth] plot (\x, {exp(2*(\x+0.25))});
        \node[anchor=north west] at (0.33,1) {\small $\propto e^{-z/ \sigma \sqrt{\lambda^{-1}/2}}$};
        \draw[thick, domain=0.25:2, smooth] plot (\x, {exp(-2*(\x-0.25))});
        \node[anchor=north east] at (-0.33,1) {\small $\propto e^{+z/ \sigma \sqrt{\lambda^{-1}/2}}$};
        \draw[dashed] (-0.25,1) -- (-0.25,0) node[below left, shift={(0.25,0)}] {\small $-\gamma$};
        \draw[dashed] (0.25,1) -- (0.25,0) node[below right, shift={(-0.35,0)}] {\small
          $+\gamma$};
        \node [below] at (0,0) {\small $0$};
        \node at (0,0.5) [rectangle,align=center,fill=white] {\small no-trade \\ w.p.~$\pi_0$};
        \node at (-0.7,0.2) [rectangle,align=center,yshift=-5pt] {\small sell trade \\ w.p.~$\pi_-$};
        \node at (+0.7,0.2) [rectangle,align=center,yshift=-5pt] {\small buy trade \\ w.p.~$\pi_+$};
    \end{tikzpicture}
    \caption{The density $p_\pi(z)$ of the stationary distribution $\pi(\cdot)$ of mispricing $z$,
      illustrating trade and no-trade regions for an arbitrageur.}
    \label{fig:mispricing_stationary}
\end{figure}

Under this distribution, the probability that an arbitrageur arrives and can make a
profitable trade, i.e., the fraction of time that the mispricing process is outside the no-trade
region in steady state, is given by
\[
\Ptr \defeq \pi_+ + \pi_- = \frac{1}{1 + \sqrt{2\lambda} \gamma / \sigma}.
\]
Equivalently, $\Ptr$ can be interpreted as the long run fraction of blocks that contain an
arbitrage trade.

Note that \Ptr \emph{does not depend} on the bonding function or feasible set defining the CFMM
pool; the only pool property relevant is the fee $\gamma$.
\Ptr has intuitive structure in that
it is a function of the composite parameter $\eta \defeq \gamma / (\sigma \sqrt{\lambda^{-1}/2})$,
the fee measured as a multiple of the typical (one standard deviation) movement of returns over
half the average interarrival time. When $\eta$ is large (e.g., high fee, low volatility, or
frequent blocks), the width of the no-fee region is large relative to typical interarrival price
moves, so the mispricing process is less likely to exit the no-trade region in between arrivals,
and $\Ptr \approx \eta^{-1}$.  Example calculations of $\Ptr$ are shown in \Cref{table:Ptr} for
$\sigma=5\% \text{ (daily)}$ volatility and varying mean interblock times
$\Delta t \defeq \lambda^{-1}$ and fee levels $\gamma$, as well as in \Cref{fig:ptr}.

\begin{table}[]\centering
    \begin{tabular}{c|ccccc}
        \toprule
        $\Delta t$ \textbackslash{} $\gamma$ & 1 bp & 5 bp & 10 bp & 30 bp & 100 bp \\ \midrule
        10 min & 96.7\% & 85.5\% & 74.7\% & 49.6\% & 22.8\% \\
        2 min & 92.9\% & 72.5\% & 56.9\% & 30.5\% & 11.6\% \\
        12 sec & 80.7\% & 45.6\% & 29.5\% & 12.3\% & 4.0\% \\
        2 sec & 63.0\% & 25.4\% & 14.5\% & 5.4\% & 1.7\% \\
      50 msec & 21.2\% & 5.1\% & 2.6\% & 0.9\% & 0.3\% \\
      \bottomrule
    \end{tabular}
    \caption{The probability of trade $\Ptr$, or, equivalently, the fraction of blocks
      containing an arbitrage trade, given asset price volatility
      $\sigma=5\% \text{ (daily)}$, with varying mean interblock times
      $\Delta t \defeq \lambda^{-1}$ and fee levels $\gamma$ (in basis points).}
    \label{table:Ptr}
\end{table}

\begin{figure}
  \begin{subfigure}[t]{.475\textwidth}
    \centering
    \begin{tikzpicture}
      \begin{axis}[
        clip=true,
        no markers,
        xlabel={$\gamma$ (bp)},
        ylabel={$\Ptr$},
        width=3.0in,
        height=3.0in,
        grid=major,
        xmin=0, xmax=100,
        ymin=0, ymax=1.0,
        ];
        \addplot[mark=none,line width=1pt,smooth] table [x=gamma,y=ptr] {\efftable};
        \addplot[mark=none,dashed,line width=0.5pt,domain=5:100] { 0.05 / sqrt(2.0 / 7200) / x }
        node [pos=0,below right=5pt,rectangle,fill=white] {\small $\sigma \sqrt{\lambda^{-1}/2}/ \gamma$};
      \end{axis}
    \end{tikzpicture}
    \caption{The probability of trade \Ptr, or, equivalently, the fraction of blocks containing an
      arbitrage trade, as a function of the fee $\gamma$.\label{fig:ptr}}
  \end{subfigure}
  \hfill
  \begin{subfigure}[t]{.475\textwidth}
    \centering
    \begin{tikzpicture}
      \begin{axis}[
        clip=true,
        no markers,
        xlabel={$\gamma$ (bp)},
        ylabel={$\sigma_z$ (bp)},
        width=3.0in,
        height=3.0in,
        grid=major,
        xmin=0, xmax=100,
        ymin=0, ymax=100,
        ];
        \addplot[mark=none,line width=1pt,smooth] table [x=gamma,y=stdev] {\efftable};
        \addplot[mark=none,dashed,line width=0.5pt,domain=0:100] {5.89255651}
        node [above left,rectangle,fill=white] {\small $\sigma\sqrt{\lambda^{-1}}$};
        \addplot[mark=none,dashed,line width=0.5pt,domain=0:100] {x / sqrt(3)}
        node [pos=0.7,below right,rectangle,fill=white] {\small $\gamma/\sqrt{3}$};
      \end{axis}
    \end{tikzpicture}
    \caption{The standard deviation of mispricing $\sigma_z$, as a function of the fee $\gamma$.\label{fig:stdev}}
  \end{subfigure}

  \caption{Probability of trade and typical mispricing errors as a function of the fee, with
    $\sigma = 5\% \text{ (daily)}$ and mean interblock time
    $\Delta t \defeq \lambda^{-1}= 12 \text{ (seconds)}$.}
\end{figure}

The following immediate corollary quantifies the magnitude of a typical mispricing. This is
illustrated in \Cref{fig:stdev}.
\begin{corollary}[Standard Deviation of Mispricing]\label{lem:stdev}
  Under the invariant distribution $\pi(\cdot)$, the standard deviation of the mispricing is given
  by
  \[
    \sigma_z \defeq \sqrt{\E_\pi[ z^2 ]}
    = \sqrt{(1-\Ptr) \times \tfrac{1}{3} \gamma^2
      + \Ptr \times
      \left\{
        \left( \gamma + \frac{\sigma}{\sqrt{2\lambda}} \right)^2 + \frac{\sigma^2}{2\lambda}
         \right\}
    }.
  \]
\end{corollary}
Note that \Cref{fig:stdev} quantifies the typical mispricing under the invariant distribution
$\pi(\cdot)$, this is the steady-state distribution that would be observed at the instance of block
generation (at the ``top-of-the-block'', i.e., before any arbitrage transaction). In the fast
block regime ($\lambda\tends\infty$), we have that
\[
\sigma_z = \frac{\gamma}{\sqrt{3}} + O(\lambda^{-1/2}).
\]
In this regime, there is a nonvanishing limit to the mispricing that scales with size of the
fee. This is intuitive, as the no-fee band creates a friction that inhibits price corrections.

\subsection{Rate of Arbitrageur Profit}
\label{subsec:def_rates}

Denote by $N_T$ the total number of arbitrageur arrivals in $[0,T]$. Suppose an arbitrageur arrives at time
$\tau_i$, observing external price $P_{\tau_i}$ and mispricing
$z_{\tau_i^{-}}$. From \Cref{lem:myopic}, the arbitrageur
profit is given by
\[
  A(P_{\tau_i},z_{\tau_i^{-}}) \defeq
  A_+(P_{\tau_i},z_{\tau_i^{-}}) +
  A_-(P_{\tau_i},z_{\tau_i^{-}}) \geq 0,
\]
where we define
\[
  A_+(P,z) \defeq
  \left[
    P
    \left\{ x^*\left(P e^{-z} \vphantom{\tilde P} \right)
      - x^*\left(P e^{-\gamma}\right)  \right\}
    + e^{+\gamma}
    \left\{y^*\left(P e^{-z} \vphantom{\tilde P} \right)
      -  y^*\left(P e^{-\gamma} \right) \right\}
  \right] \I{ z > +\gamma } \geq 0,
\]
\[
  A_-(P,z) \defeq
  \left[
    e^{+\gamma} P
    \left\{ x^*\left(P e^{-z} \vphantom{\tilde P} \right)
      - x^*\left(P e^{+\gamma} \vphantom{\tilde P} \right)  \right\}
    + \left\{ y^*\left(P e^{-z} \vphantom{\tilde P} \right)
      -  y^*\left(P e^{+\gamma} \vphantom{\tilde P} \right) \right\}
  \right] \I{ z < -\gamma } \geq 0.
\]
Similarly, the fees paid by the arbitrageur in this scenarios are given by
\[
  F(P_{\tau_i},z_{\tau_i^{-}}) \defeq
  F_+(P_{\tau_i},z_{\tau_i^{-}}) +
  F_-(P_{\tau_i},z_{\tau_i^{-}}) \geq 0,
\]
where we define
\[
  F_+(P,z) \defeq
  -\left( e^{+\gamma} - 1 \right)
  \left[ y^*\left(P e^{-z}  \vphantom{\tilde P} \right)
    -  y^*\left(P e^{-\gamma}  \vphantom{\tilde P} \right) \right]
  \I{ z > +\gamma } \geq 0,
\]
\[
  F_-(P,z) \defeq
  -\left(e^{+\gamma}  - 1  \vphantom{\tilde P} \right) P
  \left[ x^*\left(P e^{-z} \vphantom{\tilde P} \right)
    -  x^*\left(P e^{+\gamma}  \vphantom{\tilde P} \right) \right]
  \I{ z < -\gamma } \geq 0.
\]
We can write the total arbitrage profit and fees paid over $[0,T]$ by summing over all
arbitrageurs arriving in that interval, i.e.,
\[
  \ARB_T \defeq \sum_{i=1}^{N_T}
  A(P_{\tau_i},z_{\tau_i^{-}}),
  \quad
  \FEE_T \defeq \sum_{i=1}^{N_T}
  F(P_{\tau_i},z_{\tau_i^{-}}).
\]
Clearly these are non-negative and monotonically increasing processes. The following theorem
characterizes their instantaneous expected rate of growth or intensity:\footnote{Mathematically,
  \bARB is the intensity of the compensator for the monotonically increasing jump process $\ARB_T$
  at time $T=0$, similarly \bFEE is the intensity of the compensator for $\FEE_T$.}
\begin{theorem}[Rate of Arbitrage Profit and Fees]\label{th:arb-rate}
  Define the intensity, or instantaneous rate of arbitrage profit, by
  \[
    \bARB  \defeq \lim_{T\tends 0} \frac{\E\left[ \ARB_T \right]}{T}.
  \]
  Given initial price $P_0=P$, suppose that $z_{0-}=z$ is distributed according to its stationary
  distribution $\pi(\cdot)$. Then, the instantaneous rate of arbitrage profit is given by
  \begin{equation}\label{eq:arb-rate}
    \begin{split}
    \bARB
    & =
      \lambda \E_\pi\left[  A(P,z)  \right]
      =  \lambda \Ptr
      \frac{\sqrt{2 \lambda}}{\sigma}
      \int_{0}^{\infty}
      \frac{
      A_+(P,x+\gamma)
      +
      A_-(P,-x-\gamma)
      }{2}
      e^{-\sqrt{2\lambda} x / \sigma} \,
      dx.
    \end{split}
  \end{equation}
  Similarly, defining the intensity of the fee process by
  \[
    \bFEE  \defeq \lim_{T\tends 0} \frac{\E\left[ \FEE_T \right]}{T},
  \]
  we have that
  \begin{equation}\label{eq:fee-rate}
    \begin{split}
    \bFEE
    & =
      \lambda \E_\pi\left[  F(P,z)  \right]
     =  \lambda \Ptr
      \frac{\sqrt{2 \lambda}}{\sigma}
      \int_{0}^{\infty}
      \frac{
      F_+(P,x+\gamma)
      +
      F_-(P,-x-\gamma)
      }{2}
      e^{-\sqrt{2\lambda} x / \sigma} \,
      dx.
    \end{split}
  \end{equation}
\end{theorem}
\begin{proof}
  This result follows from standard properties of Poisson processes.
  The smoothing
  formula \citep[e.g., Theorem~13.5.7,][]{bremaud2020markov} yields that, for $T > 0$,
  \[
    \E\left[ \ARB_T \right]
    =
    \E\left[
      \sum_{i=1}^{N_T}
      A(P_{\tau_i},z_{\tau_i^{-}})
    \right]
    =
    \E\left[
      \int_0^T
      A(P_{t},z_{t^{-}})
      \, dN_t
    \right]
    =
    \E\left[
      \int_0^T
      A(P_{t},z_{t^{-}})
      \times
      \lambda
      \, dt
    \right].
  \]
  Applying Tonelli's theorem and the fundamental theorem of calculus,
  \[
    \lim_{T\tends 0} \frac{\E\left[ \ARB_T \right]}{T}
    = \lim_{T\tends 0} \frac{\lambda}{T}
    \int_0^T
    \E\left[
      A(P_{t},z_{t^{-}})
    \right]
    \, dt
    = \lambda
    \E\left[
      A(P_{0},z_{0^{-}})
    \right],
  \]
  and the result then follows from \Cref{lem:stationary}. The same argument applies to the
  intensity of the fee process.
\end{proof}

\subsection{Example: Constant Product Market Maker}

\Cref{th:arb-rate} provides an exact, semi-analytic closed form expression for the rate of
arbitrage profit, in terms of a certain Laplace transfrom of the functions
$\{A_\pm(P,\cdot)\}$. This expression can be evaluated as an explicit closed form for many
CFMMs. For example, consider the case of constant product market makers:

\begin{corollary}\label{cor:cpmm}
  Consider a constant product market maker, with invariant $f(x,y) \defeq \sqrt{xy} = L$.  Under
  the assumptions of \Cref{th:arb-rate}, the intensity per dollar value in the pool of arbitrage
  profits and fees are
  given by\footnote{Note that there are infinite expected arbitrage profits if
    $\lambda < \sigma^2/8$. This is a consequence of the interaction of the lognormal returns and
    the exponential interblock time. When blocks arrive very slowly, the interblock price changes can have very large tails. This regime is not practically relevant, however. In particular, if
    $\sigma=5\% \text{ (daily)}$, then this occurs when the mean interblock time satisfies
    $\Delta t \defeq \lambda^{-1} > 8/\sigma^2 = 3200 \text{ (days)}$.  }
  \begin{align*}
    \frac{\bARB}{V(P)} & =
    \begin{cases}
      \displaystyle
      \frac{ \sigma^2 }{8}
      \times
      \Ptr
      \times
      \frac{ e^{+\gamma/2} }{
       1 - \sigma^2 / (8 \lambda)
      }
      & \text{if $\sigma^2/8 < \lambda$,} \\
      +\infty & \text{otherwise,}
    \end{cases} \\
    \frac{\bFEE}{V(P)} & =
    \frac{\sigma^2}{8}
    \times
    \frac{
      e^{+\gamma/2} - e^{-\gamma/2}
    }{\gamma}
    \times
    \frac{1}
    {\parens*{1 + \sigma/\parens*{\sqrt{2\lambda} \gamma}}
             \parens*{1 + \sigma/\parens*{ 2\sqrt{2\lambda}}}},
  \end{align*}
  where the quantities on the right side do not depend on the value of $P_0=P$.
\end{corollary}
The proof of \Cref{cor:cpmm} is deferred until \Cref{app:cpmm}. Under the
normalization of \Cref{cor:cpmm}, where the intensity of arbitrage profits is normalized
relative the pool value, the resulting quantity does not depend on the price. The same property
will hold for the more general class of geometric mean market  makers; this is analogous to the
property that \LVR is proportional to pool value for this class \citep{lvr2022}.

As a comparison point, for a constant product market maker, \citet{lvr2022} establish that
\[
  \bLVR \defeq \lim_{T\tends 0} \frac{\E\left[ \LVR_T \right]}{T}
  =
  \frac{\sigma^2}{8} \times V(P),
\]
so that, when $\sigma^2/8 < \lambda$,
\[
  \bARB =
  \bLVR \times
  \Ptr
  \times
  \underbrace{ e^{+\gamma/2} }_{\approx 1 + O(\gamma)}
  \times
  \underbrace{
    \frac{1}{
      1 - \sigma^2 / (8 \lambda)
    }
  }_{\approx 1 + O(\lambda^{-1})}.
\]
Therefore, when fees are small ($\gamma \tends 0$) and the block rate is high
($\lambda \tends \infty$), we have the approximation
\begin{equation}\label{eq:lvr-approx}
  \bARB \approx \bLVR \times \Ptr.
\end{equation}
In \Cref{fig:lvr-approx}, we see that for typical parameter values this approximation is quite accurate, with a relative error of less that $10^{-2}$.

\begin{figure}
  \begin{subfigure}[t]{.475\textwidth}
    \centering
    \begin{tikzpicture}
      \begin{axis}[
        clip=true,
        no markers,
        xlabel={$\gamma$ (bp)},
        ylabel={$\bARB/V(P)$ (bp, daily)},
        width=3.0in,
        height=3.0in,
        grid=major,
        xmin=0, xmax=100,
        ymin=0, ymax=3.5,
        ];
        \addplot[mark=none,line width=1pt,smooth] table [x=gamma,y=arb] {\efftable};
        \addplot[dashed,line width=0.5pt,smooth] table [x=gamma,y=lvr] {\efftable}
        node[below left,pos=0.4,rectangle,fill=white] {\small $\bLVR/V(P) = \sigma^2/8$};
      \end{axis}
    \end{tikzpicture}
    \caption{The normalized intensity of arbitrage profit $\bARB/V(P)$ as a function of the fee
      $\gamma$.}
  \end{subfigure}
  \hfill
  \begin{subfigure}[t]{.475\textwidth}
    \centering
    \begin{tikzpicture}
      \begin{semilogyaxis}[
        clip=true,
        no markers,
        xlabel={$\gamma$ (bp)},
        ylabel={$(\bARB-\bLVR\times \Ptr)/\bARB$},
        width=3.0in,
        height=3.0in,
        grid=major,
        xmin=0, xmax=100,
        ];
        \addplot[mark=none,line width=1pt,smooth] table [x=gamma,y=pcterror] {\efftable};
      \end{semilogyaxis}
    \end{tikzpicture}
    \caption{The relative error of the approximation \eqref{eq:lvr-approx}, i.e.,
      $(\bARB-\bLVR\times \Ptr)/\bARB$, as a function of the fee $\gamma$.\label{fig:lvr-approx}}
  \end{subfigure}

  \caption{The constant product market maker case, with
    $\sigma = 5\% \text{ (daily)}$ and mean interblock time
    $\Delta t \defeq \lambda^{-1}= 12 \text{ (seconds)}$.\label{fig:cpmm}
  }
\end{figure}

%% file: asymptotic-analysis.tex
\section{Asymptotic Analysis}
\label{sec:asymptotic}

In this section, we consider a \emph{fast block} regime, where $\lambda \tends \infty$. In this
setting, blocks are generated very quickly, or, equivalently, the interblock time
$\Delta t \defeq \lambda^{-1} \tends 0$ is very small. First, we characterize asymptotic arbitrage
profits in this regime:

\begin{theorem}\label{th:asymptotic_arb}
Define
\[
  \begin{split}
    \bar A(P,x)
    & \defeq
      \frac{
      A_+(P,x+\gamma)
      +
      A_-(P,-x-\gamma)
      }{2} \geq 0.
  \end{split}
\]
Assume that, for each $P > 0$, $\bar A(P,\cdot)$ is twice continuously differentiable, and that
there exists $A_0$ and $c$ (possibly depending on $P$) such that
\begin{equation}\label{eq:dct}
  \partial_{xx} \bar A(P,x) \leq A_0 e^{c x},\quad\forall\ x \geq 0.
\end{equation}
Consider the fast block regime where $\lambda \tends \infty$. Then,
\begin{equation}\label{eq:barb}
    \bARB
    =
    \frac{\sigma^2 P}{2}
    \times
      \frac{
        y^{*\prime}\left(P e^{-\gamma} \right)
        +
        e^{+\gamma} \cdot y^{*\prime}\left(P e^{+\gamma} \right)
      }{2}
    \times
    \Ptr
    + o\left(\sqrt{\lambda^{-1}}\right).
  \end{equation}
\end{theorem}
Equation \eqref{eq:barb} highlights the dependence of arbitrage profits on the problem
parameters. In the regime where volatility $\sigma$ is large, the fee $\gamma$ is small, and the
block rate $\lambda$ is high, we have that
$\Ptr \approx \eta^{-1} = \sigma\sqrt{\lambda^{-1}/2}/\gamma$. This implies that arbitrage profits
are proportional to the square root of the mean interblock time ($\sqrt{\lambda^{-1}})$, the cube
of the volatility ($\sigma^3$), and the reciprocal of the fee ($\gamma^{-1}$).
This result suggests that faster blockchains (higher $\lambda$) will result in reduced arbitrage profits. We discuss this result in more detail in \Cref{sec:amm-design}.

The next result will similarly characterize fees in this regime:
\begin{theorem}\label{th:asymptotic_fee}
    Define
    \[
    \begin{split}
        \bar F(P,x)
        & \defeq
        \frac{
            F_+(P,x+\gamma)
            +
            F_-(P,-x-\gamma)
        }{2} \geq 0.
    \end{split}
    \]
    Assume that, for each $P > 0$, $\bar F(P,\cdot)$ is continuously differentiable, and that
    there exists $F_0$ and $c$ (possibly depending on $P$) such that
    \begin{equation}\label{eq:dct2}
        \partial_{x} \bar F(P,x) \leq F_0 e^{c x},\quad\forall\ x \geq 0.
    \end{equation}
    Consider the fast block regime where $\lambda \tends \infty$. Then, the instantaneous rate of fees (defined similarly to \Cref{th:arb-rate}) is
    \begin{equation}\label{eq:bfee}
        \bFEE
        =
        \frac{\sigma^2 P}{2}
        \times
            \frac{
                ( 1 - e^{-\gamma} )
                y^{*\prime}\left(P e^{-\gamma} \right)
                +
                ( e^{+\gamma} - 1 )
                y^{*\prime}\left(P e^{+\gamma} \right)
            }{2 \gamma}
        \times
        \left( 1 - \Ptr \right)
        + o\left(1\right).
      \end{equation}
\end{theorem}

The proofs of \Cref{th:asymptotic_arb,th:asymptotic_fee} are deferred to \Cref{app:asymptotic}. \Cref{eq:dct}
is a mild technical condition bounding the convexity of the arbitrage profit as a function of the mispricing.
\Cref{th:asymptotic_arb} provides theoretical justification for the discussion in \Cref{sec:results} comparing \eqref{eq:blvr-def}--\eqref{eq:barb-def}:
we have that, for arbitrary AMMs satisfying the technical condition of \eqref{eq:dct},  $\bARB \approx \bLVR \times \Ptr$ when the fee $\gamma$ is small in the fast block regime.
Additionally,
the instantaneous rate of fees is shown by \Cref{eq:bfee} to be $\bFEE \approx \bLVR \times \left( 1 - \Ptr \right)$ when the fee $\gamma$ is small in the fast block regime.
The last two results mean that, conditioned on the fee $\gamma$ being small in the fast block regime, $\bARB + \bFEE \approx \bLVR$, which can be interpreted as $\bLVR$ being split among fees and arbitrage profits, according to $\Ptr$.
In particular, as the blocks become more and more frequent (for a fixed fee $\gamma$), $\bLVR$ switches from arbitrage profits to fees, where it is eventually consumed.

Empirically, this decomposition that $\bARB + \bFEE \approx \bLVR$ was implicitly validated in the original  work of \citet{lvr2022}. There, the authors empirically validated that delta-hedged LP 
\pnl in a pool with fees matches the difference between total fees collected and \LVR. This is consistent with our setting: since total fees can be decomposed into fees from noise traders plus fees from arbitrageurs, the difference between total fees collected and \LVR is noise trading fees  minus arbitrage profits.

%% file: gas-fees.tex
\section{Modeling of Fixed Gas Fees}\label{sec:gas}

In this section, we give a formulation of arbitrage profits that takes into account the presence
of gas fees as costs for arbitrageurs, and analyze these profits in an asymptotic way as in
\Cref{sec:asymptotic}.  Gas fees are a cost required to be paid to include the arbitrage
transaction in a block. From a financial perspective, they are fixed cost in that they do not
depend on the size of a trade.  Gas fees occur due to the competition of arbitrage transactions
with other transactions to be included in a finite block with limited blockspace for transactions.
While the proportional swap transaction fees determined by $(\gamma_+,\gamma_-)$ go to the AMM
(a smart contract living at the application layer) where they are then distributed to LPs (cf.\
\Cref{sec:model}), gas fees go to the infrastructure, where they are typically earned by the
producers of a block (also called the ``validators'' or ``proposers'' of a block).  Intuitively,
then, one might hope that gas fees acts as a second but analogous friction to arbitrage, with the
main difference being the recipient of the fees paid by arbitrageurs (LPs vs.\ the
protocol/validators). Our results formalize this intuition.

At each block, every arbitrage transaction must pay the gas fee which is constant for that block,
just for interacting with the pool.  Here, we are interpreting the fixed gas fee of a block as the
market-clearing price for transaction inclusion in the block (i.e., for blockspace).  Even though
for a given block the gas fee is fixed, it can vary from block to block, due to dynamic pricing
mechanisms of the underlying blockchain or competition with other transactions.  Fixing the gas
cost, and given a fixed price of the AMM and liquidity, charging a fixed gas fee (i.e., reducing
any arbitrageur's profit by this fixed amount) is equivalent to an additional threshold
$\delta_+ \ge 0$ (in the same units as log-mispricing, e.g., basis points) on top of the boundary
$\gamma_+$ (and $\delta_- \ge 0$ on top of the boundary $\gamma_-$, respectively) that needs to be
overcome in the mispricing process for a profitable trade to exist for the arbitrageur.  More
details on this are provided in \Cref{app:fixed_gas_assumption}.  For analytical tractability, we
will make the assumption that $\delta_+$ and $\delta_-$ are constant, independent of the current
price of the AMM or liquidity thereof.  In \Cref{app:fixed_gas_assumption}, we also discuss why in
practice this is a good approximation.

Formally, as per our model of \Cref{lem:myopic}, suppose that a block is generated at time $t$, so that the arbitrageur observes the price $P_{t-}$. According to the analysis there, we had that if $z_{t^{-}} > +\gamma_+$, then the pool is underpriced, and the arbitrageur can buy from the pool and earn a profit. Once they committed to buying from the pool, they were doing so until their marginal profit was zero, i.e., until $z_t=+\gamma_+$.
To incorporate gas fees, we define the quantities of the new boundaries $\bar \gamma_+ \defeq \gamma_+ + \delta_+ \ge \gamma_+$ and $\bar \gamma_- \defeq \gamma_- + \delta_- \ge \gamma_-$, such that the effect to the mispricing process will now be:
\begin{assumption}
\label{ass:gas_mispricing_process}
The mispricing process with gas fees follows the rules:
\begin{itemize}
\item If $z_{t^{-}} > \bar\gamma_+ = \gamma_+ + \delta_+$, we will have that $z_{t}=\gamma_+$ (after the arbitrageur's trade).
\item If $z_{t^{-}} < -\bar\gamma_- = -\gamma_- -\delta_-$, we will have that $z_{t}=-\gamma_-$ (after the arbitrageur's trade).
\item Otherwise, $z_{t^{-}} = z_t$.
\end{itemize}
\end{assumption}

\paragraph{Summary of results.}
Our theorems below in the fast block regime indicate the following observations: first, the losses to liquidity providers (measured as the sum of the profits of the arbitrageurs and the gas fees)\footnote{Note that all trades are zero-sum between the trader/arbitrageur, the LPs, and the validators or other parties who obtain the gas fees.} increase with increasing gas fees. This means that the lower the gas fees, the smaller the losses to LPs. Effectively, gas fees act as a friction to arbitrageurs which delays profits and artificially decreases competition in a similar manner to the action of block times, analyzed in \Cref{sec:asymptotic}. Second, arbitrageurs trade less frequently with higher gas fees, but make higher profits per arbitrage trade. Third, we can quantify the instantaneous rates of each of: the arbitrage profits (\bARB), the trading fees paid to the AMM (\bFEE), and the gas fees paid to validators (\bGAS) as a split of $\bLVR$ akin to the split we had without gas fees.
This shows again that \LVR is the fundamental quantity due to the stale information induced by the structure of any AMM.
Finally, asymptotically in the fast block regime, all of the LP losses leak to the validators in
gas. In effect, in this regime, arbitrageurs compete all of their profits away to validators.

\subsection{Stationary distribution of mispricing}

We continue the rest of the section (for ease of exposition) under the symmetric
\Cref{as:symmetric,ass:gas_mispricing_process} along with $\delta_+ = \delta_- \defeq \delta$, so
that $\bar\gamma_+ = \bar\gamma_- \defeq \bar\gamma$.  The following result characterizes the
stationary distribution of the mispricing process in this case. This result is analogous to
\Cref{lem:stationary}, however it is not mathematically equivalent to setting a fee of
$\bar \gamma \defeq \gamma + \delta$ in that setting. In particular, in \Cref{lem:stationary}, the
threshold that determines whether a trade occurs ($\gamma$) is the same as the mispricing the
arbitrageur trades to (cf.\ \Cref{lem:myopic}) if the threshold is exceeded. Under
\Cref{ass:gas_mispricing_process}, however the threshold that determines trading
($\bar \gamma \defeq \gamma + \delta$) is different than the trade-to midpricing ($\gamma$).  We
defer the proof to \Cref{app:stationary_gas}.

\begin{theorem}[Stationary Distribution of Mispricing]\label{lem:gas_stationary}
  Under \Cref{as:symmetric,ass:gas_mispricing_process}, the process $z_t$ is an ergodic process on $\R$, with unique invariant distribution
  $\pi(\cdot)$ given by the density
  \[
    p_\pi(z) =
    \begin{cases}
      \pi_+ \times p^{\exp}_{\eta/\bar\gamma}(z-\bar\gamma) & \text{if $z > +\bar\gamma$}, \\
      \pi_+ \times \frac{\eta}{\bar\gamma} \Big[
      u(z+\bar\gamma) - u(z-\bar\gamma) \\
      \qquad\qquad +\frac{\eta}{\bar\gamma} \left(
      r(z+\bar\gamma) + r(z-\bar\gamma) - r(z+\gamma) - r(z-\gamma) \right)
      \Big] & \text{if $z\in[-\bar\gamma,+\bar\gamma]$}, \\
      \pi_- \times p^{\exp}_{\eta/\bar\gamma}(-\bar\gamma-z) & \text{if $z < -\bar\gamma$},
    \end{cases}
  \]
  for $z \in \R$, where $u(x) \defeq \I{x\ge 0}$ is the standard unit step function,
  and $r(x) \defeq \max(x,0)$ is the standard ramp function. Here, we re-define the composite parameter $\eta = \sqrt{2 \lambda}
  \bar\gamma/\sigma$.
  The probabilities ${\pi_{-},\pi_{+}}$ are given by
  \begin{align}
    \pi_+ \defeq \pi\big((+\bar\gamma,+\infty)\big)
    = \pi_- \defeq \pi\big((-\infty,-\bar\gamma)\big) = \frac{1}{2 + 2\eta + \eta^2 \left( 1- \left(\frac{\gamma}{\gamma + \delta}\right)^2 \right)}.
  \label{eq:gas_probtrade}
  \end{align}
  Finally, $p^{\exp}_{\eta/\bar\gamma}(x) \defeq (\eta/\bar\gamma) e^{-(\eta/\bar\gamma) x}$ is the density of an exponential distribution over
  $x\geq 0$ with parameter $\eta/\bar\gamma = \sqrt{2 \lambda} /\sigma$.
\end{theorem}

\begin{figure}[H]
    \centering
    \begin{tikzpicture}[scale=3.5]
        \draw[<->] (-2,0) -- (2,0) node[below left] {pool mispricing $z$};
        \draw[->] (0,0) -- (0,1.5) node[right] {$p_\pi(z)$};

        \ifthenelse{\boolean{simplegraphics}}{
          \fill[gray!20,opacity=0.5] (-0.4,1) -- (-0.25,1.3)
          -- (0.25,1.3) -- (0.4,1) -- (0.4,0) -- (-0.4,0)
          -- cycle;
        }{
          \draw[draw=none, pattern=north west lines, pattern color=blue]
          (-0.4,1) -- (-0.25,1.3) -- (0.25,1.3) -- (0.4,1) -- (0.4,0) -- (-0.4,0) -- cycle;
        }

        \draw[thick] (-0.4,1) -- (-0.25,1.3) -- (0.25,1.3) -- (0.4,1);
        \draw[dashed] (-0.4,1) -- (0.4,1);

        \draw[thick, domain=-2:-0.4, smooth]
          plot (\x, {exp(2*(\x+0.4))});

        \draw[thick, domain=0.4:2, smooth]
          plot (\x, {exp(-2*(\x-0.4))});

        \draw[dashed] (-0.4,1) -- (-0.4,0) node[below left, shift={(0.25,0)}]
          {\small $-\gamma-\delta$};
        \draw[dashed] (0.4,1) -- (0.4,0) node[below right, shift={(-0.35,0)}]
          {\small $+\gamma+\delta$};
        \draw[dashed] (-0.25,1.3) -- (-0.25,0) node[below left, shift={(0.55,0)}]
          {\small $-\gamma$};
        \draw[dashed] (0.25,1.3) -- (0.25,0) node[below right, shift={(-0.55,0)}]
          {\small $+\gamma$};

        \node [below] at (0,0) {\small $0$};
        \node at (0,0.5) [rectangle,align=center,fill=white]
          {\small no-trade \\ w.p.~$\pi_0$};
        \node at (-0.7,0.2) [rectangle,align=center,yshift=-5pt]
          {\small sell trade \\ w.p.~$\pi_-$};
        \node at (+0.7,0.2) [rectangle,align=center,yshift=-5pt]
          {\small buy trade \\ w.p.~$\pi_+$};
    \end{tikzpicture}
    \caption{The density $p_\pi(z)$ of the stationary distribution $\pi(\cdot)$ of mispricing $z$ with gas fees, illustrating trade and no-trade regions for an arbitrageur. Comparing to \Cref{fig:mispricing_stationary}, notice the regions $[-\gamma-\delta, -\gamma]$ and $[+\gamma, +\gamma+\delta]$ where no arbitrage trade happens and which have a different, trapezoid shape.}
    \label{fig:gas_mispricing_stationary}
\end{figure}

The stationary distribution is illustrated in \Cref{fig:gas_mispricing_stationary}. Comparings the
stationary distributions of \Cref{lem:stationary} and \Cref{lem:gas_stationary}, observe that
they both have exponential tails for large mispricings $z$, as well as a uniform density near $z =
0$. However, the distribution in \Cref{lem:gas_stationary} introduces two additional
intervals $[-\gamma-\delta,-\gamma]$ and $[+\gamma,+\gamma+\delta]$ where the density is linear.

\paragraph{Probability of trade.}
Note that according to \Cref{eq:gas_probtrade}, the probability of a profitable trade, i.e., $\pi_+ + \pi_-$, is decreasing in gas fees.

\subsection{Asymptotic Analysis}

In this section, we will characterize the asympotic rate of arbitrage profits, trading fees, gas
fees, and LP losses, in the fast block regime, similar to the analysis of
\Cref{sec:asymptotic}. Proofs are relegated to
\Cref{app:asymptotic_fixed_gas}.

\subsubsection{Arbitrage Profits and Trading Fees}

We characterize arbitrage profits as follows:

\begin{theorem}\label{th:asymptotic_arb_fixed_gas}
Re-define
\[
  \begin{split}
    \bar A(P,x)
    & \defeq
      \frac{
      A_+(P,x+\bar\gamma)
      +
      A_-(P,-x-\bar\gamma)
      }{2} \geq 0,
  \end{split}
\]
where the rate of arbitrageurs' profits is re-defined to exclude gas fees paid to validators, i.e.,
\begin{align*}
  A_+(P,z) &\defeq
  \left[
    P
    \left\{ x^*\left(P e^{-z} \vphantom{\tilde P} \right)
      - x^*\left(P e^{-\gamma}\right)  \right\}
    + e^{+\gamma}
    \left\{y^*\left(P e^{-z} \vphantom{\tilde P} \right)
      -  y^*\left(P e^{-\gamma} \right) \right\}
    - g_+
  \right] \I{ z > +\bar\gamma } \geq 0, \text{ and} \\
    A_-(P,z) &\defeq
    \left[
      e^{+\gamma} P
      \left\{ x^*\left(P e^{-z} \vphantom{\tilde P} \right)
        - x^*\left(P e^{+\gamma} \vphantom{\tilde P} \right)  \right\}
      + \left\{ y^*\left(P e^{-z} \vphantom{\tilde P} \right)
        -  y^*\left(P e^{+\gamma} \vphantom{\tilde P} \right) \right\}
       - g_-
    \right] \I{ z < -\bar\gamma } \geq 0,
\end{align*}
where $g_+, g_-$ are the expressions preceding each one, evaluated at $z=+\bar\gamma, z=-\bar\gamma$ respectively, i.e.,
\begin{align*}
g_+ &\defeq
    P
    \left\{ x^*\left(P e^{-\gamma-\delta} \vphantom{\tilde P} \right)
      - x^*\left(P e^{-\gamma}\right)  \right\}
    + e^{+\gamma}
    \left\{y^*\left(P e^{-\gamma-\delta} \vphantom{\tilde P} \right)
      -  y^*\left(P e^{-\gamma} \right) \right\}
\geq 0, \text{ and} \\
g_- &\defeq
      e^{+\gamma} P
      \left\{ x^*\left(P e^{+\gamma+\delta} \vphantom{\tilde P} \right)
        - x^*\left(P e^{+\gamma} \vphantom{\tilde P} \right)  \right\}
      + \left\{ y^*\left(P e^{+\gamma+\delta} \vphantom{\tilde P} \right)
        -  y^*\left(P e^{+\gamma} \vphantom{\tilde P} \right) \right\}
\geq 0
\,.
\end{align*}
Here, $(g_+,g_-)$ are the gas costs $(+\delta,-\delta)$ valued in the num\'eraire rather than
as a proportional fee.

Assume that, for each $P > 0$, $\bar A(P,\cdot)$ is continuously differentiable, and that
there exists $A_0$ and $c$ (possibly depending on $P$) such that
\begin{equation}\label{eq:gas_dct}
  \partial_{x} \bar A(P,x) \leq A_0 e^{c x},\quad\forall\ x \geq 0.
\end{equation}
Consider the fast block regime where $\lambda \tends \infty$. Then,
\begin{equation}\label{eq:gas_barb}
    \bARB
    =
    \frac{\sigma^2 P}{2}
    \times
    \frac{
            y^{*\prime}\left(P e^{-\gamma-\delta} \right)
            +
            e^{+\gamma+\delta}
            \cdot
            y^{*\prime}\left(P e^{+\gamma+\delta} \right)
    }{2}
    \times
    (1-e^{-\delta})
    \times
    \frac{\sqrt{2\lambda}}{\sigma}
    \times
    \Ptr
    + o\left(\sqrt{\lambda^{-1}}\right).
  \end{equation}
\end{theorem}

Next, we consider the trading fees generated:

\begin{theorem}\label{th:asymptotic_fee_fixed_gas}
    Re-define
    \[
    \begin{split}
        \bar F(P,x)
        & \defeq
        \frac{
            F_+(P,x+\bar\gamma)
            +
            F_-(P,-x-\bar\gamma)
        }{2} \geq 0,
    \end{split}
    \]
    where
    \begin{align*}
    F_+(P,z) &\defeq
          -\left( e^{+\gamma} - 1 \right)
          \left[ y^*\left(P e^{-z}  \vphantom{\tilde P} \right)
            -  y^*\left(P e^{-\gamma}  \vphantom{\tilde P} \right) \right]
          \I{ z > +\bar\gamma } \geq 0,
    \text{ and} \\
    F_-(P,z) &\defeq
          -\left(e^{+\gamma}  - 1  \vphantom{\tilde P} \right) P
          \left[ x^*\left(P e^{-z} \vphantom{\tilde P} \right)
            -  x^*\left(P e^{+\gamma}  \vphantom{\tilde P} \right) \right]
          \I{ z < -\bar\gamma } \geq 0.
    \,.
    \end{align*}
    Assume that, for each $P > 0$, $\bar F(P,\cdot)$ is continuous, and that
    there exists $F_0$ and $c$ (possibly depending on $P$) such that
    \begin{equation}\label{eq:gas_dct2}
        \bar F(P,x) \leq F_0 e^{c x},\quad\forall\ x \geq 0.
    \end{equation}
    Consider the fast block regime where $\lambda \tends \infty$. Then,
    \begin{equation}\label{eq:gas_bfee}
        \bFEE
        =
        (1-e^\gamma)
        \times
            \frac{
                P\cdot
                \left(
                x^*\left(P e^{\gamma+\delta} \right)
                -
                x^*\left(P e^{\gamma} \right)
                \right)
                +
                y^*\left(P e^{-\gamma-\delta} \right)
                -
                y^*\left(P e^{-\gamma} \right)
            }{2}
        \times
        \lambda
        \Ptr
        + o\left(1\right).
      \end{equation}
\end{theorem}

\paragraph{No-gas fee limits of $\bARB, \bFEE$.}
When there is no gas fee, \Cref{eq:gas_barb,eq:gas_bfee} of \Cref{th:asymptotic_arb_fixed_gas,th:asymptotic_fee_fixed_gas} yield the same results as our previous \Cref{eq:barb,eq:bfee} of \Cref{th:asymptotic_arb,th:asymptotic_fee} respectively. To show that, one needs to be careful in handling the asymptotic limits as $\lambda\to\infty$. For the full details, please reference \Cref{app:consistent_no_gas}, where this consistency is proven.

\paragraph{Asymptotic block time rate with gas fees.}
Comparing with the previous setting where there is no gas fee (of \Cref{th:asymptotic_arb}), the rate of the decrease of arbitrage profits with block times remains $\Theta(\sqrt{\lambda^{-1}})$. It is interesting to observe that the rate of the probability of a profitable arbitrage decreases, and the rate of the arbitrage profits conditioned on trade increases. More specifically, the rate of the probability of a profitable arbitrage becomes $\Theta(\lambda^{-1})$, and the rate of arbitrage profits conditioned on a profitable trade becomes $\Theta(\sqrt{\lambda})$. Due to the differences illustrated in \Cref{fig:gas_mispricing_stationary}, since gas fees are a roablock to arbitrageur profitability, they make arbitrage trades more infrequent, and at the same time, conditioned on trade, more profitable, because (in the same intuitive fashion as block times) it delays trading on the AMM. Due to prices deviating farther (i.e., varying with the delay time), there's higher expected rate of return to be made by the arbitrageurs.

\subsubsection{Gas Fees and LP Losses}

We now compute the gas fees that go to the validators, as follows. This is a straightforward corollary of \Cref{lem:gas_stationary}.

\begin{corollary}
[Gas fees]
\label{th:gas}
Under \Cref{as:symmetric,ass:gas_mispricing_process}, in the setting described by \Cref{subsec:def_rates}, the instantaneous rate of the gas fees that go to validators is given by\footnote{The first line is by definition.}
\begin{align}
\bGAS
\defeq \quad &\lambda \cdot (g_+ \cdot \pi_+ + g_- \cdot \pi_-)
\nonumber \\
=\quad &\frac{\lambda}{1 + (\gamma+\delta) \sqrt{2\lambda}/\sigma + \lambda((\gamma+\delta)^2 - \gamma^2)/\sigma^2} \cdot
\Bigg(
\frac{
      P
      \left\{ x^*\left( P e^{-\gamma-\delta} \vphantom{\tilde P} \right)
        - x^*\left(P e^{-\gamma}  \vphantom{\tilde P} \right)  \right\}
}{2}
\nonumber \\
&+ \frac{e^{+\gamma} \left\{ y^*\left(P e^{-\gamma-\delta}  \vphantom{\tilde P}\right)
        -  y^*\left(P e^{-\gamma} \vphantom{\tilde P}
        \right)\right\}
}{2}
+ \frac{
      P
      e^{+\gamma}
      \left\{ x^*\left(P e^{+\gamma+\delta}  \vphantom{\tilde P} \right)
        - x^*\left(P e^{+\gamma}  \vphantom{\tilde P} \right) \right\}
}{2}
\nonumber \\
&+  \frac{y^*\left(P e^{+\gamma+\delta}  \vphantom{\tilde P}  \right)
        - y^*\left(P e^{+\gamma}  \vphantom{\tilde P} \right)
}{2}
\Bigg)
\,,
\label{eq:gas_valid}
\end{align}
where $g_+, g_-$ are the instantaneous transactional costs at the mispricing boundaries that go to the validators (formally defined in \Cref{th:asymptotic_arb_fixed_gas}).

In the limit of the fast block regime where $\lambda \tends \infty$, the first factor of the product (which is the only $\lambda$-dependent factor) of \Cref{eq:gas_valid} becomes
\[
\frac{\sigma^2}{(\gamma+\delta)^2 - \gamma^2}
\,,
\]
and separately, for small gas fees, the second factor of \Cref{eq:gas_valid} becomes
\[
\frac{P\delta^2}{2} \cdot
\frac{
y^{*\prime}\left(P e^{-\gamma} \right)
+
e^{+\gamma} \cdot y^{*\prime}\left(P e^{+\gamma} \right)
}{2}
\,.
\]
\end{corollary}

\paragraph{Losses in the asymptotically fast regime.}
Comparing \Cref{eq:barb,eq:gas_valid} in the asymptotic fast block regime, if we take the limit $\lambda\to\infty$, the arbitrageurs do not make any profits, but there is leakage to validators in the form of positive gas fees paid. Therefore, in this limit of fast blocks, LPs still lose a constant amount of money, but this is taken by validators rather than arbitrageurs. More specifically, the rate with respect to $\lambda$ is $\Theta(1)$. We highlight that this observation does not require any assumption of small gas fees.

\begin{corollary}
[Gas fees $\delta$-asymptotics]
Consider the limit of the fast block regime where $\lambda \tends \infty$, and small gas fees. Then, \Cref{eq:gas_valid} becomes
\begin{equation}
\bGAS
=
\frac{\sigma^2 P}{2}
\times
\frac{
        y^{*\prime}\left(P e^{-\gamma} \right)
        +
        e^{+\gamma} \cdot y^{*\prime}\left(P e^{+\gamma} \right)
}{2}
\times
\frac{\delta}{2\gamma} + o(\delta)
\,.
\label{eq:bgas}
\end{equation}
\end{corollary}

\paragraph{Parametric dependence of asymptotics.}
From \Cref{eq:bgas}, we see that the gas fees given to validators in the fast block regime with
small fee are proportional to the ratio of the gas margin to the roundtrip swapping fee
$\delta/(2\gamma)$, the quadratic variation of the price process, as well as the marginal
liquidity available on the AMM at the mispricing boundary. Note that this quantity is bounded away
from zero even as the block rate tends to infinity, thus LPs lose money to validators no matter  how fast the block arrival rate is. This loss to validators decreases with lower gas fees, asymptotically vanishing with zero gas fees.

\paragraph{LP losses.} Finally, in a similar manner to \Cref{th:arb-rate}, we define the rate of the profits of LPs as $\bLP$. Here, the losses of LPs are no longer just the trading profits of arbitrageurs (because validators also exist here which are obtaining their own gas fees), and are thus going to be lost to both arbitrage profits and total gas fees, i.e., \Cref{eq:gas_barb,eq:gas_valid}, so that $\bLP = - \bARB -\bGAS$.
Therefore, in the fast block regime, the dominating term will be the $\Theta(1)$ term of the gas fees given to validators, according to \Cref{eq:bgas}. In particular, following the observation of the previous paragraph, \emph{lower gas fees yield less losses for LPs}.

\paragraph{LVR as the total sum due to stale prices.}
In the fast block regime with small fees ($\delta$ as well as $\gamma$), from
\Cref{eq:gas_barb,eq:gas_bfee,eq:bgas} as well as \Cref{th:gas}, we have that
$\bARB + \bFEE + \bGAS = \bLVR$, namely that the entire quantity existing in the system due to the
informational lag imposed by the stale prices of AMMs is \LVR. In particular, when fees (gas and
trading) are small but finite, in the fast block regime, the split is as follows:
\begin{align}
\bARB &\approx
\bLVR\times \frac{\sqrt{\sigma^2 /(2\lambda)}}{\gamma+\frac{\delta}{2}}
\\
\bGAS &\approx
\bLVR\times \frac{\delta}{2\gamma},\text{ and }
\\
\bFEE &\approx
\bLVR\times \left(1-\frac{\delta}{2\gamma}-\frac{\sqrt{\sigma^2 /(2\lambda)}}{\gamma+\frac{\delta}{2}}\right)
\,.
\end{align}

Example calculations of the split into $\bARB$ and $\bGAS$ are shown in \Cref{table:gas_split} for
$\sigma=5\% \text{ (daily)}$ volatility and varying mean interblock times $\Delta t \defeq
\lambda^{-1}$, fee levels $\gamma$, and gas fees $\delta$. For these calculations, we use the formulas from \Cref{eq:gas_valid,eq:gas_barb} that make fewer asymptotic assumptions than the ones above which show the parametric dependencies.
Also, as per the discussion in \Cref{app:fixed_gas_assumption}, the gas mispricing $\delta$ mostly depends on $g/\sqrt{L}$, and is not expected to vary much with either the proportional trading fee of the AMM or the block time.\footnote{This is assuming fixed gas fee to a first order; second order effects empirically show that faster blockchains exhibit lower gas fees \citep[see, e.g.,][]{lioba_optimistic_MEV_L2s}.}

\begin{table}[ht]\centering
    \begin{subtable}{\linewidth}\centering
        \begin{tabular}{c|ccccc}
            \toprule
            $\Delta t$ \textbackslash{} $\gamma$ & 1 bp & 5 bp & 10 bp & 30 bp & 100 bp \\ \midrule
            12 sec & (34.5\%, 37.3\%) & (23.3\%, 25.1\%) & (16.5\%, 17.8\%) & (7.6\%, 8.3\%) & (2.7\%, 2.9\%) \\
            2 sec & (22.0\%, 58.4\%) & (13.6\%, 36.1\%) & (9.2\%, 24.4\%) & (4.0\%, 10.7\%) & (1.4\%, 3.6\%) \\
            50 msec & (4.6\%, 77.5\%) & (2.7\%, 45.3\%) & (1.8\%, 29.8\%) & (0.8\%, 12.6\%) & (0.2\%, 4.2\%) \\
            \bottomrule
        \end{tabular}
        \caption{Table varying mean inter-block times
        $\Delta t \defeq \lambda^{-1}$ and fee levels $\gamma$ (in basis points), given gas fee $\delta=9\text{bp}$.\footnotemark}
        \label{table:gas_split_a}
    \end{subtable}

    \vspace{1em} %

    \begin{subtable}{\linewidth}\centering
        \begin{tabular}{c|ccccc}
            \toprule
            $\Delta t$ \textbackslash{} $\delta$ & 1 bp & 5 bp & 10 bp & 30 bp & 100 bp \\ \midrule
            12 sec & (2.4\%, 0.3\%) & (6.4\%, 3.8\%) & (7.8\%, 9.4\%) & (7.7\%, 27.8\%) & (4.8\%, 58.4\%) \\
            2 sec & (2.0\%, 0.6\%) & (3.8\%, 5.6\%) & (4.0\%, 11.9\%) & (3.5\%, 30.9\%) & (2.1\%, 60.8\%) \\
            50 msec & (0.7\%, 1.3\%) & (0.8\%, 7.3\%) & (0.7\%, 13.9\%) & (0.6\%, 32.9\%) & (0.3\%, 62.2\%) \\
            \bottomrule
        \end{tabular}
        \caption{Table varying mean inter-block times
        $\Delta t \defeq \lambda^{-1}$ and gas fee levels $\delta$ (in basis points), given trading fee $\gamma=30\text{bp}$.}
        \label{table:gas_split_b}
    \end{subtable}
    
    \caption{The percentage split of $\bLVR=\bARB + \bFEE + \bGAS$ into ($\bARB, \bGAS$) for each entry of the table respectively, given asset price volatility
    $\sigma=5\%\text{ (daily)}$. The two tables vary different parameters.}
    \label{table:gas_split}
\end{table}
\footnotetext{Per \Cref{app:fixed_gas_assumption}, a gas fee $\delta=9\text{ (bp)}$ corresponds to a
  typical value calibrated to a real deployment of a Uniswap trading pool.}

%% file: optimal-fees.tex
\section{Implications}

This section discusses the practical implications of our model for AMM design and blockchain
architecture.

\subsection{Blockchain Architecture Implications}\label{sec:amm-design}

Our results provide clear guidance for reducing arbitrage profits and improving LP
performance. The key insight is that faster block times directly reduce arbitrage profits through
the $\Ptr$ factor, which decreases as $\lambda$ increases (or, equivalently, as the mean
block-time $\Delta t \defeq \lambda^{-1}$ decreases). In particular, arbitrage profits per unit
time scale according to $\Delta t^{1/2}$, while arbitrage profits per block scale according to
$\Delta t^{3/2}$.  This suggests that blockchain designers should prioritize faster block times to
protect LPs from adverse selection. Similarly, blockchain designers can reduce arbitrage profits
by reducing gas fees.

Indeed, in our model, arbitrage profits go to zero in the limit as $\Delta t \tends 0$. However,
this is likely an artifact of the fact that prices are diffusions in our model, and are thus
continuous. In reality, at very short time scales, market prices are better described with
discountinous jump-diffusion processes. Such processes may exhibit greater movements over short
time horizons and result in larger arbitrage profits.

Our results have been empirically validated by \citet{fritsch2024measuring}. They simulated
arbitrage profits on an AMM against real world asset prices. Their work tests two assumptions made in
this paper: the assumption of Poisson block times (they assume deterministic block times), and
the assumption that asset prices follow a diffusion process (they used historical data from the
Binance exchange). In their simulation, they considered counterfactuals involving varying the block
time. With respect to the ``square-root'' decay of arbitrage profits predicted by the present
paper, they conclude
\begin{quote}
  \emph{
    While our empirical
    findings come close to [the square model of Milionis et al.]  for most pairs and block times
    larger than 1s, we observe a different regime for block times shorter
    than 1s. More precisely, arbitrage profits appear to decline more
    slowly than the theoretical model would suggest.
  }
\end{quote}
These empirical results suggest our model is useful in a broad range of settings. However, on
sub-second time scales, it is likely that the absence of jumps in our model limits its reach. An
interesting direction for future exploration would be to model arbitrage profits when the asset
price follows a jump process.

Our model of arbitrage profits also motivates the
discussion of ``multi-block'' MEV. Here, consider a situation where a single agent controls many
blocks in a row. By censoring the arbitrage trades of other agents, this agent can effectively
increase the block time. Since the arbitrage profits increase as the block time increases, the
agent now has incentive to seek control on contiguous blocks in order to censor the competing trades of
others. Consensus mechanisms should factor these incentives in their design.

\subsection{Pricing Accuracy}\label{sec:pricing}

\citet{trianglefees2023} suggest a trade-off for AMM designers between pricing accuracy, measured
by the standard deviation of mispricing $\sigma_z$, and arbitrage profits. Setting fees
that are low ensures accurate prices, but results in high arbitrage profits, while setting fees
that are high has the opposite effect.

In our setting, we can crisply and analytically quantify this trade-off. Namely, the standard
deviation of mispricing can be computed by \Cref{lem:stdev}, while the arbitrage profits can be
computed by \Cref{th:arb-rate} (exactly) or \Cref{th:asymptotic_arb} (asymptotically).

\Cref{fig:eff} illustrates this trade-off for a constant product market maker, where the arbitrage
profits are computed exactly using \Cref{cor:cpmm}. This figure illustrates two bounds in the low
fee regime ($\gamma \tends 0$). First, as $\gamma\tends 0$,
$\bARB/V(P) \uparrow \bLVR/V(P) = \sigma^2/8$. In this sense, $\bLVR$ captures the worse case loss
to arbitrageurs. Second, as $\gamma \tends 0$, $\sigma_z \downarrow \sigma
\sqrt{\lambda^{-1}}$. The latter quantity is the standard deviation of log-price changes over the
mean interblock time $\Delta t \defeq \lambda^{-1}$. This is the minimal mispricing error forced
by the discrete nature of the blockchain.

\begin{figure}
    \centering
    \begin{tikzpicture}
      \begin{axis}[
        clip=true,
        no markers,
        xlabel={$\sigma_z$ (bp)},
        ylabel={$\bARB/V(P)$ (bp, daily)},
        width=5.0in,
        height=5.0in,
        grid=major,
        xmin=0, xmax=100,
        ymin=0, ymax=3.5,
        ];
        \addplot[mark=none,line width=1pt,smooth] table [x=stdev,y=arb] {\efftable};

        \addplot[red,mark=*,only marks]
        coordinates {
          (5.892564982,3.117519651)
          (5.966002846,2.520287411)
          (7.115568363,1.420809765)
          (9.362873421,0.919577361)
          (20.19635992,0.381669653)
          (60.28865749,0.125626571)
        };
        \addplot[scatter,black,only marks, nodes near coords,
        nodes near coords align={anchor=south west},point meta=explicit symbolic]
        coordinates {
          (5.892564982,3.117519651) [$\gamma=0.01\text{ (bp)}$]
          (5.966002846,2.520287411) [$\gamma=1\text{ (bp)}$]
          (7.115568363,1.420809765) [$\gamma=5\text{ (bp)}$]
          (9.362873421,0.919577361) [$\gamma=10\text{ (bp)}$]
          (20.19635992,0.381669653) [$\gamma=30\text{ (bp)}$]
          (60.28865749,0.125626571) [$\gamma=100\text{ (bp)}$]
        };
        \addplot[mark=none,dashed,line width=0.5pt,domain=0:100] {3.125}
        node [pos=0.5,below,rectangle,fill=white] {\small $\bLVR/V(P)=\sigma^2/8$};
        \addplot[mark=none,dashed,line width=0.5pt]
        coordinates {(5.89255651,0) (5.89255651,3.5)}
        node [pos=0.0,above right=3pt,rectangle,fill=white] (sigmamin)
        {\small $\sigma\sqrt{\lambda^{-1}}$};

      \end{axis}
    \end{tikzpicture}
    \caption{Efficient frontier between mispricing error and arbitrage profits for different
      choices of fees, for a constant product market maker. Here, we set
      $\sigma = 5\% \text{ (daily)}$ and $\lambda^{-1}= 12 \text{ (seconds)}$. The horizontal axis
      is the standard deviation of the steady state pool mispricing, $\sigma_z$. The vertical axis
      is the intensity per unit time of arbitrage profits per dollar value of the pool,
      $\bARB/V(P)$. \label{fig:eff}  }
\end{figure}

\subsection{LP P\&L Decomposition}\label{sec:lp_decomp}

In this section, we consider the original liquidity provider profit and loss decomposition framework established by \citet{lvr2022}. This foundational work provides the theoretical basis for understanding how LPs generate returns in automated market makers. We extend this framework to incorporate our structurally micro-founded model for arbitrage profits, and discuss how this can be utilized in broader settings to understanding LP economics.

\paragraph{Noise Traders.} First, we will augment our model to incorporate a population of AMM-specific noise traders. Noise traders trade only
in the AMM, and trade for totally idiosyncratic reasons (e.g., convenience of executing on chain) and not for informational reasons.

While noise traders’ trades have an initial impact on AMM pool prices, these effects are mitigated by arbitrageurs, who immediately move (or, ``backrun'') the AMM so that its marginal price (net of fees) is consistent with the external price. For tractability,  we will make the simplifying assumption that noise traders do not have an impact on the price dynamics of the asset. Thus, from the LP’s perspective, noise traders simply contribute a flow of fees. We denote by $\text{\sf NT\_FEE}_N$ the total fees collected from noise traders over the first $N$ blocks.

\paragraph{Rebalancing Strategy.}
The rebalancing strategy is a self-financing strategy that takes exactly the same position in the risky asset as the AMM pool, but does so at external market prices. We denote by $R_N$ the rebalancing strategy \pnl over the first $N$ blocks. This is given by
\[
R_N = \sum_{i=0}^{N-1}
x^*\left( \tilde{P}_{\tau_i} \right)
\left( P_{\tau_{i+1}} - P_{\tau_i} \right),
\]
where $\tilde{P}_t$ is the spot price of the AMM pool at time $t$, and $\{ \tau_i \}$ is the sequence of block times.

\paragraph{LP P\&L Decomposition.}
Following \citet{lvr2022}, we decompose the LP \pnl over the first $N$ blocks into three components: the rebalancing strategy \pnl (i.e., directional \pnl from the holdings of the AMM pool), the noise trader fees (i.e., pure revenue from noise traders), and the arbitrage profits (i.e., the value extracted by arbitrageurs from the LP position), according to
\[
  \begin{split}
  \text{LP \pnl}_N
  & =
R_N
  + \text{\sf NT\_FEE}_N
  - \ARB_N.
  \end{split}
\]
The rebalancing strategy can be perfectly hedged through delta-hedging techniques. By taking offsetting positions in the underlying asset, LPs can eliminate the directional risk associated with price movements. This leaves only the fee revenue from noise traders minus the arbitrage profits as the net economic benefit to the LP. In expectation, this is given by
\begin{equation}\label{eq:delta-hedged-pnl}
  \begin{split}
  \E\left[\text{delta-hedged LP \pnl}_N \right]
  & =
  \E\left[\text{\sf NT\_FEE}_N\right]
  - \E\left[\ARB_N\right].
  \end{split}
\end{equation}

\paragraph{Applications.}
Our paper provides a structural model for quantifying expected arbitrage profits ($\E[\ARB_N]$), which represents the second term in \eqref{eq:delta-hedged-pnl}. This structural approach can be combined with reduced-form models of noise trader activity to create a comprehensive framework for understanding LP economics. The structural model allows for micro-founded predictions of arbitrage costs under different market conditions and parameter settings. For example:

\begin{itemize}
    \item Consider a setting where a monopolist LP wants to set optimal fees. Such an agent would pick the fees to maximize \eqref{eq:delta-hedged-pnl}. This LP faces a trade-off: higher fees reduce noise trader activity (decreasing $\E\left[\text{\sf NT\_FEE}_N\right]$) but also reduce arbitrage profits (decreasing $\E\left[\ARB_N\right]$). Our model provides the analytical framework for understanding how fees affect the second term of \eqref{eq:delta-hedged-pnl}, allowing the LP to optimize the fee level that maximizes net revenue.

    \item From a modeling perspective, our framework enables analysis of counterfactual equilibria
      in AMM markets. For instance, we can determine the equilibrium level of liquidity that would
      emerge if fees were changed. In a competitive market with free entry and exit of LPs,
      economic theory suggests that the delta-hedged LP \pnl should be zero in equilibrium. Our
      model provides the analytical foundation for the arbitrage cost component of this
      equilibrium condition, allowing researchers to predict how changes in market parameters
      affect equilibrium liquidity levels. This approach is taken by \citet{adams2024amm}, for example.
\end{itemize}

%% file: appendix.tex
\appendix

\section{Proof of \Cref{lem:myopic}}\label{app:myopic}

\begin{proof}[\proofnamest{Proof of \Cref{lem:myopic}}]
  We consider Part~(\ref{en:myopic-i}), the others follow by analogy. Suppose the arbitrageur considers
  buying from the pool, and selling on the external market at price $P_t$. Then, the arbitrageur
  will face the optimization problem
  \[
    \begin{array}{lll}
      \maximize_{\Delta x, \Delta y} & P_t \Delta x - e^{+\gamma_+} \Delta y \\
      \subjectto & f\left(  x^*(\tilde P_{t^{-}}) - \Delta x,  y^*(\tilde P_{t^{-}}) + \Delta y \right) = L, \\
                & \Delta x,\Delta y \geq 0,
    \end{array}
  \]
  where $( x^*(\tilde P_{t^{-}}), y^*(\tilde P_{t^{-}}) )$ are the reserves of the pool
  immediately prior to the arrival of the arbitrageur. Here, the decision variables
  $\Delta x$ describes the quantity of risky asset purchased by the arbitrageur, while $\Delta y$
  is the amount of num\'eraire paid. Instead, we can parameterize the
  decision through the variables
  \[
    x \defeq  x^*(\tilde P_{t^{-}})  - \Delta x,\quad y \defeq   y^*(\tilde P_{t^{-}})  + \Delta y,
  \]
  which describe the post-trade reserves of the pool. Thus, we can equivalently optimize
  \begin{equation}\label{eq:arb-opt}
    \begin{array}{lll}
      \minimize_{x,y} & P_t e^{-\gamma_+} x + y \\
      \subjectto & f\left( x, y \right) = L, \\
                & x \leq  x^*(\tilde P_{t^{-}}),\ y \geq y^*(\tilde P_{t^{-}}).
    \end{array}
  \end{equation}
  Comparing to \eqref{eq:pool-min} and using the fact that $x^*(\cdot)$ is monotonically
  decreasing while $y^*(\cdot)$ is monotonically increasing, it is clear that the solution to
  \eqref{eq:arb-opt} is given by
  \[
    x =
    \begin{cases}
      x^*\left(\vphantom{\tilde P} P_t e^{-\gamma_+} \right)
      & \text{if $ P_t e^{-\gamma_+} > \tilde P_{t^{-}}$}, \\[0.2\baselineskip]
      x^*\left(\tilde P_{t^{-}}\right) & \text{otherwise,}
    \end{cases}
    \quad
    y =
    \begin{cases}
      y^*\left(\vphantom{\tilde P} P_t e^{-\gamma_+} \right)
      & \text{if $ P_t e^{-\gamma_+} > \tilde P_{t^{-}}$}, \\[0.2\baselineskip]
      y^*\left(\tilde P_{t^{-}}\right) & \text{otherwise.}
    \end{cases}
  \]
  Therefore a profitable trade where the arbitrageur purchases from the pool is only possible when
  $ P_t > \tilde P_{t^{-}} e^{+\gamma_+}$, and the profit is as given in Part~(\ref{en:myopic-i}).
\end{proof}

\section{Proof of \Cref{lem:stationary}}\label{app:stationary}
Define the infinitesimal generator $\Ascr$ by
\[
  \Ascr f(z) \defeq \lim_{\Delta t \tends 0} \frac{1}{\Delta t}
  \E\left[ \left. f(z_{\Delta t}) - f(z_0) \right| z_0 = z \right],
\]
for $f \colon \R \rightarrow \R$ that is twice continuously differentiable. Then, it is easy to
verify that
\[
  \Ascr f(z) = \frac{\sigma^2}{2} f''(z)
  + \lambda\left[ f(+\gamma) - f(z) \right] \I{z>+\gamma}
  + \lambda\left[ f(-\gamma) - f(z) \right] \I{z<-\gamma}.
\]

\begin{lemma}
  The process $z_t$ is ergodic with a unique invariant distribution $\pi(\cdot)$ on $\R$, and this
  distribution is symmetric around $z=0$.
\end{lemma}
\begin{proof}
  Consider the Lyapunov function $V(z) \defeq z^2$. Then,
  \[
    \Ascr V(z) = \sigma^2
    - \lambda \left[ z^2 - \gamma^2 \right]  \I{z \notin (-\gamma,+\gamma)}
    \leq \sigma^2 + \lambda \gamma^2 - \lambda V(z),
  \]
  i.e., this function satisfies the Foster-Lyapunov negative
  drift condition of Theorem~6.1 of
  \citet{meyn1993stability}. Hence, the process is ergodic and a unique
  stationary distribution exists.
  This stationary distribution $\pi(\cdot)$ must also be symmetric around $z=0$. If not, define
  $
    \tilde \pi(C) \defeq \pi\left( \left\{ -z \ : \ z \in C \right\}\right),
  $
  for any measurable set $C \subset \R$.  Since the dynamics \eqref{eq:z-dyn3} are symmetric around $z=0$ by
  \Cref{as:symmetric}, $\tilde \pi(\cdot)$ must also be an invariant distribution,
  contradicting uniqueness.
\end{proof}

\begin{proof}[\proofnamest{Proof of \Cref{lem:stationary}}]
The invariant distribution $\pi(\cdot)$ must satisfy
\begin{equation}\label{eq:stationary}
\E_\pi[\Ascr f(z)] = \int_{-\infty}^{+\infty} \Ascr f(z)\, \pi(dz) = 0,
\end{equation}
for all test functions $f \colon \R \rightarrow \R$.
We will guess that $\pi(\cdot)$ decomposes according to three different densities over the three
regions, and compute the conditional density on each segment via
Laplace transforms using \eqref{eq:stationary}.

Define, for $\alpha\in\R$, the test function
\[
  f_+(z) =
  \begin{cases}
    e^{-\alpha (z-\gamma)} & \text{if $z > +\gamma$}, \\
    1-\alpha(z-\gamma) & \text{otherwise}.
  \end{cases}
\]
Then, from \eqref{eq:stationary},
\[
  \begin{split}
    0 & = \E_\pi[\Ascr f_+(z)] \\
      & = \frac{\sigma^2 \alpha^2}{2} \E_\pi\left[e^{-\alpha (z-\gamma) } \I{z>+\gamma} \right]
        + \lambda \E_\pi\left[\left(1 - e^{-\alpha (z-\gamma) }\right) \I{z>+\gamma} \right]
        + \lambda \alpha \E_\pi\left[\left(z+\gamma\right) \I{z<-\gamma} \right] \\
      & = \frac{\sigma^2 \alpha^2}{2} \E_\pi\left[e^{-\alpha (z-\gamma) } \I{z>+\gamma} \right]
        + \lambda \E_\pi\left[\left(1 - e^{-\alpha (z-\gamma) }\right) \I{z>+\gamma} \right]
        - \lambda \alpha \E_\pi\left[\left(z-\gamma \right) \I{z>+\gamma} \right],
  \end{split}
\]
where for the last step we use symmetry.
Dividing by $\lambda \pi_+$ and conditioning,
\[
  0 = \left( \frac{\alpha^2 \gamma^2}{\eta^2} - 1 \right)
  \E_\pi\left[\left. e^{-\alpha (z-\gamma) }\ \right|\ z>+\gamma \right]
  +
  1
  - \alpha \E_\pi\left[ \left. z-\gamma\ \right|\ z>+\gamma \right].
\]
Then,
\[
  \begin{split}
    \E_\pi\left[\left. e^{-\alpha (z-\gamma) }\ \right|\ z>+\gamma \right]
    =
    \frac{\alpha \E_\pi\left[ \left. z-\gamma\ \right|\ z>+\gamma \right] - 1}
    {\alpha^2\gamma^2/\eta^2 - 1}
  \end{split}
\]
The denominator of this Laplace transform has two real roots, $\pm \eta/\gamma$. We can exclude
the positive root since $\pi(\cdot)$ is a probability distribution. Then, conditioned on
$z >+\gamma$, $z-\gamma$ must be exponential with parameter $\eta/\gamma=\sqrt{2
  \lambda}/\sigma$. This establishes that $\pi(\cdot)$ is exponential conditioned on $z>+\gamma$,
and by symmetry, also conditioned on $z<-\gamma$. Note that
\begin{equation}\label{eq:upexp}
  \E_\pi\left[\left. z-\gamma\ \right|\ z>+\gamma \right] = \gamma/\eta.
\end{equation}

Next, consider the test function
\[
  f_0(z) =
  \begin{cases}
    e^{-\alpha\gamma} - \alpha e^{-\alpha\gamma}(z-\gamma) & \text{if $z>+\gamma$}, \\
    e^{-\alpha z} & \text{if $z \in [-\gamma,+\gamma]$}, \\
    e^{\alpha\gamma} - \alpha e^{\alpha\gamma}(z+\gamma) & \text{if $z<-\gamma$}.
  \end{cases}
\]
Then, from \eqref{eq:stationary},
\[
  \begin{split}
    0 & = \E_\pi[\Ascr f_0(z)] \\
      & = \frac{\sigma^2 \alpha^2}{2} \E_\pi\left[e^{-\alpha z} \I{z\in[-\gamma,+\gamma]} \right]
        + \lambda \alpha e^{-\alpha\gamma} \E_\pi\left[ (z-\gamma) \I{z>+\gamma} \right]
        + \lambda \alpha e^{\alpha\gamma} \E_\pi\left[\left(z+\gamma\right) \I{z<-\gamma} \right] \\
      & = \frac{\sigma^2 \alpha^2}{2} \E_\pi\left[e^{-\alpha z} \I{z\in[-\gamma,+\gamma]} \right]
        + \lambda \alpha \left( e^{-\alpha\gamma} - e^{+\alpha\gamma} \right) \E_\pi\left[ (z-\gamma) \I{z>+\gamma} \right],
  \end{split}
\]
where for the last step we use symmetry. Dividing by $\lambda \pi_0$, conditioning, and using
\eqref{eq:upexp},
\[
  0 = \frac{\alpha^2 \gamma^2}{\eta^2} \E_\pi\left[\left. e^{-\alpha z}\ \right|\ z\in[-\gamma,+\gamma] \right]
  + \alpha \gamma \frac{ e^{-\alpha\gamma} - e^{+\alpha\gamma} }{\eta}
  \frac{\pi_+}{\pi_0}.
\]
Rearranging,
\[
  \E_\pi\left[\left. e^{-\alpha z}\ \right|\ z\in[-\gamma,+\gamma] \right]
  = \frac{\eta}{\gamma} \frac{e^{+\alpha\gamma} - e^{-\alpha\gamma}}
  {\alpha}
  \frac{\pi_+}{\pi_0}.
\]
Inverting this Laplace transform, conditioned on $z\in [-\gamma,+\gamma]$, $\pi(\cdot)$ is the uniform
distribution.
Moreover, we must have
\[
  1 = \lim_{\alpha\tends 0} \E_\pi\left[\left. e^{-\alpha z} \right| z\in[-\gamma,+\gamma] \right]
  = 2 \eta \pi_+ / \pi_0,
\]
so that $\pi_0/\pi_+=2 \eta$. Combining with the fact that $\pi_0 + 2 \pi_+ = 1$,
the result follows.
\end{proof}

\section{Non-Symmetric Analysis}
\label{app:general_stationary}

In this section, we consider dropping \Cref{as:symmetric}. The central implication of
\Cref{as:symmetric} is that the log-price process $z_t$ is a driftless Brownian motion. In
the absence of \Cref{as:symmetric}, $z_t$ is a Brownian motion with drift, and a separate analysis
is required for the stationary distribution. This is analogous to the two cases for stationary
distribution of reflected Brownian motion \citep[e.g., Prop.~6.6,][]{harrison2013brownian}.
In this section, we will establish the stationary distribution in the non-symmetric case with
drift. Once this result is established, the balance of the results in the paper can be derived as
in the symmetric case.

In what follows, we will assume that the drift of the
mispricing process with dynamics \eqref{eq:z-dyn1}--\eqref{eq:z-dyn3} is non-zero, i.e.,
\[
\Delta \defeq \mu - \tfrac{1}{2} \sigma^2 \neq 0.
\]
Here, the generator takes the form
\[
  \Ascr f(z) = \Delta f'(z)
  + \tfrac{1}{2} \sigma^2 f''(z)
  + \lambda \left[ f(+\gammau) - f(z) \right] \I{z > +\gammau}
  + \lambda \left[ f(-\gammal) - f(z) \right] \I{z < -\gammal},
\]

\begin{theorem}
  The process $z_t$ is an ergodic process on $\R$, with unique invariant distribution
  $\pi(\cdot)$ given by the density
  \[
    p_\pi(z) =
    \begin{cases}
      \pi_+ \times p^{\exp}_{\zeta_+}(z-\gammau) & \text{if $z > +\gammau$}, \\
      \pi_0 \times \frac{\zeta_0 e^{-\zeta_0 x}}{e^{+\zeta_0 \gammal} - e^{-\zeta_0 \gammau}} & \text{if $z\in[-\gammal,+\gammau]$}, \\
      \pi_- \times p^{\exp}_{\zeta_-}(-\gammal-z) & \text{if $z < -\gammal$},
    \end{cases}
  \]
  for $z \in \R$. Here, the parameters are given by
  \[
    \zeta_+ \defeq
    \frac{
      \sqrt{ \Delta^2 + 2 \lambda \sigma^2 }
      - \Delta
    }{\sigma^2} > 0,\quad
    \zeta_0 \defeq \frac{2 \Delta}{\sigma^2},\quad
    \zeta_- \defeq
    \frac{
      \sqrt{ \Delta^2 + 2 \lambda \sigma^2 }
      + \Delta
    }{\sigma^2} > 0.
  \]
  The probabilities ${\pi_{-},\pi_0,\pi_{+}}$ of the three segments are given by
  \begin{align*}
    \pi_0 &\defeq
      \left\{
      1
      +
      \zeta_0
      \left[
      \frac{1}{\zeta_+}
      \cdot
      \frac
      {1}
      {
      1
      -
      e^{
      -\zeta_0 (\gammau + \gammal)
      }}
      +
      \frac{1}{\zeta_-}
      \cdot
      \left(
      \frac
      {1}
      {
      1
      -
      e^{
      -\zeta_0 (\gammau + \gammal)
      }}
      -
      1
      \right)
      \right]
      \right\}^{-1},
    \\
    \pi_+
    &\defeq
      \left\{
      1
      +
      \zeta_+
      \cdot
      \frac{\sigma^2}{2\Delta}
      +
      \zeta_+
      \left( \frac{1}{\zeta_-} - \frac{\sigma^2}{2\Delta} \right)
      e^{
      -\zeta_0 (\gammau + \gammal)
      }
      \right\}^{-1},
    \\
    \pi_-
    &\defeq
      \left\{
      1
      +
      \zeta_-
      \left[
      \frac{1}{\zeta_+}
      +
      \frac{\sigma^2}{2\Delta}
      \left\{
      1
      -
      e^{
      -\zeta_0 (\gammau + \gammal)
      }
      \right\}
      \right]
      e^{
      \zeta_0 (\gammau + \gammal) + \gammal)
      }
      \right\}^{-1}
      .
\end{align*}
  Finally, $p_{\zeta}(x) \defeq \zeta e^{-\zeta x}$ is the density of an exponential distribution over
  $x\geq 0$ with parameter $\zeta$.
\end{theorem}
\begin{proof}
  The proof follows that of \Cref{lem:stationary}.

\noindent
\textbf{\sffamily Upper test function:}
\[
  f_+(z) =
  \begin{cases}
    e^{-\alpha (z-\gammau)} & \text{if $z > \gammau$}, \\
    1-\alpha(z-\gammau) & \text{otherwise}.
  \end{cases}
\]
\[
  \begin{split}
    0 & = \E_\pi[\Ascr f_+(z)] \\
      & = \alpha \left( \tfrac{1}{2} \sigma^2 \alpha
        - \Delta \right)
        \E_\pi\left[e^{-\alpha (z-\gammau) } \I{z>\gammau} \right]
        - \Delta \alpha \left( \pi_0 + \pi_- \right)
    \\
      & \quad
        + \lambda \E_\pi\left[\left(1 - e^{-\alpha (z-\gammau) }\right)
        \I{z>\gammau} \right]
        + \lambda \alpha \E_\pi\left[\left(z - \gammal \right) \I{z<\gammal } \right]
  \end{split}
\]
Dividing by $\pi_+$ and conditioning,
\[
  \begin{split}
    0
    & = \alpha \left( \tfrac{1}{2} \sigma^2 \alpha
        - \Delta \right)
        \E_\pi\left[ \left . e^{-\alpha (z-\gammau) } \right| z>\gammau \right]
        - \Delta \alpha
        \frac{\pi_0 + \pi_- }{\pi_+}
    \\
      & \quad
        + \lambda \E_\pi\left[\left. \left(1 - e^{-\alpha (z-\gammau) }\right)
        \right| z>\gammau \right]
        + \lambda \alpha \E_\pi\left[\left.
        z - \gammal \right| z<\gammal \right] \frac{\pi_-}{\pi_+}
    \\
    & = \left\{ \alpha \left( \tfrac{1}{2} \sigma^2 \alpha
      - \Delta \right)
      - \lambda
       \right\}
        \E_\pi\left[ \left . e^{-\alpha (z-\gammau) } \right| z>\gammau \right]
        - \Delta \alpha
        \frac{\pi_0 + \pi_- }{\pi_+}
    \\
      & \quad
        + \lambda
        + \lambda \alpha \E_\pi\left[\left.
        z - \gammal \right| z<\gammal \right] \frac{\pi_-}{\pi_+}
  \end{split}
\]
Rearranging,
\[
  \E_\pi\left[ \left . e^{-\alpha (z-\gammau) } \right| z>\gammau \right]
  = \frac{\Delta \alpha
    \frac{\pi_0 + \pi_- }{\pi_+}
     - \lambda
     + \lambda \alpha \E_\pi\left[\left.
         \gammal - z \right| z<\gammal \right] \frac{\pi_-}{\pi_+}
  }{
    \tfrac{1}{2} \sigma^2 \alpha^2
    - \Delta \alpha
    - \lambda
  }
\]
The denominator has two real roots, only one of which is negative. Then, the
  conditional distribution of $z - \gamma_+$ must be exponential, with parameter
\[
  \zeta_+
  = \frac{1}{\sigma^2}
  \left(
    \sqrt{ \Delta^2 + 2 \lambda \sigma^2 }
    - \Delta
  \right) > 0.
\]

Additionally, note that
\begin{equation}\label{eq:gen_upexp}
    \E_\pi\left[\left. z - \gammau \right| z>\gammau \right] = \frac{1}{\zeta_+}.
\end{equation}

\noindent
\textbf{\sffamily Lower test function:}
\[
  f_-(z) =
  \begin{cases}
    e^{-\alpha (\gammal - z)} & \text{if $z < \gammal$}, \\
    1+\alpha(z-\gammal) & \text{otherwise}.
  \end{cases}
\]
By analogous arguments to the above, we have that
\[
\E_\pi\left[ \left . e^{-\alpha (\gammal-z) } \right| z<\gammal \right]
= \frac{-\Delta \alpha
    \frac{\pi_0 + \pi_+ }{\pi_-}
    - \lambda
    + \lambda \alpha \E_\pi\left[\left.
    z - \gammau \right| z>\gammau \right] \frac{\pi_+}{\pi_-}
}{
    \tfrac{1}{2} \sigma^2 \alpha^2
    + \Delta \alpha
    - \lambda
}
,
\]
and therefore, the distribution of $\gammal - z$, conditioned on
$z < \gammal$, is exponential with parameter
\[
  \zeta_-
  = \frac{1}{\sigma^2}
  \left(
    \sqrt{ \Delta^2 + 2 \lambda \sigma^2 }
    + \Delta
  \right) > 0.
\]

Similarly, note that
\begin{equation}\label{eq:gen_lowexp}
    \E_\pi\left[\left. \gammal - z \right| z<\gammal \right] = \frac{1}{\zeta_-}.
\end{equation}

\noindent
\textbf{\sffamily Middle test function:}
\[
  f_0(z) =
  \begin{cases}
    e^{-\alpha \gammau} - \alpha e^{-\alpha \gammau}(z-\gammau) & \text{if $z > \gammau$}, \\
    e^{-\alpha z} & \text{if $z \in [\gammal,\gammau]$}, \\
    e^{-\alpha\gammal} - \alpha e^{-\alpha\gammal}(z - \gammal) & \text{if $z<\gammal$}.
  \end{cases}
\]
\[
  \begin{split}
    0 & = \E_\pi[\Ascr f_0(z)] \\
      & =
        \alpha \left( \tfrac{1}{2} \sigma^2 \alpha - \Delta \right)
        \E_\pi\left[e^{-\alpha z} \I{z\in[\gammal,\gammau]} \right] \\
      & \quad
        -\Delta \alpha
        \left(e^{-\alpha \gammau} \pi_+ + e^{-\alpha \gammal} \pi_-\right)
    \\
      &\quad  + \lambda \alpha e^{-\alpha \gammau} \E_\pi\left[ (z-\gammau) \I{z>\gammau} \right]
        + \lambda \alpha e^{-\alpha\gammal} \E_\pi\left[\left(z-\gammal \right) \I{z<\gammal} \right].
  \end{split}
\]
Dividing by $\pi_0$ and conditioning,
\[
  \begin{split}
    0 & =
        \alpha \left( \tfrac{1}{2} \sigma^2 \alpha - \Delta \right)
        \E_\pi\left[\left. e^{-\alpha z} \right| z\in[\gammal,\gammau] \right] \\
      & \quad
        -\Delta \alpha
        \left(e^{-\alpha \gammau} \frac{\pi_+}{\pi_0} + e^{-\alpha \gammal}
        \frac{\pi_-}{\pi_0}\right)
    \\
      &\quad  + \lambda \alpha
        \left(
        e^{-\alpha \gammau} \E_\pi\left[\left. z-\gammau \right|
        z>\gammau \right]
        \frac{\pi_+}{\pi_0}
        - e^{-\alpha\gammal} \E_\pi\left[\left. \gammal-z \right|
        z<\gammal \right]
        \frac{\pi_-}{\pi_0}
        \right).
  \end{split}
\]
Rearranging, and using \eqref{eq:gen_upexp} and \eqref{eq:gen_lowexp},
\[
  \begin{split}
    \E_\pi\left[\left. e^{-\alpha z} \right| z\in[\gammal,\gammau] \right]
    & =
      \frac{
      \Delta
      \left(e^{-\alpha \gammau} \frac{\pi_+}{\pi_0} + e^{-\alpha \gammal}
      \frac{\pi_-}{\pi_0}\right)
      -
      \lambda
        \left(
        e^{-\alpha \gammau} \E_\pi\left[\left. z-\gammau \right|
        z>\gammau \right]
        \frac{\pi_+}{\pi_0}
        - e^{-\alpha\gammal} \E_\pi\left[\left. \gammal-z \right|
        z<\gammal \right]
        \frac{\pi_-}{\pi_0}
        \right)
      }{
         \tfrac{1}{2} \sigma^2 \alpha - \Delta
      }
    \\
    &=
      \frac{
          e^{-\alpha \gammau}
          \left(\Delta -\frac{\lambda}{\zeta_+} \right)
          \frac{\pi_+}{\pi_0}
          + e^{-\alpha \gammal}
          \left(\Delta +\frac{\lambda}{\zeta_-} \right)
          \frac{\pi_-}{\pi_0}
      }{
          \tfrac{1}{2} \sigma^2 \alpha - \Delta
      }
    \\
    &=
    \frac{
        - \zeta_+
        \cdot
        \frac{\pi_+}{\pi_0}
        e^{-\alpha \gammau}
        +
        \zeta_-
        \cdot
        \frac{\pi_-}{\pi_0}
        e^{-\alpha \gammal}
    }{
        \alpha - \zeta_0
    }
  \end{split}
\]
Inverting this Laplace transform, conditioned on $z\in [\gammal,\gammau]$, $\pi(\cdot)$ is the superposition of two appropriately-centered truncated exponential distributions.
Moreover, we must have
\[
  \begin{split}
    1 & = \lim_{\alpha\tends 0} \E_\pi\left[\left. e^{-\alpha z} \right| z\in[\gammal,\gammau]
        \right]
    =
      \frac{
          \zeta_+
          \cdot
          \frac{\pi_+}{\pi_0}
          -
          \zeta_-
          \cdot
          \frac{\pi_-}{\pi_0}
      }{
          \zeta_0
      }
  ,
  \end{split}
\]
and additionally, since the Laplace transform corresponds to the conditional density for $z\in[\gammal,\gammau]$, the density
\begin{align*}
&\zeta_+ \cdot \frac{\pi_+}{\pi_0}
\left[
\exp\left( \zeta_0 (z - \gammal) \right) u(z - \gammal)
-
\exp\left( \zeta_0 (z - \gammau) \right) u(z - \gammau)
\right]
\\ -&
\zeta_0
\exp\left( \zeta_0 (z - \gammal) \right) u(z - \gammal)
\end{align*}
must be zero for $z > \gammau$, yielding the equation (only if $\mu\ne \sigma^2/2$)
\[
\zeta_+ \cdot \frac{\pi_+}{\pi_0}
=
\left( \zeta_0 \right)
\Bigg/
\left(
1
-
\exp
\left(
-\zeta_0 (\gammau - \gammal)
\right)
\right)
.
\]

Finally, solving the linear system of equations, combining with the fact that $\pi_0 + \pi_+ + \pi_- = 1$, yields the result (only if $\mu\ne \sigma^2/2$)
\begin{align*}
\pi_0
&=
1
\Bigg/
\left\{
1
+
\zeta_0
\cdot
\left[
\frac{1}{\zeta_+}
\cdot
\frac
{1}
{
    1
    -
    \exp
    \left(
    -\zeta_0 (\gammau - \gammal)
    \right)
}
+
\frac{1}{\zeta_-}
\cdot
\left(
\frac
{1}
{
    1
    -
    \exp
    \left(
    -\zeta_0 (\gammau - \gammal)
    \right)
}
-
1
\right)
\right]
\right\}
\\
\pi_+
&=
1
\Bigg/
\left\{
1
+
\zeta_+
\cdot
\frac{\sigma^2}{2\Delta}
+
\zeta_+
\left( \frac{1}{\zeta_-} - \frac{\sigma^2}{2\Delta} \right)
\exp
\left(
-\zeta_0 (\gammau - \gammal)
\right)
\right\}
\\
\pi_-
&=
1
\Bigg/
\left\{
1
+
\zeta_-
\left[
\frac{1}{\zeta_+}
+
\frac{\sigma^2}{2\Delta}
\left\{
1
-
\exp
\left(
-\zeta_0 (\gammau - \gammal)
\right)
\right\}
\right]
\exp
\left(
\zeta_0 (\gammau - \gammal) - \gammal)
\right)
\right\}
.
\end{align*}
\end{proof}

\section{Proof of \Cref{cor:cpmm}}\label{app:cpmm}

\begin{proof}[\proofnamest{Proof of \Cref{cor:cpmm}}]
  For this pool, we have that
  \[
    V(P) = 2 L \sqrt{P},\quad x^*(P) = L/\sqrt{P},\quad y^*(P) = L \sqrt{P}.
  \]
  Following from \Cref{th:arb-rate},
  \begin{equation}\label{eq:nbarb-rate}
    \frac{\bARB}{V(P)} =
    \lambda \E_\pi\bracks*{\frac{ A_+(P,z) + A_-(P,z) }{V(P)}}.
  \end{equation}
  Note that, in this case,
  \[
    \begin{split}
      \frac{A_+(P,z)}{V(P)}
      & =
        \frac{1}{2 L \sqrt{P}}
        \left[
        P
        \left\{ x^*\left(P e^{-z} \vphantom{\tilde P} \right)
        - x^*\left(P e^{-\gamma}\right)  \right\}
        + e^{+\gamma}
        \left\{y^*\left(P e^{-z} \vphantom{\tilde P} \right)
        -  y^*\left(P e^{-\gamma} \right) \right\}
        \right] \I{ z > +\gamma }
      \\
      & = \tfrac{1}{2}
        \left[
        \left\{
        e^{+z/2} -  e^{+\gamma/2}
        \right\}
        + e^{+\gamma}
        \left\{e^{-z/2} -  e^{-\gamma/2}  \right\}
        \right] \I{ z > +\gamma }
      \\
      & = \tfrac{1}{2}
        e^{+\gamma/2}
        \left[
        e^{+(z-\gamma)/2}
        -  2
        +
        e^{-(z-\gamma)/2}
        \right] \I{ z > +\gamma }.
    \end{split}
  \]
  Taking a conditional expectation over $z  > +\gamma$,
  \[
    \begin{split}
      \E_\pi\left[ \left. \frac{A_+(P,z)}{V(P)} \right| z > +\gamma\right]
      & =
        \begin{cases}
          \tfrac{1}{2}
          e^{+\gamma/2}
          \left[
          \frac{\eta/\gamma}{\eta/\gamma - 1/2}
          -  2
          +
          \frac{\eta/\gamma}{\eta/\gamma + 1/2}
          \right]
          & \text{if $1/2 < \eta/\gamma$}, \\
          +\infty & \text{otherwise,}
        \end{cases}
      \\
      & =
        \begin{cases}
          \tfrac{1}{2}
          e^{+\gamma/2}
          \left[
          \frac{ \sqrt{2\lambda}/\sigma  }{\sqrt{2\lambda}/\sigma - 1/2}
          -  2
          +
          \frac{\sqrt{2\lambda}/\sigma}{\sqrt{2\lambda}/\sigma + 1/2}
          \right]
          & \text{if $\sigma/\sqrt{2 \lambda} < 2$}, \\
          +\infty & \text{otherwise,}
        \end{cases}
      \\
      & =
        \begin{cases}
          \frac{
          e^{+\gamma/2}
          }
          {
          8 \lambda / \sigma^2 - 1
          }
          & \text{if $\sigma^2/8 < \lambda$,} \\
          +\infty & \text{otherwise.}
        \end{cases}
    \end{split}
  \]
  For the remainder of the proof, assume that $\sigma^2/8 < \lambda$. Taking an unconditional
  expectation and multiplying by $\lambda$,
  \[
    \begin{split}
      \lambda \E_\pi\left[  \frac{A_+(P,z)}{V(P)} \right]
      & =
        \pi_+
        \times
        \lambda
        \E_\pi\left[ \left. \frac{A_+(P,z)}{V(P)} \right| z > +\gamma\right]
      \\
      & =
        \frac{\sigma^2}{8}
        \times
        \Ptr
        \times
        \frac{e^{+\gamma/2}}{2
        \parens[\Big]{ 1 - \sigma^2 / (8 \lambda) }
        }.
    \end{split}
  \]
  Combining with the symmetric case
  for $A_-(P,z)/V(P)$, and applying \eqref{eq:nbarb-rate}, the result follows.

  Now, we consider fees. Following from \Cref{th:arb-rate},
  \begin{equation}\label{eq:nbfee-rate}
    \nbFEE =
    \lambda \E_\pi\bracks*{  \frac{ F_+(P,z) + F_-(P,z) }{V(P)} }.
  \end{equation}
  Then,
  \[
    \begin{split}
      \frac{F_+(P,z)}{V(P)}
      & =
        -
        \frac{1}{2 L \sqrt{P}}
        \parens*{ e^{+\gamma} - 1 }
        \left[
        y^*\left(P e^{-z} \vphantom{\tilde P} \right)
        -  y^*\left(P e^{-\gamma} \right)
        \right] \I{ z > +\gamma }
      \\
      & = -\tfrac{1}{2}
        \parens*{ e^{+\gamma} - 1 }
        \left[
        e^{-z/2} -  e^{-\gamma/2}
        \right] \I{ z > +\gamma }
      \\
      & =
        \frac{
        e^{+\gamma/2} - e^{-\gamma/2}
        }{2}
        \left[
        1
        -
        e^{-(z-\gamma)/2}
        \right] \I{ z > +\gamma }.
    \end{split}
  \]
  Taking a conditional expectation over $z$,
  \[
    \begin{split}
      \E_\pi\left[ \left. \frac{F_+(P,z)}{V(P)} \right| z > +\gamma\right]
      & =
        \frac{
        e^{+\gamma/2} - e^{-\gamma/2}
        }{2}
        \left[
        1
        -
        \frac{\eta/\gamma}{\eta/\gamma + 1/2}
        \right]
      \\
      & =
        \frac{
        e^{+\gamma/2} - e^{-\gamma/2}
        }{4}
        \times
        \frac{1}{\sqrt{2\lambda}/\sigma + 1/2}
    \end{split}
  \]
  Taking an unconditional expectation,
  \[
    \begin{split}
      \E_\pi\left[  \frac{F_+(P,z)}{V(P)} \right]
      & =
        \pi\parens*{ z > +\gamma }
        \times
        \E_\pi\left[ \left. \frac{F_+(P,z)}{V(P)} \right| z > +\gamma\right]
      \\
      & =
        \frac{
        e^{+\gamma/2} - e^{-\gamma/2}
        }{4}
        \times
        \frac{1}
        {\parens*{ \sqrt{2\lambda} \gamma / \sigma  + 1}
        \parens*{\sqrt{2\lambda}/\sigma + 1/2}}
      \\
      & =
        \frac{
        e^{+\gamma/2} - e^{-\gamma/2}
        }{4\gamma}
        \times
        \frac{\sigma^2}{4 \lambda}
        \times
        \frac{1}
        {\parens*{1 + \sigma/\parens*{\sqrt{2\lambda} \gamma}}
        \parens*{1 + \sigma/\parens*{ 2\sqrt{2\lambda}}}}.
    \end{split}
  \]
  Combining with the symmetric case
  for $F_-(P,z)/V(P)$, and applying \eqref{eq:nbfee-rate},
  \[
    \nbFEE =
    \frac{\sigma^2}{8}
    \times
    \frac{
      e^{+\gamma/2} - e^{-\gamma/2}
    }{\gamma}
    \times
    \frac{1}
    {\parens*{1 + \sigma/\parens*{\sqrt{2\lambda} \gamma}}
      \parens*{1 + \sigma/\parens*{ 2\sqrt{2\lambda}}}}.
  \]
\end{proof}

\section{Proof of \Cref{th:asymptotic_arb,th:asymptotic_fee}}%
\label{app:asymptotic}

\begin{proof}[\proofnamest{Proof of \Cref{th:asymptotic_arb}}]
  Fix $P >0$. Note that, from the definitions of $A_+(P,\cdot)$ and $A_-(P,\cdot)$, it is easy to
  see that
  \begin{align}
    \label{eq:A1}
    \bar A(P,0) & = 0,
    &
      \bar A(P,x) & \geq 0,\ \forall\ x \geq 0, \\
    \label{eq:A2}
    \partial_x \bar A(P,0) & = 0,
    &
      \partial_x \bar A(P,x) & \geq 0,\ \forall\ x \geq 0,
  \end{align}
  \begin{equation}\label{eq:A3}
  \partial_{xx} \bar A(P,0) =
    P
    \frac{
      y^{*\prime}\left(P e^{-\gamma} \right)
      +
      e^{+\gamma} \cdot y^{*\prime}\left(P e^{+\gamma} \right)
    }{2}.
  \end{equation}

  Define the Laplace transform
  \begin{equation}\label{eq:laplace}
    F(s) = \int_0^{\infty} \bar A(P, x) e^{-s x}\, dx,
  \end{equation}
  for $s \in \R$.
  Observe that, from \eqref{eq:arb-rate},
  \begin{equation}\label{eq:arb-laplace}
    \bARB
    =  \lambda \Ptr
    \frac{\sqrt{2 \lambda}}{\sigma}
    F\left(
      \frac{\sqrt{2 \lambda}}{\sigma}
    \right).
  \end{equation}
  Applying the derivative formula for Laplace transforms (integration-by-parts) twice to
  \eqref{eq:laplace}, and using \eqref{eq:A1}--\eqref{eq:A2},
  \[
    s F(s) =  \underbrace{\bar A(P,0)}_{=0}
    + \int_0^{\infty} e^{-s x} \partial_x \bar A(P, x) \, dx,
  \]
  \[
    s^2 F(s) =  \underbrace{\partial_x \bar A(P,0)}_{=0}
    + \int_0^{\infty} e^{-s x} \partial_{xx} \bar A(P, x) \, dx.
  \]
  Observe that $s^2 F(s)$ is the Laplace transform of the function
  $\partial_{xx} \bar A(P, \cdot)$. Then, applying the initial value theorem for Laplace
  transforms\footnote{This, in turn, relies on the dominated convergence theorem, with the
    dominating function provided by \eqref{eq:dct}.} and \eqref{eq:A3},
  \[
    \lim_{s\tends\infty} s \times s^2 F(s) = \lim_{x\tends 0} \partial_{xx} \bar A(P, x) =
    P
    \frac{
      y^{*\prime}\left(P e^{-\gamma} \right)
      +
      e^{+\gamma} \cdot y^{*\prime}\left(P e^{+\gamma} \right)
    }{2}.
  \]
  Comparing with \eqref{eq:arb-laplace},
  \[
    P
    \frac{
      y^{*\prime}\left(P e^{-\gamma} \right)
      +
      e^{+\gamma} \cdot y^{*\prime}\left(P e^{+\gamma} \right)
    }{2}
    =
    \lim_{\lambda\tends\infty}
    \left(    \frac{\sqrt{2 \lambda}}{\sigma} \right)^3 F\left(
      \frac{\sqrt{2 \lambda}}{\sigma} \right)
    =
    \lim_{\lambda\tends\infty}
    \frac{\bARB}{\sigma^2/2 \times \Ptr}.
  \]
  The result follows.
\end{proof}

\begin{proof}[\proofnamest{Proof of \Cref{th:asymptotic_fee}}]
  We will follow the proof of \Cref{th:asymptotic_arb}. Fix $P >0$. Note that, from the
  definitions of $F_+(P,\cdot)$ and $F_-(P,\cdot)$, it is easy to see that
    \begin{align}
        \label{eq:F1}
        \bar F(P,0) & = 0,
        &
        \bar F(P,x) & \geq 0,\ \forall\ x \geq 0, \\
        \label{eq:F2}
        \partial_x \bar F(P,0) & = P \frac{
            ( 1 - e^{-\gamma} )
            y^{*\prime}\left(P e^{-\gamma} \right)
            +
            ( e^{+\gamma} - 1 )
            y^{*\prime}\left(P e^{+\gamma} \right)
        }{2},
        &
        \partial_x \bar F(P,x) & \geq 0,\ \forall\ x \geq 0.
    \end{align}

    Define the Laplace transform
    \begin{equation}\label{eq:laplace2}
        G(s) = \int_0^{\infty} \bar F(P, x) e^{-s x}\, dx,
    \end{equation}
    for $s \in \R$.
    Observe that, from \eqref{eq:fee-rate},
    \begin{equation}\label{eq:fee-laplace}
        \bFEE
        =  \lambda \Ptr
        \frac{\sqrt{2 \lambda}}{\sigma}
        G\left(
        \frac{\sqrt{2 \lambda}}{\sigma}
        \right).
    \end{equation}
    Applying the derivative formula for Laplace transforms (integration-by-parts) to
    \eqref{eq:laplace2}, and using \eqref{eq:F1},
    \[
    s G(s) =  \underbrace{\bar F(P,0)}_{=0}
    + \int_0^{\infty} e^{-s x} \partial_x \bar F(P, x) \, dx.
    \]
    Observe that $s G(s)$ is the Laplace transform of the function
    $\partial_{x} \bar F(P, \cdot)$. Then, applying the initial value theorem for Laplace
    transforms\footnote{This, in turn, relies on the dominated convergence theorem, with the
        dominating function provided by \eqref{eq:dct2}.} and \eqref{eq:F2}, we get that
    \[
    \lim_{s\tends\infty} s \times s G(s) = \lim_{x\tends 0} \partial_{x} \bar F(P, x) =
    P \frac{
        ( 1 - e^{-\gamma} )
        y^{*\prime}\left(P e^{-\gamma} \right)
        +
        ( e^{+\gamma} - 1 )
        y^{*\prime}\left(P e^{+\gamma} \right)
    }{2}.
    \]
    Comparing with \eqref{eq:fee-laplace},
    \[
      \begin{split}
        P \frac{
        ( 1 - e^{-\gamma} )
        y^{*\prime}\left(P e^{-\gamma} \right)
        +
        ( e^{+\gamma} - 1 )
        y^{*\prime}\left(P e^{+\gamma} \right)
        }{2 \gamma}
        & =
          \frac{1}{\gamma}
          \lim_{\lambda\tends\infty}
          \left(    \frac{\sqrt{2 \lambda}}{\sigma} \right)^2 G\left(
          \frac{\sqrt{2 \lambda}}{\sigma} \right)
        \\
        & =
          \lim_{\lambda\tends\infty}
          \frac{\bFEE}{\sigma^2/2 \times \left( 1 - \Ptr \right)}.
      \end{split}
    \]
    The result follows.
\end{proof}

\section{Proof of \Cref{lem:gas_stationary}}\label{app:stationary_gas}
Define the infinitesimal generator $\Ascr$ by
\[
  \Ascr f(z) \defeq \lim_{\Delta t \tends 0} \frac{1}{\Delta t}
  \E\left[ \left. f(z_{\Delta t}) - f(z_0) \right| z_0 = z \right],
\]
for $f \colon \R \rightarrow \R$ that is twice continuously differentiable. Then, it is easy to
verify that
\[
  \Ascr f(z) = \frac{\sigma^2}{2} f''(z)
  + \lambda\left[ f(+\gamma) - f(z) \right] \I{z>+\bar\gamma}
  + \lambda\left[ f(-\gamma) - f(z) \right] \I{z<-\bar\gamma}.
\]

\begin{lemma}
  The process $z_t$ is ergodic with a unique invariant distribution $\pi(\cdot)$ on $\R$, and this
  distribution is symmetric around $z=0$.
\end{lemma}
\begin{proof}
  Consider the Lyapunov function $V(z) \defeq z^2$. Then,
  \[
    \Ascr V(z) = \sigma^2
    - \lambda \left[ z^2 - \gamma^2 \right]  \I{z \notin (-\bar\gamma,+\bar\gamma)}
    \leq \sigma^2 + \lambda \gamma^2 - \lambda V(z),
  \]
  i.e., this function satisfies the Foster-Lyapunov negative
  drift condition of Theorem~6.1 of
  \citet{meyn1993stability}. Hence, the process is ergodic and a unique
  stationary distribution exists.
  This stationary distribution $\pi(\cdot)$ must also be symmetric around $z=0$. If not, define
  $
    \tilde \pi(C) \defeq \pi\left( \left\{ -z \ : \ z \in C \right\}\right),
  $
  for any measurable set $C \subset \R$.  Since the dynamics \eqref{eq:z-dyn3} are symmetric around $z=0$ by
  \Cref{as:symmetric}, $\tilde \pi(\cdot)$ must also be an invariant distribution,
  contradicting uniqueness.
\end{proof}

\begin{proof}[\proofnamest{Proof of \Cref{lem:gas_stationary}}]
The invariant distribution $\pi(\cdot)$ must satisfy
\begin{equation}\label{eq:gas_stationary}
\E_\pi[\Ascr f(z)] = \int_{-\infty}^{+\infty} \Ascr f(z)\, \pi(dz) = 0,
\end{equation}
for all test functions $f \colon \R \rightarrow \R$.
We will guess that $\pi(\cdot)$ decomposes according to three different densities over the three
regions, and compute the conditional density on each segment via
Laplace transforms using \eqref{eq:gas_stationary}.

Define, for $\alpha\in\R$, the test function
\[
  f_+(z) =
  \begin{cases}
    e^{-\alpha (z-\bar\gamma)} & \text{if $z > +\bar\gamma$}, \\
    1-\alpha(z-\bar\gamma) & \text{otherwise}.
  \end{cases}
\]
Then, from \eqref{eq:gas_stationary},
\[
  \begin{split}
    0 & = \E_\pi[\Ascr f_+(z)] \\
      & = \frac{\sigma^2 \alpha^2}{2} \E_\pi\left[e^{-\alpha (z-\bar\gamma) } \I{z>+\bar\gamma} \right]
        + \lambda \E_\pi\left[\left(1 + \alpha(\bar\gamma-\gamma) - e^{-\alpha (z-\bar\gamma) }\right) \I{z>+\bar\gamma} \right]
        + \lambda \alpha \E_\pi\left[\left(z+\gamma\right) \I{z<-\bar\gamma} \right] \\
      & = \frac{\sigma^2 \alpha^2}{2} \E_\pi\left[e^{-\alpha (z-\bar\gamma) } \I{z>+\bar\gamma} \right]
        + \lambda \E_\pi\left[\left(1 +\alpha(\bar\gamma-\gamma) - e^{-\alpha (z-\bar\gamma) }\right) \I{z>+\bar\gamma} \right]
        - \lambda \alpha \E_\pi\left[\left(z-\gamma \right) \I{z>+\bar\gamma} \right],
  \end{split}
\]
where for the last step we use symmetry.
Dividing by $\lambda \pi_+$ and conditioning,\footnote{We remind that \Cref{lem:gas_stationary} re-defined $\eta$ to be with respect to $\bar\gamma$.}
\[
  0 = \left( \frac{\alpha^2 \bar\gamma^2}{\eta^2} - 1 \right)
  \E_\pi\left[\left. e^{-\alpha (z-\bar\gamma) }\ \right|\ z>+\bar\gamma \right]
  +
  1
  +
  \alpha (\bar\gamma-\gamma)
  - \alpha \E_\pi\left[ \left. z-\gamma\ \right|\ z>+\bar\gamma \right].
\]
Then,
\[
  \begin{split}
    \E_\pi\left[\left. e^{-\alpha (z-\bar\gamma) }\ \right|\ z>+\bar\gamma \right]
    =
    \frac{\alpha \E_\pi\left[ \left. z-\bar\gamma\ \right|\ z>+\bar\gamma \right] - 1}
    {\alpha^2\bar\gamma^2/\eta^2 - 1}
  \end{split}
\]
The denominator of this Laplace transform has two real roots, $\pm \eta/\bar\gamma$. We can exclude
the positive root since $\pi(\cdot)$ is a probability distribution. Then, conditioned on
$z >+\bar\gamma$, $z-\bar\gamma$ must be exponential with parameter $\eta/\bar\gamma=\sqrt{2
  \lambda}/\sigma$. This establishes that $\pi(\cdot)$ is exponential conditioned on $z>+\bar\gamma$,
and by symmetry, also conditioned on $z<-\bar\gamma$. Note that
\begin{equation}\label{eq:gas_upexp}
  \E_\pi\left[\left. z-\bar\gamma\ \right|\ z>+\bar\gamma \right] = \bar\gamma/\eta.
\end{equation}

Next, consider the test function
\[
  f_0(z) =
  \begin{cases}
    e^{-\alpha\bar\gamma} - \alpha e^{-\alpha\bar\gamma}(z-\bar\gamma) & \text{if $z>+\bar\gamma$}, \\
    e^{-\alpha z} & \text{if $z \in [-\bar\gamma,+\bar\gamma]$}, \\
    e^{\alpha\bar\gamma} - \alpha e^{\alpha\bar\gamma}(z+\bar\gamma) & \text{if $z<-\bar\gamma$}.
  \end{cases}
\]
Then, from \eqref{eq:gas_stationary},
\[
  \begin{split}
    0 & = \E_\pi[\Ascr f_0(z)] \\
      & = \frac{\sigma^2 \alpha^2}{2} \E_\pi\left[e^{-\alpha z} \I{z\in[-\bar\gamma,+\bar\gamma]} \right]
        + \lambda\left( e^{-\alpha\gamma} - e^{-\alpha\bar\gamma} \right) \pi_+
        + \lambda \alpha e^{-\alpha\bar\gamma} \E_\pi\left[ (z-\bar\gamma) \I{z>+\bar\gamma} \right] \\
        & + \lambda\left( e^{+\alpha\gamma} - e^{+\alpha\bar\gamma} \right) \pi_-
        + \lambda \alpha e^{\alpha\bar\gamma} \E_\pi\left[\left(z+\bar\gamma\right) \I{z<-\bar\gamma} \right].
  \end{split}
\]
Dividing by $\lambda \pi_0$, conditioning, using symmetry, and using
\eqref{eq:gas_upexp},
\[
  0 = \frac{\alpha^2 \bar\gamma^2}{\eta^2} \E_\pi\left[\left. e^{-\alpha z}\ \right|\ z\in[-\bar\gamma,+\bar\gamma] \right]
  + \left(
      \alpha \bar\gamma \frac{ e^{-\alpha\bar\gamma} - e^{+\alpha\bar\gamma} }{\eta}
      + e^{-\alpha\gamma} + e^{+\alpha\gamma} - e^{+\alpha\bar\gamma} - e^{-\alpha\bar\gamma}
  \right) \cdot \frac{\pi_+}{\pi_0}.
\]
Rearranging,
\[
  \E_\pi\left[\left. e^{-\alpha z}\ \right|\ z\in[-\bar\gamma,+\bar\gamma] \right]
  =
  \frac{\pi_+}{\pi_0}
  \left(
       \frac{\eta}{\bar\gamma}
       \cdot
       \frac{e^{+\alpha\bar\gamma} - e^{-\alpha\bar\gamma}}{\alpha}
       +
       \frac{e^{\alpha\bar\gamma} + e^{-\alpha\bar\gamma}}{\alpha^2}
       \cdot
       \frac{\eta^2}{\bar\gamma^2}
       -
       \frac{e^{\alpha\gamma} + e^{-\alpha\gamma}}{\alpha^2}
       \cdot
       \frac{\eta^2}{\bar\gamma^2}
  \right)
  .
\]
Inverting this Laplace transform, conditioned on $z\in [-\bar\gamma,+\bar\gamma]$, $\pi(\cdot)$ is the trapezoid distribution with the following conditional density:
\[
\frac{\pi_+}{\pi_0} \cdot \frac{\eta}{\bar\gamma} \cdot \left[
u(z+\bar\gamma) - u(z-\bar\gamma)
+ \frac{\eta}{\bar\gamma}
\cdot
\left(
r(z+\bar\gamma) + r(z-\bar\gamma) - r(z+\gamma) - r(z-\gamma)
\right)
\right]
,
\]
where we use the standard notation $u(\cdot), r(\cdot)$ for the unit and ramp functions, respectively.

Moreover, we must have
\[
  1 = \lim_{\alpha\tends 0} \E_\pi\left[\left. e^{-\alpha z} \right| z\in[-\bar\gamma,+\bar\gamma] \right]
  = \frac{\pi_+}{\pi_0} \cdot \frac{\eta}{\bar\gamma} \cdot \left( 2\bar\gamma + \frac{\eta}{\bar\gamma} (\bar\gamma^2 - \gamma^2) \right).
\]
Combining with the fact that $\pi_0 + 2 \pi_+ = 1$, the result follows.
\end{proof}

\section{Proof of \Cref{th:asymptotic_arb_fixed_gas,th:asymptotic_fee_fixed_gas}}
\label{app:asymptotic_fixed_gas}

\begin{proof}[\proofnamest{Proof of \Cref{th:asymptotic_arb_fixed_gas}}]
  Fix $P >0$. We continue from the proof in \Cref{app:asymptotic}. Note that the re-defined $\bARB$ formula has $g_+, g_-$ which are not dependent on $z$; therefore, when differentiating with respect to $z$, these terms will no longer be there. We only have the mismatch between the boundary ($\pm\bar\gamma$) and the expressions inside (with $\pm\gamma$), basically, along with the modified stationary distribution.
  Note that, from the definitions of $A_+(P,\cdot)$ and $A_-(P,\cdot)$, it is easy to
  see that
  \begin{align}
    \label{eq:A1_gas}
    \bar A(P,0) & = 0,
    &
      \bar A(P,x) & \geq 0,\ \forall\ x \geq 0,
  \end{align}
  \begin{equation}\label{eq:A2_gas}
  \partial_x \bar A(P,0) =
    P\times\frac{
                y^{*\prime}\left(P e^{-\gamma-\delta} \right)
                +
                e^{+\gamma+\delta}
                \cdot
                y^{*\prime}\left(P e^{+\gamma+\delta} \right)
        }{2}
        \times
        (1-e^{-\delta})
  .
  \end{equation}

  Define the Laplace transform
  \begin{equation}\label{eq:laplace_gas}
    F(s) = \int_0^{\infty} \bar A(P, x) e^{-s x}\, dx,
  \end{equation}
  for $s \in \R$.
  Observe that, from \eqref{eq:arb-rate},
  \begin{equation}\label{eq:arb-laplace_gas}
    \bARB
    =  \lambda \Ptr
    \frac{\sqrt{2 \lambda}}{\sigma}
    F\left(
      \frac{\sqrt{2 \lambda}}{\sigma}
    \right).
  \end{equation}
  Applying the derivative formula for Laplace transforms (integration-by-parts) to
  \eqref{eq:laplace_gas}, and using \eqref{eq:A1_gas},
  \[
    s F(s) =  \underbrace{\bar A(P,0)}_{=0}
    + \int_0^{\infty} e^{-s x} \partial_x \bar A(P, x) \, dx
    .
  \]
  Observe that $s F(s)$ is the Laplace transform of the function
  $\partial_{x} \bar A(P, \cdot)$. Then, applying the initial value theorem for Laplace
  transforms\footnote{This, in turn, relies on the dominated convergence theorem, with the
    dominating function provided by \eqref{eq:gas_dct}.} and \eqref{eq:A2_gas},
  \[
    \lim_{s\tends\infty} s \times s F(s) = \lim_{x\tends 0} \partial_{x} \bar A(P, x) =
    P\times
        \frac{
                y^{*\prime}\left(P e^{-\gamma-\delta} \right)
                +
                e^{+\gamma+\delta}
                \cdot
                y^{*\prime}\left(P e^{+\gamma+\delta} \right)
        }{2}
        \times
        (1-e^{-\delta})
  \]
  Comparing with \eqref{eq:arb-laplace_gas},
  \begin{align*}
    P\times
        \frac{
                y^{*\prime}\left(P e^{-\gamma-\delta} \right)
                +
                e^{+\gamma+\delta}
                \cdot
                y^{*\prime}\left(P e^{+\gamma+\delta} \right)
        }{2}
        \times
        (1-e^{-\delta})
    &=
    \lim_{\lambda\tends\infty}
    \left(    \frac{\sqrt{2 \lambda}}{\sigma} \right)^2 F\left(
      \frac{\sqrt{2 \lambda}}{\sigma} \right)
    \\ &=
    \lim_{\lambda\tends\infty}
    \frac{\bARB}{\sigma^2/2 \times (\sqrt{2\lambda}/\sigma) \times \Ptr}.
  \end{align*}
  The result follows.
\end{proof}

\begin{proof}[\proofnamest{Proof of \Cref{th:asymptotic_fee_fixed_gas}}]
  We will follow the proof of \Cref{th:asymptotic_arb_fixed_gas}. Fix $P >0$. Note that, from the
  definitions of $F_+(P,\cdot)$ and $F_-(P,\cdot)$, it is easy to see that
    \begin{align}
        \label{eq:F1_gas}
        \bar F(P,0) & = (1-e^\gamma) \times
                    \frac{
                        P\cdot
                        \left(
                        x^*\left(P e^{\gamma+\delta} \right)
                        -
                        x^*\left(P e^{\gamma} \right)
                        \right)
                        +
                        y^*\left(P e^{-\gamma-\delta} \right)
                        -
                        y^*\left(P e^{-\gamma} \right)
                    }{2}
        ,
        \\
        \bar F(P,x) & \geq 0,\ \forall\ x \geq 0.\nonumber
    \end{align}

    Define the Laplace transform
    \begin{equation}\label{eq:laplace2_gas}
        G(s) = \int_0^{\infty} \bar F(P, x) e^{-s x}\, dx,
    \end{equation}
    for $s \in \R$.
    Observe that, from \eqref{eq:fee-rate},
    \begin{equation}\label{eq:fee-laplace_gas}
        \bFEE
        =  \lambda \Ptr
        \frac{\sqrt{2 \lambda}}{\sigma}
        G\left(
        \frac{\sqrt{2 \lambda}}{\sigma}
        \right).
    \end{equation}
    Applying the initial value theorem for Laplace
    transforms\footnote{This, in turn, relies on the dominated convergence theorem, with the
        dominating function provided by \eqref{eq:gas_dct2}.} and \eqref{eq:F1_gas}, we get that
    \[
    \lim_{s\tends\infty} s G(s) = \lim_{x\tends 0} \bar F(P, x) =
    (1-e^\gamma) \times
                    \frac{
                        P\cdot
                        \left(
                        x^*\left(P e^{\gamma+\delta} \right)
                        -
                        x^*\left(P e^{\gamma} \right)
                        \right)
                        +
                        y^*\left(P e^{-\gamma-\delta} \right)
                        -
                        y^*\left(P e^{-\gamma} \right)
                    }{2}
    .
    \]
    Comparing with \eqref{eq:fee-laplace_gas},
    \[
      \begin{split}
      (1-e^\gamma) \times
                    \frac{
                        P\cdot
                        \left(
                        x^*\left(P e^{\gamma+\delta} \right)
                        -
                        x^*\left(P e^{\gamma} \right)
                        \right)
                        +
                        y^*\left(P e^{-\gamma-\delta} \right)
                        -
                        y^*\left(P e^{-\gamma} \right)
                    }{2}
        & =
          \lim_{\lambda\tends\infty}
          \left(    \frac{\sqrt{2 \lambda}}{\sigma} \right) G\left(
          \frac{\sqrt{2 \lambda}}{\sigma} \right)
        \\
        & =
          \lim_{\lambda\tends\infty}
          \frac{\bFEE}{\lambda\Ptr}.
      \end{split}
    \]
    The result follows.
\end{proof}

\section{Consistency of asymptotic results with no gas fee}
\label{app:consistent_no_gas}

We start from the asymptotic case in arbitrage profits.
Define $\bARB_0$ to be the formula from \Cref{eq:barb}.
By taking the limiting ratio of $\lim\limits_{\delta\to 0^+} \lim\limits_{\lambda\to\infty} \frac{\bARB}{\bARB_0}$, and showing that this is 1, we observe that \Cref{eq:gas_barb} correctly yields the same asymptotics as \Cref{eq:barb} without gas fees when $\delta\to 0^+$.
Specifically, the following calculation confirms this:\footnote{It is important to carefully consider the order of limiting operations, so that the asymptotic result is correctly computed. Specifically, we show that the limit of the ratio of the asymptotic expressions goes to 1.}
\begin{align*}
\lim_{\delta\to 0^+}
\lim_{\lambda\to\infty}
(1-e^{-\delta})
\cdot
\frac{
    \frac{\sqrt{2\lambda}}{\sigma} \left( 1+\frac{\gamma\sqrt{2\lambda}}{\sigma} \right)
}{
    1+\frac{(\gamma+\delta)\sqrt{2\lambda}}{\sigma}+\frac{\lambda}{\sigma^2}\cdot ((\gamma+\delta)^2 - \gamma^2)
}
=
\lim_{\delta\to 0^+}
(1-e^{-\delta})
\cdot
\frac{2\gamma}{(\gamma+\delta)^2-\gamma^2}
= 1
\,.
\end{align*}

For the case on fees, our technique needs to be a bit different, and we need to be mindful of the limits of $\delta\to 0^+$, since the second factor in \Cref{eq:gas_bfee} goes to zero, while the third to infinity. First, observe that
\[
\lim_{\delta\to 0^+}
\lim_{\lambda\to\infty}
\delta \cdot \lambda\Ptr
=
\frac{\sigma^2}{2\gamma}
.
\]
Then, note that
\[
\lim_{\delta\to 0^+}
\frac{1}{\delta} \cdot
\frac{
    P\cdot
    \left(
    x^*\left(P e^{\gamma+\delta} \right)
    -
    x^*\left(P e^{\gamma} \right)
    \right)
    +
    y^*\left(P e^{-\gamma-\delta} \right)
    -
    y^*\left(P e^{-\gamma} \right)
}{2}
=
-P e^{-\gamma} \cdot y^{*\prime}\left(P e^{-\gamma}\right)
,
\]
by first-order expansion.
Finally, combining these two limiting equations along with the remaining terms of \Cref{eq:gas_bfee} yields the matching terms of the asymptotics of \Cref{eq:bfee}.

\section{Discussion of fixed gas costs}
\label{app:fixed_gas_assumption}

An alternative model would be to assume that the gas fee is a fixed cost $g \geq 0$, paid in the num\'eraire, every time that a trade occurs. This alternative assumption (as opposed to keeping the boundary shifts $\delta_+, \delta_-$ constant) closely corresponds to our setting; specifically, it is an accurate depiction with smaller block times (we assume that we are in the fast block regime for most of the core analysis of arbitrageur profit rates) and larger fees $\gamma_+, \gamma_-$ such that $\delta_+/\gamma_+, \delta_-/\gamma_-$ are small.

In the setting of \Cref{lem:myopic}, with a fixed gas cost, an arbitrageur will only buy from the pool if the total profit exceeds the gas cost $g$, i.e., if
\begin{equation*}
  P_t
  \left\{ x^*\left( P_t e^{-z_{t^{-}}} \vphantom{\tilde P} \right)
    - x^*\left(P_{t} e^{-\gamma_+}  \vphantom{\tilde P} \right)  \right\}
  + e^{+\gamma_+} \left\{ y^*\left(P_t e^{-z_{t^{-}}}  \vphantom{\tilde P}\right)
    -  y^*\left(P_{t} e^{-\gamma_+} \vphantom{\tilde P}
    \right)\right\} \geq g.
\end{equation*}
What we would then take as $\bar \gamma_+ \geq \gamma_+$ would be the value of the mispricing $z_{t^{-}}$ for which the arbitrageur would break even, i.e., the unique solution to
\[
  P_t
  \left\{ x^*\left( P_t e^{-\bar \gamma_+} \vphantom{\tilde P} \right)
    - x^*\left(P_{t} e^{-\gamma_+}  \vphantom{\tilde P} \right)  \right\}
  + e^{+\gamma_+} \left\{ y^*\left(P_t e^{-\bar \gamma_+}  \vphantom{\tilde P}\right)
    -  y^*\left(P_{t} e^{-\gamma_+} \vphantom{\tilde P}
    \right)\right\} = g.
\]
Then, the mispricing process would behave as per \Cref{ass:gas_mispricing_process}: if $z_{t^{-}} > \bar\gamma_+$, we will have that $z_{t}=\gamma_+$. In full generality, $\bar \gamma_+$ will depend on $P_t$. Our model makes the assumption that $\bar \gamma_+$ is constant. This is substantiated in many cases, where the asset volatility is not as high as to significantly move the boundary. The symmetric case holds for negative mispricing. We show that \Cref{ass:gas_mispricing_process} is an appropriate assumption as in the example case of a CPMM below.

\begin{example}[CPMM]
The dependence of $\bar\gamma_+$ on $P_t$ is indicated as follows:
\begin{align}
&\frac{P_t}{\sqrt{P_t} e^{-\bar\gamma_+}} - \frac{P_t}{\sqrt{P_t} e^{-\gamma_+}}
+ e^{\gamma_+} \left( \sqrt{P_t e^{-\bar\gamma_+}} - \sqrt{P_t e^{-\gamma_+}} \right)
= g/\sqrt{L}
\Leftrightarrow\nonumber\\
&e^{\bar\gamma_+/2}(1+e^{\gamma_+-\bar\gamma_+}) - 2e^{\gamma_+/2} = \frac{g}{\sqrt{P_t L}}
\,.
\label{ex:gas_cpmm}
\end{align}

Plotting this dependence (via the inverse function), for example with a normalized instantaneous price of 1, a fee of 30 bps, and gas-equivalent quantity $g/\sqrt{P_t L} = 2\cdot 10^{-7}$ (examples taken from calculations according to the Uniswap v2 ETH-USDC pool), we notice from \Cref{fig:gas_vs_price} that the boundary moves only slightly by 0.3 bps with the price within a 6\% variation, and is roughly constant around 39 bps, hence $\delta_+ \approx 9$ bps.

\begin{figure}[H]
\centering
\includegraphics[width=0.5\textwidth]{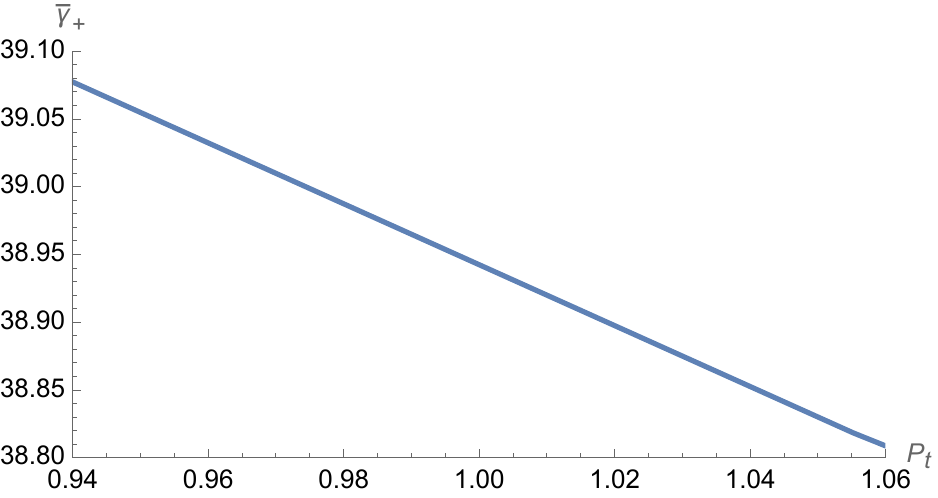}
\caption{Sample plot of variation of $\bar\gamma_+$ (in bps) with price $P_t$ based on the parameter settings above.}
\label{fig:gas_vs_price}
\end{figure}

Getting the Taylor expansion of \Cref{ex:gas_cpmm}, we notice that for the CPMM, as $\bar\gamma_+, \gamma_+\to 0$,\footnote{We note that this approximation results in very good accuracy with the above price variation.}
\[
\bar\gamma_+ \approx \gamma_+ + 2\sqrt{\frac{g}{\sqrt{P_t L}}}
\,.
\]
\end{example}

%% file: lvr-fee-model.bbl
\begin{thebibliography}{23}
\providecommand{\natexlab}[1]{#1}
\providecommand{\url}[1]{\texttt{#1}}
\expandafter\ifx\csname urlstyle\endcsname\relax
  \providecommand{\doi}[1]{doi: #1}\else
  \providecommand{\doi}{doi: \begingroup \urlstyle{rm}\Url}\fi

\bibitem[Adams et~al.(2024)Adams, Moallemi, Reynolds, and
  Robinson]{adams2024amm}
Austin Adams, Ciamac~C Moallemi, Sara Reynolds, and Dan Robinson.
\newblock am-amm: An auction-managed automated market maker.
\newblock \emph{arXiv preprint arXiv:2403.03367}, 2024.

\bibitem[Adams et~al.(2020)Adams, Zinsmeister, and Robinson]{adams2020uniswap}
Hayden Adams, Noah Zinsmeister, and Dan Robinson.
\newblock Uniswap v2 core, 2020.

\bibitem[Adams et~al.(2021)Adams, Zinsmeister, Salem, Keefer, and
  Robinson]{adams2021uniswap}
Hayden Adams, Noah Zinsmeister, Moody Salem, River Keefer, and Dan Robinson.
\newblock Uniswap v3 core, 2021.

\bibitem[Angeris and Chitra(2020)]{angeris2020improved}
Guillermo Angeris and Tarun Chitra.
\newblock Improved price oracles: Constant function market makers.
\newblock In \emph{Proceedings of the 2nd ACM Conference on Advances in
  Financial Technologies}, pages 80--91, 2020.

\bibitem[Angeris et~al.(2021{\natexlab{a}})Angeris, Evans, and
  Chitra]{angeris2021replicatingmarketmakers}
Guillermo Angeris, Alex Evans, and Tarun Chitra.
\newblock Replicating market makers.
\newblock \emph{arXiv preprint arXiv:2103.14769}, 2021{\natexlab{a}}.

\bibitem[Angeris et~al.(2021{\natexlab{b}})Angeris, Evans, and
  Chitra]{angeris2021replicatingmonotonicpayoffs}
Guillermo Angeris, Alex Evans, and Tarun Chitra.
\newblock Replicating monotonic payoffs without oracles.
\newblock \emph{arXiv preprint arXiv:2111.13740}, 2021{\natexlab{b}}.

\bibitem[Br{\'e}maud(2020)]{bremaud2020markov}
Pierre Br{\'e}maud.
\newblock \emph{{Markov} chains: {Gibbs} fields, {Monte Carlo} simulation, and
  queues}, volume~31.
\newblock Springer Science \& Business Media, 2nd edition, 2020.

\bibitem[Clark(2020)]{clark2020replicating}
Joseph Clark.
\newblock The replicating portfolio of a constant product market.
\newblock \emph{Available at SSRN 3550601}, 2020.

\bibitem[Daian et~al.(2020)Daian, Goldfeder, Kell, Li, Zhao, Bentov,
  Breidenbach, and Juels]{juels2019flash}
Philip Daian, Ittay Goldfeder, Tyler Kell, Yunqi Li, Xueyuan Zhao, Iddo Bentov,
  Lorenz Breidenbach, and Ari Juels.
\newblock Flash boys 2.0: Frontrunning, transaction reordering, and consensus
  instability in decentralized exchanges.
\newblock In \emph{Proceedings of the 2020 ACM SIGSAC Conference on Computer
  and Communications Security}, pages 910--928. ACM, 2020.

\bibitem[Dewey and Newbold(2023)]{dewey2023}
Richard Dewey and Craig Newbold.
\newblock The pricing and hedging of constant function market makers.
\newblock Working paper, 2023.

\bibitem[Evans(2020)]{evans2020liquidity}
Alex Evans.
\newblock Liquidity provider returns in geometric mean markets.
\newblock \emph{arXiv preprint arXiv:2006.08806}, 2020.

\bibitem[Evans et~al.(2021)Evans, Angeris, and
  Chitra]{evans2021optimalfeesGeoMeanAMMs}
Alex Evans, Guillermo Angeris, and Tarun Chitra.
\newblock Optimal fees for geometric mean market makers.
\newblock In \emph{International Conference on Financial Cryptography and Data
  Security}, pages 65--79. Springer, 2021.

\bibitem[Fritsch and Canidio(2024)]{fritsch2024measuring}
Robin Fritsch and Andrea Canidio.
\newblock Measuring arbitrage losses and profitability of amm liquidity.
\newblock In \emph{Companion Proceedings of the ACM Web Conference 2024}, pages
  1761--1767, 2024.

\bibitem[Harrison(2013)]{harrison2013brownian}
J~Michael Harrison.
\newblock \emph{Brownian models of performance and control}.
\newblock Cambridge University Press, 2013.

\bibitem[Meyn and Tweedie(1993)]{meyn1993stability}
S.~P. Meyn and R.~L. Tweedie.
\newblock Stability of {Markovian} processes {III}: {Foster--Lyapunov} criteria
  for continuous-time processes.
\newblock \emph{Advances in Applied Probability}, 25\penalty0 (3):\penalty0
  518--548, 1993.

\bibitem[Milionis et~al.(2022)Milionis, Moallemi, Roughgarden, and
  Zhang]{lvr2022}
Jason Milionis, Ciamac~C. Moallemi, Tim Roughgarden, and Anthony~Lee Zhang.
\newblock Quantifying loss in automated market makers.
\newblock In \emph{Proceedings of the 2022 ACM CCS Workshop on Decentralized
  Finance and Security}, DeFi'22, page 71–74, New York, NY, USA, 2022.
  Association for Computing Machinery.
\newblock ISBN 9781450398824.
\newblock \doi{10.1145/3560832.3563441}.
\newblock URL \url{https://doi.org/10.1145/3560832.3563441}.

\bibitem[Milionis et~al.(2023)Milionis, Moallemi, and
  Roughgarden]{jason_exchange_complexity}
Jason Milionis, Ciamac~C. Moallemi, and Tim Roughgarden.
\newblock {Complexity-Approximation Trade-Offs in Exchange Mechanisms: AMMs vs.
  LOBs}.
\newblock In \emph{Financial Cryptography and Data Security}, pages 326--343,
  Cham, 2023. Springer Nature Switzerland.
\newblock ISBN 978-3-031-47754-6.

\bibitem[Milionis et~al.(2024)Milionis, Moallemi, and
  Roughgarden]{jason_revenue_optimal_LP}
Jason Milionis, Ciamac~C. Moallemi, and Tim Roughgarden.
\newblock {A Myersonian Framework for Optimal Liquidity Provision in Automated
  Market Makers}.
\newblock In \emph{15th Innovations in Theoretical Computer Science Conference
  (ITCS 2024)}, Leibniz International Proceedings in Informatics (LIPIcs),
  Dagstuhl, Germany, 2024. Schloss Dagstuhl -- Leibniz-Zentrum f{\"u}r
  Informatik.

\bibitem[Nakamoto(2008)]{nakamoto2008bitcoin}
Satoshi Nakamoto.
\newblock Bitcoin: A peer-to-peer electronic cash system.
\newblock Technical report, 2008.

\bibitem[Nezlobin and Tassy(2025)]{alexDET}
Alex Nezlobin and Martin Tassy.
\newblock Loss-versus-rebalancing under deterministic and generalized
  block-times, 2025.
\newblock URL \url{https://arxiv.org/abs/2505.05113}.

\bibitem[Rao and Shah(2023)]{trianglefees2023}
Rithvik Rao and Nihar Shah.
\newblock Triangle fees, 2023.

\bibitem[Solmaz et~al.(2025)Solmaz, Heimbach, Vonlanthen, and
  Wattenhofer]{lioba_optimistic_MEV_L2s}
Ozan Solmaz, Lioba Heimbach, Yann Vonlanthen, and Roger Wattenhofer.
\newblock Optimistic mev in ethereum layer 2s: Why blockspace is always in
  demand, 2025.
\newblock URL \url{https://arxiv.org/abs/2506.14768}.

\bibitem[Tassy and White(2020)]{tassy2020growth}
Martin Tassy and David White.
\newblock Growth rate of a liquidity provider's wealth in $xy= c$ automated
  market makers, 2020.

\end{thebibliography}
